\documentclass{amsart} 
\usepackage{babel}
\usepackage[utf8]{inputenc}
\usepackage[T1]{fontenc}
\usepackage{lmodern}
\usepackage[a4paper]{geometry}
\usepackage{xcolor}

\usepackage{mathtools,amssymb,amsmath,amsthm}

\usepackage{varioref}
\usepackage[hidelinks,pdfusetitle]{hyperref}
\usepackage{cleveref}

\usepackage{dsfont}
\usepackage{enumerate}
\usepackage{array}
\usepackage{stmaryrd}
\usepackage{wasysym}
\usepackage{pifont}
\usepackage{multicol}
\usepackage{enumitem}

\newcommand{\cuthere}{%
\noindent
\raisebox{-2.8pt}[0pt][0.95\baselineskip]{\ding{34}}
\unskip{\tiny\dotfill}
}
\newcolumntype{M}[1]{>{\centering\arraybackslash}m{#1}}

\newcommand{\wrt}{w.r.t.\ }
\newcommand{\aE}{a.e.\ }
\newcommand{\as}{a.s.\ }
\newcommand{\ie}{i.e.\ }
\newcommand{\eg}{e.g.\ }

\theoremstyle{plain}
\newtheorem{Th}{Theorem}[section]
\newtheorem{theo}[Th]{Theorem}

\newtheorem{lemme}[Th]{Lemma}
\newtheorem{lem}[Th]{Lemma}
\newtheorem{corol}[Th]{Corollary}
\newtheorem{cor}[Th]{Corollary}
\newtheorem{prop}[Th]{Proposition}
\newtheorem{defin}[Th]{Definition}
\newtheorem{defi}[Th]{Definition}

\newtheorem{quest}[Th]{Question}

\theoremstyle{definition}
\newtheorem{remark}[Th]{Remark}
\newtheorem{rem}[Th]{Remark}

\newtheorem{notation}[Th]{Notation}
\newtheorem*{remark*}{Remark}

\newcommand{\Space}{\mathbf{Z}}
\newcommand{\Topo}{\mathcal{O}_\Space}
\newcommand{\Borel}{\mathcal{B}(\Space)}
\newcommand{\dspace}{d_0}

\newcommand{\BorelX}[1]{\mathcal{B}(#1)}

\newcommand{\signed}{\pm}

\newcommand{\ca}{\mathcal{A}}
\newcommand{\cf}{\mathcal{F}}
\newcommand{\ck}{\mathcal{K}}
\newcommand{\cm}{\mathcal{M}}
\newcommand{\co}{\mathcal{O}}
\newcommand{\cp}{\mathcal{P}}

\newcommand{\cA}{\mathcal{A}}
\newcommand{\cB}{\mathcal{B}}

\newcommand{\cD}{\mathcal{D}}

\newcommand{\cG}{\mathcal{G}}
\newcommand{\cH}{\mathcal{H}}
\newcommand{\cI}{\mathcal{I}}

\newcommand{\cK}{\mathcal{K}}

\newcommand{\cM}{\mathcal{M}}

\newcommand{\cO}{\mathcal{O}}
\newcommand{\cP}{\mathcal{P}}
\newcommand{\cQ}{\mathcal{Q}}
\newcommand{\cR}{\mathcal{R}}

\newcommand{\cT}{\mathcal{T}}

\newcommand{\cW}{\mathcal{W}}

\newcommand{\cZ}{\mathcal{Z}}

\newcommand{\bbG}{\mathbb{G}}
\newcommand{\bbH}{\mathbb{H}}

\newcommand{\Proba}{\mathcal{M}_1(\Space)}
\newcommand{\SubProba}{\mathcal{M}_{\leq1}(\Space)}
\newcommand{\Meas}{\mathcal{M}_+(\Space)}
\newcommand{\SignedMeas}{\mathcal{M}_{\pm}(\Space)}
\newcommand{\MeasEps}{\mathcal{M}_{\epsilon}(\Space)}

\newcommand{\ProbaX}[1]{\mathcal{M}_1(#1)}

\newcommand{\MeasX}[1]{\mathcal{M}_+(#1)}
\newcommand{\SignedMeasX}[1]{\mathcal{M}_{\pm}(#1)}

\newcommand{\MeasEpsK}{\mathcal{M}_{\epsilon}(K)}
\newcommand{\SignedMeasK}{\mathcal{M}_{\pm}(K)}

\newcommand{\CbFunct}{C_b(\Space)}

\newcommand{\Graphon}{\mathcal{W}_1}

\newcommand{\SubGraphon}{\mathcal{W}_{\leq 1}}
\newcommand{\Kernel}{\mathcal{W}_{\signed}}
\newcommand{\Kernelp}{\mathcal{W}_{+}}

\newcommand{\KernelM}{\mathcal{W}_\cm}
\newcommand{\KernelEps}{\mathcal{W}_{\epsilon}}
\newcommand{\UGraphon}{\widetilde{\mathcal{W}}_1}
\newcommand{\UGraphonm}{\widetilde{\mathcal{W}}_{1,\m}}
\newcommand{\UGraphond}{\widetilde{\mathcal{W}}_{1,d}}

\newcommand{\UKernel}{\widetilde{\mathcal{W}}_{\signed}}

\newcommand{\UKernelpm}{\widetilde{\mathcal{W}}_{+,\m}}
\newcommand{\UKerneld}{\widetilde{\mathcal{W}}_{\signed,d}}
\newcommand{\UKernelpd}{\widetilde{\mathcal{W}}_{+,d}}
\newcommand{\UKernelp}{\widetilde{\mathcal{W}}_{+}}
\newcommand{\UKernelM}{\widetilde{\mathcal{W}}_\cm}
\newcommand{\UKernelEps}{\widetilde{\mathcal{W}}_{\epsilon}}
\newcommand{\cKd}{\widetilde{\ck}_{d}}
\newcommand{\cKm}{\widetilde{\ck}_{\m}}
\newcommand{\cKdm}{\widetilde{\ck}_{\dcut}}

\newcommand{\TotalMass}[1]{\norm{#1}_\infty}
\newcommand{\TM}[1]{\norm{#1}_\infty}

\newcommand{\InvRelabel}{S_{[0,1]}}
\newcommand{\Relabel}{\bar{S}_{[0,1]}}

\newcommand{\dcutR}{d_{\square,\R}}
\newcommand{\NcutR}[1]{\Vert#1\Vert_{\square,\R}}

\newcommand{\NcutRSymbol}{\Vert\cdot\Vert_{\square,\R}}
\newcommand{\NcutRpos}[1]{\Vert#1\Vert_{\square,\R}^+}
\newcommand{\NcutRposLarge}[1]{\left\Vert#1\right\Vert_{\square,\R}^+}
\newcommand{\NcutRposSymbol}{\Vert\cdot\Vert_{\square,\R}^+}

\newcommand{\m}{\mathrm{m}}

\newcommand{\simd}{\sim_{d}}
\newcommand{\dmeas}{d_\m}

\newcommand{\NmeasLLarge}[1]{N_\m\Bigl(#1\Bigr)}
\newcommand{\NmeasSymbol}{N_\m}
\newcommand{\dcut}{d_{\square,\m}}
\newcommand{\Ncut}[1]{N_{\square,\m}(#1)}

\newcommand{\NcutSymbol}{N_{\square,\m}}
\newcommand{\ddcut}{\delta_{\square,\m}}
\newcommand{\dd}{\delta_{\square}}
\newcommand{\ddPrime}{\delta_{\square}'}

\newcommand{\mPrime}{\mathrm{m'}}
\newcommand{\dmeasPrime}{d_{\mPrime}}

\newcommand{\dcutPrime}{d_{\square,\mPrime}}
\newcommand{\dPrime}{d'}

\newcommand{\ddcutPrime}{\delta_{\square,\mPrime}}

\newcommand{\dmeasF}{d_{\mathcal{F}}}
\newcommand{\NmeasF}[1]{\Vert#1\Vert_{\mathcal{F}}}

\newcommand{\NmeasFSymbol}{\Vert\cdot\Vert_{\mathcal{F}}}
\newcommand{\dcutF}{d_{\square,\mathcal{F}}}
\newcommand{\NcutF}[1]{\Vert#1\Vert_{\square,\mathcal{F}}}

\newcommand{\NcutFSymbol}{\Vert\cdot\Vert_{\square,\mathcal{F}}}
\newcommand{\ddcutF}{\delta_{\square,\mathcal{F}}}
\newcommand{\F}{\mathcal{F}}

\newcommand{\ddcutFtilde}{\delta_{\square,\tilde{\F}}}

\newcommand{\ccc}{\mathrm{c}}
\newcommand{\dmeasC}{d^\ccc}
\newcommand{\dcutC}{d^\ccc_{\square, m}}
\newcommand{\ddcutC}{\delta^\ccc_{\square, m}}

\newcommand{\dmeasP}{d_{\mathcal{P}}}
\newcommand{\dcutP}{d_{\square,\mathcal{P}}}
\newcommand{\ddcutP}{\delta_{\square,\mathcal{P}}}

\newcommand{\mKR}{\text{KR}}
\newcommand{\dmeasKR}{d_{\text{KR}}}
\newcommand{\NmeasKR}[1]{\Vert#1\Vert_{\mKR}}

\newcommand{\NmeasKRSymbol}{\Vert\cdot\Vert_{\mKR}}
\newcommand{\dcutKR}{d_{\square,\mKR}}

\newcommand{\NcutKRSymbol}{\Vert\cdot\Vert_{\square,\mKR}}
\newcommand{\ddcutKR}{\delta_{\square,\mKR}}

\newcommand{\mFM}{\text{FM}}

\newcommand{\NmeasFM}[1]{\Vert#1\Vert_{\mFM}}

\newcommand{\NmeasFMSymbol}{\Vert\cdot\Vert_{\mFM}}
\newcommand{\dcutFM}{d_{\square,\mFM}}

\newcommand{\NcutFMSymbol}{\Vert\cdot\Vert_{\square,\mFM}}
\newcommand{\ddcutFM}{\delta_{\square,\mFM}}

\newcommand{\Prb}{\mathbb{P}}
\newcommand{\Esp}{\mathbb{E}}

\newcommand{\E}{\mathbb{E}}
\newcommand{\ind}{\mathds{1}}
\newcommand{\un}{\mathds{1}}

\newcommand{\rooot}{\partial} 

\newcommand{\N}{\mathbb{N}}
\newcommand{\Z}{\mathbb{Z}}

\newcommand{\R}{\mathbb{R}}

\newcommand{\drv}{\mathrm{d}}
\newcommand{\rd}{\mathrm{d}}

\newcommand{\norm}[1]{\Vert#1\Vert}

\DeclareMathOperator*{\esssup}{\mathrm{essup}}

\DeclareMathOperator{\diam}{diam}

\newcommand{\eps}{\varepsilon}
\DeclareMathOperator{\e}{e}

\newcommand{\leqlex}{\leq_{\mathrm{lex}}}
\newcommand{\linflex}{<_{\mathrm{lex}}}

\title{Probability-graphons: Limits of large dense weighted graphs}

\date{\today}

\author{Romain Abraham}
\address{Romain Abraham,
Institut Denis Poisson,
Universit\'{e} d'Orl\'{e}ans,
Universit\'e de Tours,
CNRS,
France}
\email{romain.abraham@univ-orleans.fr}

\author{Jean-Fran\c{c}ois Delmas}
\address{Jean-Fran\c{c}ois Delmas,
 CERMICS, Ecole des Ponts, France}
\email{jean-francois.delmas@enpc.fr}

\author{Julien Weibel}
\address{Julien Weibel,
Institut Denis Poisson,
Universit\'{e} d'Orl\'{e}ans,
Universit\'e de Tours,
CNRS,
France
and
CERMICS, Ecole des Ponts, France}
\email{julien.weibel@normalesup.org}

\begin{document}



\begin{abstract}
  We introduce  probability-graphons which are probability  kernels that
  generalize    graphons   to    the    case    of   weighted    graphs.
  Probability-graphons appear as the limit objects to study sequences of
  large  weighted   graphs  whose  distribution  of   subgraph  sampling
  converge.  The  edge-weights are  taken from  a general  Polish space,
  which also covers  the case of decorated graphs.  Here,  graphs can be
  either  directed or  undirected.   Starting from  a distance  $\dmeas$
  inducing the  weak topology on measures,  we define a cut  distance on
  probability-graphons,  making  it  a   Polish  space,  and  study  the
  properties  of  this  cut  distance.   In  particular,  we  exhibit  a
  tightness  criterion  for  probability-graphons  related  to  relative
  compactness  in the  cut  distance.   We also  prove  that under  some
  conditions  on the  distance $\dmeas$,  which are  satisfied for  some
  well-know distances like the Prohorov distance, and the Fortet-Mourier
  and  Kantorovitch-Rubinstein norms,  the topology  induced by  the cut
  distance on the space of  probability-graphons is independent from the
  choice of $\dmeas$.  Eventually, we prove that this topology coincides
  with the  topology induced by  the convergence in distribution  of the
  sampled subgraphs.
\end{abstract}

\maketitle

\section{Introduction}

\subsection{Motivation and literature review}

Networks appear  naturally in a  wide variety of context,  including for
example: biological  networks \cite{bornholdtHandbookGraphsNetworks2002,
lewisNetworkScienceTheory2009},         epidemics        processes
\cite{draiefEpidemicsRumoursComplex2010,
  kissMathematicsEpidemicsNetworks2017},   electrical  power   grids 
\cite{albertStructuralVulnerabilityNorth2004}   and    social   networks
\cite{albertStatisticalMechanicsComplex2002, newmanStructureFunctionComplex2003}.
Most of those problems involve large dense graphs,
	that is graphs that have a large number of vertices
	and a number of edges that scales as the square of the number of vertices.
Those graphs are too large to be represented entirely in the targeted applications.
The idea is then to go from a combinatorial representation given by the graph
to an infinite continuum representation.

In the case of unweighted graphs (\ie graphs without edge-weights), 
	a theory was developed to study the asymptotic behaviour
	of large dense graphs, 
	with the limit objects being the so-called \emph{graphons}.
The properties of graphons were studied in a series of articles started by
	\cite{lovaszLimitsDenseGraph2006, 
freedmanReflectionPositivityRank2006, 
bollobasCutMetricRandom2010, 
borgsConvergentSequencesDense2008,borgsConvergentSequencesDense2012}.
We shall refer to the monograph \cite{Lovasz} which exposes in details
	the theory of graphons developed in this series of articles.
Graphons can be used to define models of random graphs 
	with latent vertex-type variables
	(called $W$-random graphs)
	generalizing the Erdös-Rényi graph and the stochastic block model (SBM).
The space of graphons can be equipped with the so-called \emph{cut distance},
	making it a compact space, and whose topology is that of 
	the convergence in distribution for all sampled subgraphs,
	or equivalently of the convergence for subgraph homomorphism densities.

In recent years, graphons have been used in several application context: 
	non-parametric estimation methods and algorithms for massive networks \cite{borgsGraphonsNonparametricMethod2017},
	SIS epidemic models \cite{delmasInfinitedimensionalMetapopulationSIS2022},
	the study of transferability properties for Graph Neural Networks \cite{kerivenWhatFunctionsCan2023}. 
Furthermore, there has been recent developments in the study of mean-field systems using graphons:
	stochastic games and their Nash equilibria \cite{lackerLabelstateFormulationStochastic2022},
	opinion dynamic on a graphon \cite{alettiOpinionDynamicsGraphon2022},
	cooperative multi-agent reinforcement learning \cite{huGraphonMeanfieldControl2023}, to cite a few.
\medskip

However,  most  real-world  phenomenon  on the  above  networks  involve
weighted  networks, where  each  edge in  the  graph carries  additional
information such as  intensity or frequency of  interaction, or transfer
capacity.  

There exists  many models of random  weighted graphs.  For
example configuration  models with edges having  independent exponential
weights        have        been       considered        in        \cite{
  bhamidiFirstPassagePercolation2010a,  aminiFloodingWeightedSparse2013,
  aminiDiameterWeightedRandom2015},               see               also
\cite{garlaschelliWeightedRandomGraph2009,    hurdWattsCascadeModel2013}
 where the  distribution of the weight of an edge depends on the types
 of its end-points.
Random geometric graphs
 with vertices and edges having independent Gaussian weights have been
 considered in \cite{alonsoWeightedRandomgeometricRandomrectangular2018}. 
 
Weighted SBMs (sometimes also called labeled SBMs), 
in which each edge independently receives  a random weight whose distribution depends on the 
	community labels of its end-points,
have been studied to solve community detection in
\cite{lelargeReconstructionLabeledStochastic2013}
(see also \cite{xuEdgeLabelInference2014} for more general models where vertex-labels come from a compact space),
and exact community recovery in \cite{jogRecoveringCommunitiesWeighted2015},
and to get bounds on the number of misclassified vertices in \cite{yunOptimalClusterRecovery2016,xuOptimalRatesCommunity2020}.
Note that weighted SBMs correspond to a special case of the probability-graphons we study in this article
where the space of vertex-labels is finite (they correspond to the stepfunction probability-graphons we define
	in \Cref{section_cut_distance}).
 \medskip        

	Concomitantly to our work,  in  \cite{ayiGraphLimitInteracting2023}, the  authors
        studied mean-field  equations on  large real-weighted  graphs modeling interactions 
        with      a     probability     kernel from
        $[0,1]^2$ to $\ProbaX{\R}$ the set of
        probability  measures  on  $\R$,  but  they  did  not  study  the
      topological   properties  of the set of  those   probability  kernels.   
      Recently, in \cite{hladkyRandomMinimumSpanning2023}, the authors
        studied the limit of the total weight of the minimum spanning tree (MST)
        for a  sequence of random  weighted graphs.  Following  what has
        been    done    for    the     uniform    spanning    tree    in
        \cite{hladkyLocalLimitUniform2018,   archerGHPScalingLimit2022},
        one  expects the  local  and scaling  limits of  the  MST to  be
        directly constructed from the
        limit of the random weighted graphs.

\medskip

Motivated by  those examples, we shall  consider \emph{probability-graphons} as
possible limits of large weighted graphs;  they are defined as maps from
$ [0,1]^2$  to the space  of probability  measures $\Proba$ on  a Polish
space  $\Space$.   When $\Space$  is  compact,  this question  has  been
considered    in    \cite{lovaszLimitsCompactDecorated2010}    and    in
\cite[Section 17.1]{Lovasz} using  convergence of homomorphism densities
of subgraphs decorated with real  functions defined on $\Space$, but the
metric  properties of  the set  of probability-graphons  $\Graphon$ have
only    been   established    when    $\Space   $    is   finite,    see
\cite{falgas-ravryMulticolourContainersEntropy2016}.
The work \cite{kunszenti-kovacsMultigraphLimitsUnbounded2022} is  
an  extension of  \cite{lovaszLimitsCompactDecorated2010} where
$\Proba$ is replaced by the dual space $\cZ$
of  a   separable  Banach   space  $\cB$. As $\Proba$ is a subset of the
dual of $\CbFunct$, this approach covers our setting when $\CbFunct$ is
separable, that is, $\Space$  compact   (see
Section~\ref{section_notations} below). The
norm introduced on the space of $\cZ$-valued graphons therein implies
the convergence of homomorphims densities of $\cB$-decorated sub-graphs,
however there is no equivalence \emph{a priori}. 

\medskip

In  this  paper we  study  the  topological  property  of the  space  of
probability-graphons $\Graphon$ when $\Space$ is a general Polish space:
the space $\Graphon$  is a Polish topological space and  we give ``natural''
cut distances on $\Graphon$ which are complete.  
One of the main difficulty is that the space of probability measures
$\Proba$ can be endowed with many distances which induce the topology
of weak convergence, each of them giving rise to a different cut
distance on $\Graphon$.  We prove that the topology induced on $\Graphon
$ does not depend on the initial choice of the distance on $\Proba$,
provided this distance satisfies some simple general conditions. 
However,   we stress that not all of these cut distances are complete.
We also check that 
this topology  characterizes the convergence 
in distribution of the sampled subgraphs with random weights on the edges or equivalently the convergence 
of the homomorphism densities of $\CbFunct$-decorated subgraphs. 
Similarly to the graphon setting, we  prove  the
convergence in distribution of 
large sampled weighted subgraphs from a probability-graphon $W$ to
itself. 
We  also  provide a tightness criterion  for studying the 
convergence  of  weighted   graphs  towards  probability-graphons.

In conclusion, we believe that the unified framework developed here is
easy-to-work-with and will allow to use  probability-graphons to study
large (random) weighted  graphs.

\subsection{New contribution}

Through the article, \emph{measure} will always be used to denote a positive measure.

\subsubsection{Definition of probability-graphons}

In this article, we define an analogue of graphons for weighted graphs,
	which we call \emph{probability-graphons}, and study their properties.
To avoid any confusion, in the rest of the article we say \emph{real-valued graphons}
	instead of graphons.
We consider the general case where weighted graphs take their edge-weights 
	in a Polish space $\Space$ (\eg $\Z$, $\R$ or $\R^d$),
	which thus also covers the case of decorated graphs,
		multi-graphs (graphs with possibly multiple edges between two vertices)
		and dynamical graphs (where edge-weights evolve over time).

We define a \emph{probability-graphon} as a probability kernel $W : [0,1]^2 \to \Proba$,
	where $\Proba$ is the space of probability measures on $\Space$.
        A  probability-graphon can  be interpreted  as follows:  for two
        ‘‘vertex type’’  $x$ and $y$  in $[0,1]$,  the weight $z$  of an
        edge between two vertices of type  $x$ and $y$ is distributed as
        the  probability measure  $W(x,y;\drv z)$.   In particular,  the
        special case  $\Space =  \{0,1\}$ allows to  recover real-valued
        graphons: as any real-valued graphon $w : [0,1]^2 \to [0,1]$ can
        be       represented        as       a       probability-graphon
        $W(x,y;\cdot)  =  w(x,y)\delta_1  +  (1-w(x,y))\delta_0$,  where
        $\delta_z$ denote the Dirac mass located at $z$.  Let us mention
        that it is  possible to define the  probability-graphons on a
        more  general  probability  space   $(\Omega,  \cA,  \mu)$  than
        $[0,     1]$    for  the    vertex-types,     see
        \Cref{rem_vertex_type_omega} for  details.  In this  article, we
        also   define  and   study   the   properties  of   \emph{signed
          measure-valued kernels} which are bounded (in total mass/total
        variation          norm)           measurable          functions
        $W  :  [0,1]^2  \mapsto  \SignedMeas$ whose  values  are  signed
        measures,    but    for    brevity   we    mainly    focus    on
        probability-graphons in this introduction.

As probability-graphons are measurable functions, we identify probability-graphons that are equal
	for almost every $(x,y)\in [0,1]^2$, and we denote by $\Graphon$
		the space of probability-graphons.
Moreover, as we consider weighted graphs that are unlabeled (that is vertices are unordered),
	we need to consider probability-graphons up to ‘‘relabeling’’:
	for a measure-preserving map $\varphi : [0,1] \to [0,1]$ (relabeling map for probability-graphons),
		we define $W^\varphi(x,y;\cdot) = W(\varphi(x), \varphi(y);\cdot)$;
	we say that two probability-graphons are weakly isomorphic if there exists
		measure-preserving maps $\varphi, \psi : [0,1] \to [0,1]$
		such that $U^\varphi = W^\psi$ for \aE $(x,y)\in[0,1]^2$.
We denote by $\UGraphon$ the space of probability-graphons where we identity probability-graphons
	that are weakly isomorphic.

We can always assume that weighted graphs are complete graphs by adding all missing edges
	and giving them a weight/decoration $\partial$ which is a cemetery point added to $\Space$.
Any weighted graph $G$ can be represented as a probability-graphon $W_G$ in the following way:
	denote by $n$ the number of vertices of $G$ and
	divide the unit interval $[0,1]$ into $n$ intervals $I_1, \cdots, I_n$ of equal lengths,
	then $W_G$ is defined for $(x,y)\in I_i \times I_j$ as $W_G(x,y;\cdot) = \delta_{M(i,j)}$,
		where $M(i,j)$ is the weight on the edge $(i,j)$ in $G$.
Note that weighted graphs can be either directed or undirected,
	in the case of undirected weighted graphs their limit objects are symmetric probability-graphons,
	that is probability-graphons $W$ such that $W(x,y;\cdot) = W(y,x;\cdot)$.

\subsubsection{The cut distance for probability-graphons and its properties}

While there is a usual distance on the field of reals $\R$, this is not the case for probability measures,
	 measures or signed measures endowed with the weak topology.
Some commonly used distances include the Prohorov distance $\dmeasP$ which can be defined on measures,
	and the Kantorovitch-Rubinstein norm $\NmeasKRSymbol$ (sometimes also called the bounded Lipschitz norm)
		and the Fortet-Mourier norm $\NmeasFMSymbol$ defined on signed measures 
		but metrizing the weak topology on measures.
	(Note that in general the weak topology is not metrizable on signed measures,
		see Section~\ref{section_notations} below.)
We also use a norm $\NmeasFSymbol$ based on a convergence determining sequence $\F \subset \CbFunct$.
	See Section~\ref{section_examples_distance} for definition of those distances.
To define an analogue of the cut norm for probability-graphons,
	we first need to choose a distance $\dmeas$ that metrizes the weak topology 
		on the space of sub-probability measures $\SubProba$
		(\ie measures with total mass at most $1$);
	we then define the \emph{cut distance} $\dcut$ for probability-graphons as:
\begin{equation*}
\dcut(U,W) = \sup_{S,T\subset [0,1]} 
\dmeas \Bigl(U(S\times T;\cdot), W(S\times T;\cdot) \Bigr)  ,
\end{equation*}
where the supremum is taken over all  measurable subsets $S$ and $T$ of $[0,1]$,
and where $W(S\times T;\cdot) = \int_{S\times T} W(x,y;\cdot)\ \drv x \drv y$ is a sub-probability measure
	and similarly for $U$.
Moreover, if the distance $\dmeas$ is derived from a norm $\NmeasSymbol$ defined on the space of signed measures $\SignedMeas$,
	then the cut distance $\dcut$ derives from the \emph{cut norm} $\NcutSymbol$ defined
		on signed measure-valued kernels: 
\begin{equation*}
\Ncut{W} = \sup_{S,T\subset [0,1]} 
\NmeasLLarge{W(S\times T;\cdot)}   .
\end{equation*}
We then define the unlabeled \emph{cut distance} $\ddcut$ on the space of unlabeled probability-graphons $\UGraphon$ as:
\begin{equation*}
\ddcut(U,W) = \inf_{\varphi} \dcut(U, W^\varphi) = \min_{\varphi, \psi} \dcut(U^\varphi, W^\psi) ,
\end{equation*}
where the infimum is taken over all measure-preserving maps $\varphi$ and $\psi$,
	see \Cref{thm_min_dist} for alternative expressions of $\ddcut$ (including proof that the minimum exist for the second expression)
	and see \Cref{theo:Wm=W} that states that $\ddcut$ is indeed a distance on $\UGraphon$.
In \Cref{prop:weak_regular_2}, we prove an equivalent of the weak regularity lemma
	for probability-graphons.

An interesting fact is that under some conditions on $\dmeas$,
	the topology induced by the associated cut distance $\ddcut$ does not depend 
		on the particular choice of $\dmeas$.
The following proposition is a particular case of Theorem~\ref{theo_equiv_topo} 
	together with \Cref{cor:reg-dist-usuel}. 

\begin{prop}
	\label{intro:theo_equiv_topo}
The cut distances $\ddcutP$, $\ddcutKR$, $\ddcutFM$ and $\ddcutF$
	induce the same topology on the space of probability-graphons $\UGraphon$.
\end{prop}

  Recall that  $\Space$  is a  Polish
  space. We now  state that $\UGraphon$ is also Polish  for the distance
  $\ddcutP$ (but not for $\ddcutF$!), and we refer to \Cref{theo_complete} for other distances.

\begin{theo}\label{intro:theo_complete}
The space of probability-graphons  $(\UGraphon,  \ddcutP)$  is  a  Polish metric  space.
\end{theo}

We prove an analogue of Prohorov's theorem with a tightness criterion for probability-graphons.
We say that a subset of probability-graphons $\cK\subset \UGraphon$ is \emph{tight}
	if the set of probability measures $\{ M_W \,:\, W\in\cK \}$ is tight (in the sense of probability measures),
	where  $M_W(\cdot) =  W([0, 1]^2;\cdot) $. 
The next result is consequence of \Cref{theo_tension_conv} as well as \Cref{cor:reg-dist-usuel}.

\begin{theo}[Compactness property]
  	\label{intro:theo_tension_conv}
Consider the topology  on   $\UGraphon$  from \Cref{intro:theo_equiv_topo}.
    \begin{enumerate}[label=(\roman*)]
    \item If a sequence of elements of
      $\UGraphon$ is  tight, 
      then  it has  a converging subsequence.
    \item  If $\Space$  is compact, then the space  $\UGraphon$ is
      compact.
       \end{enumerate}
\end{theo}

\subsubsection{Sampling from probability-graphons and its link with the cut distance}

Finally, we link the topology of the cut distance $\ddcut$ with subgraph
sampling.   The probability-graphons  allow to  define models  of random
weighted graphs  (the $W$-random graph model)  which generalize weighted
SBM random  graphs, and which  plays the  role of sampled  subgraphs for
probability-graphons.  The $W$-random graph (or sampled subgraph of size
$k$) $\mathbb{G}(k,W)$ has two parameters,  a number of vertices $k$ and
a probability-graphon $W$ for edge-weights,  and is defined as follows:
first let $X_1, \cdots, X_k$  be $k$ independent random ‘‘vertex-types’’
uniformly distributed over $[0,1]$; then  given $X_1, \cdots, X_k$, each
edge receive a  weight independently, where the weight of  the edge $(i,j)$
is distributed as $W(X_i, X_j; \cdot)$.

We also  provide the \as convergence of  sampled subgraphs for
the topology from \Cref{intro:theo_equiv_topo},
see  \Cref{theo_conv_sampled_subgraphs}
together with \Cref{cor:equiv-topo}.

\begin{theo}
Let $W$ be a probability-graphon.
Then, \as the sequence of sampled subgraphs $(\bbG(k,W))_{k\in\N^*}$
	converges to $W$ for 
 	 the topology  from \Cref{intro:theo_equiv_topo}.
\end{theo}

To prove this theorem, we adapt  the proof scheme of \cite[Sections 10.5
and 10.6]{Lovasz}  relying on the  first and second sampling  lemmas for
real-valued  graphons.   The  proof  is  done  using  the  cut  distance
$\ddcutF$    because    of    the   good    approximations    properties
of~$\NmeasFSymbol$.

In   the   case  of   unweighted   graphs,   the  homomorphism   numbers
$\text{hom}(F,G)$ count  the number of  occurence of a graph  $F$ (often
called a  \emph{motif} or a  \emph{graphlet}) as an induced  subgraph of
$G$,  and  their  normalized counterparts,  the  homomorphism  densities
$t(F,G)$  allow   to  characterize  a   graph  (up  to   relabeling  and
twin-vertices  expansion),   and  also  characterize  the   topology  on
real-valued   graphons.    In   the   case  of   weighted   graphs   and
probability-graphons,  we  need  to replace  absence/presence  of  edges
(which is $0$-$1$  valued) by test functions  from $\CbFunct$ decorating
the   edges.     Hence,   we  define   the
\emph{homomorphism density}  of a  $\cG$-graph $F^g$  which is  a finite
graph $F=(V,E)$  whose edges  are decorated with  a family  of functions
$g=(g_e)_{e\in E}$ from a subset $\cG \subset \CbFunct$ (in practice, we
only consider the cases $\cG =  \CbFunct$ or $\cG = \F \subset \CbFunct$
a convergence determining sequence), in a probability-graphon $W$ as:
\begin{equation*}
t(F^g,W) = M_W^F(g) 
:= \ \int_{[0,1]^{V}}  \prod_{(i,j)\in E} W(x_i,x_j; g_{i,j})		\	\prod_{i\in V} \drv x_i	 ,
\end{equation*}
where  $W(x,y;f) = \int_{\Space} f(z)\, W(x,y;\drv z)$.
Moreover, $M_W^F$ defines a measure on $\Space^E$
(which we still denote by $M_W^F$) which is defined by
$M_W^F(\otimes_{e\in E} g_e) = M_W^F(g)$
	for $g = (g_e)_{e\in E}$.
Note that when $F$ is the complete graph with $k$ vertices,
	$M_W^F$ is the joint measure of all the edge-weights of the random graph $\bbG(k,W)$,
	and thus characterizes the random graph $\bbG(k,W)$.

In the counting Lemma~\ref{CountingLemma} and the weak counting Lemma~\ref{WeakCountingLemma},
	we prove that the cut norm $\NcutFSymbol$ allows to control the homomorphism densities.
Conversely, in the inverse counting Lemma~\ref{InverseCountingLemma},
	we prove that the cut norm  $\NcutFSymbol$ can be controlled by the homomorphism densities.
        In particular, the topology of the  cut distance turns out to be
        exactly the topology of  convergence in distribution for sampled
        subgraphs  of any  given  size;  the next  result  is a direct
        consequence of
        \Cref{theo_equiv_topo_dens_homo}
        
\begin{theo}[Characterization of the topology]
	\label{intro:theo_equiv_topo_dens_homo}

Let $(W_n )_{n\in\N}$ and $W$ be unlabeled probability-graphons
	from $\UGraphon$.
The following properties are equivalent:
\begin{enumerate}[label=(\roman*)]
\item\label{equiv_topo_dens_homo_1}
$(W_n)_{n\in\N}$ converges to $W$ for 
	 the topology  from \Cref{intro:theo_equiv_topo}.

\item\label{equiv_topo_dens_homo_2}
$\lim_{n\to\infty} t(F^g,W_n) = t(F^g,W)$  
for all $\CbFunct$-graph $F^g$.

\item\label{equiv_topo_dens_homo_3}
$\lim_{n\to\infty} t(F^g,W_n) = t(F^g,W)$ 
for all $\F$-graph $F^g$, for some convergence determining sequence $\F$.

\item\label{equiv_topo_dens_homo_4}
For all $k\geq 2$, the sequence of sampled subgraphs $(\bbG(k,W_n))_{n\in\N}$
	converges in distribution to $\bbG(k,W)$.
\end{enumerate}
\end{theo}

Now, we can turn back to the initial problem of finding a limit object
	for a convergent sequence of weighted graphs $(G_n)_{n\in\N}$;
	here convergent means that for all $k\geq 2$,
	the sequence $(\bbG(k,G_n)=\bbG(k,W_{G_n}))_{n\in\N}$ of sampled subgraphs of size $k$
		(defined above)
		converges in distribution (to some limit random graph).
Note that the tightness criterion for a sequence of probability-graphons $(W_n)_{n\in\N}$ 
	can be equivalently rephrased as tightness of the sequence $(\bbG(2,W_n))_{n\in\N}$ of sampled subgraphs  of size $2$.
Hence, the convergence in distribution of  the sequence $(\bbG(2,G_n))_{n\in\N}$ implies its tightness,
	and thus the tightness of the sequence of probability-graphons $(W_{G_n})_{n\in\N}$.
Then, \Cref{intro:theo_tension_conv} guarantees the existence of a probability-graphon $W$
	which is a sub-sequential limit of the sequence $(W_{G_n})_{n\in\N}$ in the cut distance $\ddcutF$,
	and then \Cref{intro:theo_equiv_topo_dens_homo} guarantees that for all $k\geq 2$,
		the sequence $(\bbG(k,G_n))_{n\in\N}$ converges in distribution to $\bbG(k,W)$.
		
As a consequence, probability-graphons are precisely the limit objects for
	sequences of weighted graphs $(G_n)_{n\in\N}$ (and also for random weighted graphs)
	whose number of vertices goes to infinity (otherwise the limit would simply be a weighted graph)
	and such that for each size $k\geq 2$, the sequence of sampled subgraphs
		$(\bbG(k,G_n))_{n\in\N}$ converges in distribution.
	
\begin{remark}[Extension to vertex-weights]
The framework we have developed for probability-graphons could easily be extended to add
	weights on the vertices, or equivalently to allow for self-loops (\ie edges linking a vertex to itself).
In this case, weighted graphs and probability-graphons have 
			a two-variable kernel (probability-graphon) $W^{\text{e}}$ for edge-weights as before,
			and a one-variable kernel $W^{\text{v}} : [0,1] \to \Proba$ for vertex-weights.
Note that this implies, as expected, that the same measure-preserving map $\varphi : [0,1]\to[0,1]$
	must be used for both kernels $W^{\text{v}}$ and $W^{\text{e}}$ when relabeling.
\end{remark}

\subsection{Organization of the paper}

The rest of the paper is organized as follows.
In Section~\ref{section_notations}, we define some notations used throughout the paper,
	and remind some properties of the weak topology on the space of signed measures.
In Section~\ref{section_cut_distance},
	we define probability-graphons and signed-measure valued kernels,
	we then define the cut distance and the cut norm and study their properties,
	and we also give some exemple of distances with the Prohorov distance $\dmeasP$,
		the Kanrorovitch-Rubinstein and Fortet-Mourier norms $\NmeasKRSymbol$ and $\NmeasFMSymbol$,
		and the norm $\NmeasFSymbol$ based on a convergence determining sequence.
In Section~\ref{sec:tight-reg},
	we define the steppings
	of a probability-graphon (which are stepfunction approximations corresponding to conditional expectations on $[0,1]^2$),
	we define the tightness criterion for probability-graphons,
	and we prove the weak regularity property of the cut distance.
In Section~\ref{section_tension},
	we prove the theorem linking the tightness criterion with relative compactness for the cut distance,
	we prove that under some conditions the topology of the cut distance does not
		depend on the choice of the initial distance $\dmeas$,
	and we prove that the space of probability-graphons with the cut distance is a Polish space.
In Section~\ref{section_sampling},
	we define the subgraph $\bbG(k,W)$ sampled from a probability-graphon $W$,
	we then prove approximation bound in the cut norm $\NcutFSymbol$ between probability-graphons
		and their sampled subgraphs.
In Section~\ref{section_counting_lemmas},
	we prove the counting lemmas linking the cut distance
		with the homomorphism densities,
	and prove that the topology induced by the cut distance
	coincides with the topology of convergence in distribution for all the sampled subgraphs.

\tableofcontents

\section{Notations and topology on the space of signed measures}\label{section_notations}

Through the article, \emph{measure} will always be used to denote a positive measure.

Let     $\N=\Z_+$  be   the      set     of     non-negative     integers,
$\N^*=\N\backslash  \{0\}$  the  set  of  positive  integers,  and,  for
$n\in\N^*$, we define the integer set $[n] = \{1,\dots, n\}$.
For $k\in \N^*$, the set $[0, 1]^k$ is endowed with the Borel $\sigma$-field 
and the Lebesgue measure $\lambda_k$; and we write $\lambda$ for
$\lambda_k$ when the context is clear. The supremum of a
real-valued
function $f$ defined on $[0, 1]^k$ is denoted by $\norm{f}_{\infty}=\sup_{x\in [0,
  1]^k} f(x)$. 
\medskip

Let $d$ be a distance on a topological space $(X, \co)$.
\begin{enumerate}[label=(\roman*)]
   \item The distance $d$ is  continuous \wrt the  topology $\co$
if the identity map  from $(X,\co)$ to $(X, d)$ is
  continuous.
\item  The distance  $d$ is  sequentially continuous  \wrt the  topology
  $\co$ if for  any sequence $(x_n)_{n\in\N}$ in $X$  which converges to
  some limit $x$ for the topology $\co$, we also have that $\lim_{n\rightarrow \infty }d(x_n, x)=0$.
    \end{enumerate}

Let $d$ and $d'$ be two distances on a space $X$. We say that $d'$ is
continuous (resp.\ uniformly continuous) \wrt $d$ if the identity map from $(X, d)$ to $(X,
d')$ is continuous (resp.\ uniformly continuous). 

 \begin{remark}\label{rem_metrizable_topo_seq}
 If the topology $\co$ is metrizable (\ie can be generated by a distance on the space $X$),
 then the topology on $X$ induced by the distance $d$ is equivalent to $\co$
 if and only if for every sequence with values in $X$, 
 	convergence for $d$ is equivalent to convergence for $\co$
(see \cite[Theorem 4.1.2]{EngelkingGeneralTopology}).
Moreover, when the topology is metrizable, then topological notions
	and their sequential counterparts coincides
	(\eg compact and sequentially compact sets, closed and sequentially closed sets,
		see \cite[Proposition 4.1.1 and Theorem 4.1.17]{EngelkingGeneralTopology}).
 \end{remark}

  \begin{remark}\label{rem:cont_seq_cont}
 For a function, continuity always implies sequential continuity;
 and the converse is also true when the topology is metrizable.
 \end{remark}
 
\medskip

A map $\varphi : \Omega_1 \to \Omega_2$
between two probability spaces $(\Omega_i, \mathcal{A}_i, \pi_i)$,
$i=1,2$, is measure-preserving if it is measurable and
if for every $A\in \mathcal{A}_2$, $\pi_2(A) = \pi_1(\varphi^{-1}(A))$.
In this case, for every measurable non-negative function $f: \Omega_2 \to \R$,
we have:
\begin{equation}
   \label{eq:re-label}
 \int_{\Omega_1} f(\varphi(x))\ \pi_1(\drv x) = \int_{\Omega_2} f(x)\
 \pi_2(\drv x).
\end{equation}
We denote
by $  \InvRelabel$ the set of bijective measure-preserving maps
from $[0,1]$ with the Lebesgue measure to itself, and by
$\Relabel$ the set of   measure-preserving maps
from $[0,1]$ with the Lebesgue measure to itself. 
\medskip

Let $(\Space,\Topo)$ be some (non-empty) Polish space,
and let $\Borel$ be the Borel $\sigma$-field on $\Space$ generated by the topology $\Topo$.
We  denote by  $\CbFunct$ the  space of  real-valued continuous  bounded
functions on $(\Space,\Topo)$.
We   denote  by   $\SignedMeas$  the   space  of  finite signed  measures   on
$(\Space,\Borel)$; $\Meas$  the  subspace  of
 measures; $\SubProba$  the subspace  of  measures  with total
mass at most $1$; and $\Proba$ the subspace of probability measures.  We
have:
\[
  \Proba \subset \SubProba \subset \Meas \subset \SignedMeas .
\]

For  a   signed  measure  $\mu\in\SignedMeas$,  we remind the definition of the  Hahn-Jordan
decomposition $\mu  = \mu^+ -  \mu^-$ where $\mu^+, \mu^-  \in\Meas$ are
mutually singular measures (that  is $\mu^+(A)=0$ and $\mu^-(A^c)=0$ for
some measurable set $A$), as well as the total variation measure of $\mu$
which is defined as  $\vert\mu\vert  =   \mu^+  +  \mu^-  \in  \Meas$.    Note  that  for a measure
$\mu\in\Meas$, we  simply have $\vert  \mu \vert  = \mu$.  
For a signed-measure $\mu\in \SignedMeas$ and a real-valued measurable
function $f$ defined on $\Space$, we write $\mu(f)=\langle \mu, f \rangle=
\int f \, \rd \mu=\int_{\Space} f(x)\, \mu(\rd x)$ the integral of $f$ \wrt $\mu$ whenever it is well defined. 
For  a signed
measure         $\mu\in\SignedMeas$,         we        denote         by
$\TotalMass{\mu}=\mu^+(\Space) + \mu^-(\Space)$ its total mass, which is
also equal  to the supremum  of $\mu(f)$ over  all measurable
functions   $f$   with  values   in   $[-1,1]$.

We endow $\SignedMeas$ with the topology of weak convergence,
that is the smallest topology for which the maps $\mu \mapsto \mu(f)$
are continuous for all $f\in\CbFunct$.
In particular, a sequence  of signed measures $(\mu_n)_{n\in\N}$  weakly converges  to
some $\mu\in\SignedMeas$    if  and only if, for    every   function    $f\in   \CbFunct$,    we   have
$\lim_{n\to +\infty} \mu_n(f) = \mu(f)$.
 Let us recall that $\Meas$ and $\Proba$ endowed with  the topology of weak convergence
are Polish spaces.

\begin{remark}[The weak topology on $\SignedMeas$]
The topology of weak convergence on the set of signed measures $\SignedMeas$
is equivalent to the weak-$*$ topology on $\SignedMeas$
seen as a subspace of the topological dual of $\CbFunct$
(see the paragraph after Definition~3.1.1 in \cite{Bogachev}).
As usual in probability theory, this topology will be simply called the
weak topology (this is also consistent with \cite{Bogachev}).
\end{remark}

We recall that a sequence of $[0,1]$-valued functions $\F = (f_k )_{k\in\N}$ in $\CbFunct$,
with $f_0=\un$ the constant function equal to one, is:
\begin{enumerate}[label=(\roman*)]
\item  \textbf{Separating}   if  for   every  measures   $\mu,\nu$  from
  $\SignedMeas$ (or equivalently just from  $\Meas$) such that for every
  $k\in\N$, $\mu(f_k) = \nu(f_k)$, then $\mu = \nu$.
\item  \textbf{Convergence determining}  if   for  every
  $(\mu_n)_{n\in\N}$  and  $\mu$   measures  from  $\Meas$  such  that we have
  $\lim_{n\to +\infty} \mu_n(f_k)  = \mu(f_k)$ for
  all $k\in\N$, then $(\mu_n )_{n\in\N}$ weakly converges to $\mu$.
\end{enumerate}
Notice that a convergence determining sequence is also separating.
A sequence of functions is separating if and only if it separates the points of $\Space$
	(see \cite[Theorem 3.4.5]{Ethier}).
There always exists a convergence determining sequence on Polish spaces,
see \cite[Corollary 2.2.6]{Bogachev} or 
 the proof of Proposition  3.4.4 in \cite{Ethier} 
 (which are stated for probability measures but can be extended
to finite  positive measures  as we required  that $\un$  belongs to $\F$).
Note that there does not exist a convergence determining sequence for $\SignedMeas$
	as the weak topology is not metrizable on $\SignedMeas$
	(see \Cref{rem:topo1} below).

\begin{remark}[The Borel $\sigma$-field on $\SignedMeas$]\label{rem:borel_algebra_measures}
By \cite[Corollary 5.1.9]{Bogachev},
the Borel $\sigma$-field on $\SignedMeas$,
associated with the weak topology, is countably generated and
can be generated by either:
\begin{itemize}
\item
the family of maps $\mu \mapsto \mu(f_n)$
	where  the sequence $(f_n)_{n\in\N}$ of functions from $\CbFunct$ is separating;
\item
the family of maps $\mu \mapsto \mu(B)$ where $B\in\cA$ 
and the subset $\cA \subset\Borel$ is countable and generates the whole $\sigma$-field $\Borel$
(such subset $\cA$ always exists, see \cite[Corollary 6.7.5]{BogachevMT2}).
\end{itemize}
Note that the Borel $\sigma$-field of a Polish space is generated by 
	any family of Borel functions that separates points
	(see \cite[Theorem 6.8.9]{BogachevMT2}).

 Furthermore, the   maps $\mu   \mapsto  \mu^+$   and  $\mu   \mapsto  \mu^-$   (and  thus   also
$\mu \mapsto \vert\mu\vert$) are measurable 
(see \cite[Theorem~2.8]{MeasurableSetsMeasures} and Remark~\ref{rem:borel_algebra_measures}).
As a consequence, the map $\mu \mapsto \TM{\mu}$ is also measurable
(in fact it is even lower   semicontinuous   by \cite[Theorem~2.7.4]{Bogachev}).
Note that 
$\Proba$ and $\Meas$ are closed, and thus measurable, subsets of $\SignedMeas$.
\end{remark}

We define the following two important properties for subsets of signed measures,
which are related to relative compactness (see \Cref{lemme_Bogachev_Prohorov_theorem} below).
\begin{defin}
  \label{def:tight-bded}
  Let $\cm\subset  \SignedMeas$ be a subset of signed measures.
  \begin{enumerate}[label=(\roman*)]
\item The set  $\cm$ is \emph{bounded}  (in total variation)     
  if:
  \[
    \sup_{\mu\in \cm} \TotalMass{\mu} <+\infty .
  \]
  
  \item The set  $\cm$ is  \emph{tight} if  for all  $\varepsilon>0$, there
exists    a     compact    set     $K\subset    \Space$     such
that:
\[
  \sup_{\mu\in \cm} \vert\mu\vert(K^c)\leq   \varepsilon.
\]
\end{enumerate}
\end{defin}

\begin{remark}[On the compact sets and metrizability of the weak topology]
  \label{rem:topo1}
 Recall that  $\Space$ is a Polish space. 
   We stress that the weak topology on signed measures is not metrizable 
unless it coincides with the strong topology (see \cite[Theorem~4.1]{Varadarajan}),
which happens only when the initial space $\Space$ is finite (see
\cite[Proposition~3.1.8]{Bogachev}).

Moreover, the closed norm ball 
	$\{ \mu\in\SignedMeas \, : \, \TM{\mu} \leq 1\}$
	of $\SignedMeas$ is metrizable
	if and only if $\Space$ is compact
	(see \cite[Proposition 3.1.8 and Theorem 3.1.9]{Bogachev}).
	
Let $\cM\subset\SignedMeas$. The following properties are
          equivalent (see \cite[Theorems 2.3.4 and 3.1.9]{Bogachev}):
          \begin{enumerate}[label=(\roman*)]
          \item $\cM$ is weakly compact
          		(\ie $\cM$ is compact for the weak topology);
          \item $\cM$ is sequentially weakly compact (that is every sequence $(\mu_n)_{n\in\N}$ in $\cM$ 
          			has a subsequence that converges to some limit $\mu\in \cM$);
            \item $\cM$ is compact for the sequential weak topology (for
              which sets are closed if and only if they are closed under weak convergence).
            \end{enumerate}
Moreover, when any of those is true, $\cM$ is tight, bounded, and metrizable in the weak topology.
        Furthermore, the Kantorovitch-Rubinstein and Fortet-Mouriet 
		norms $\NmeasKRSymbol$ and $\NmeasFMSymbol$ 
		(defined in \Cref{sec:def_KR_FM_norms})
	can be used to generate the weak topology on a weakly compact set
	(see {\cite[Remark 3.2.5]{Bogachev}}).

Nevertheless, the weak topology on
	the unit sphere $\{ \mu\in\SignedMeas \, : \, \TM{\mu} = 1 \}$
	of $\SignedMeas$ 
	is always metrizable with a complete metric, making the unit sphere a Polish space,
	however,  the Kantorovitch-Rubinstein and Fortet-Mouriet 
		norms $\NmeasKRSymbol$ and $\NmeasFMSymbol$ 
		do not provide a complete metrization in this case
	(see \cite[Theorem 3.2.8]{Bogachev}).
\end{remark}

\begin{remark}[On the compactness of $\Proba$]
	\label{rem_Proba_compact}
Let $\cM$ be either $\Proba$, $\SubProba$ or the closed norm ball 
	$\{ \mu\in\SignedMeas \, : \, \TM{\mu} \leq 1\}$
	of $\SignedMeas$.
Then, $\cM$ is weakly compact if and only if $\Space$ is compact.

We give a short proof of this statement.
As $\Proba$ is closed in $\SignedMeas$ for the weak topology,
if $\cM$ is weakly compact, then $\Proba$ is also weakly compact,
and thus $\Space$ is compact by \cite[Theorem~3.4]{Varadarajan}.
Conversely, if $\Space$ is compact,
	then by \cite[Theorem 1.3.3]{Bogachev},
	we know that $\SignedMeas$ (endowed with the weak topology)
	is the topological dual space of $\CbFunct$ (endowed with the uniform convergence topology),
	thus using Banach-Alaoglu theorem (see \cite[Theorem 1.3.6]{Bogachev}),
	we get that the closed unit norm-ball of $\SignedMeas$,
		and thus $\cM$, are compact for the weak topology.
\end{remark}

We recall the following result, which is an equivalent of Prohorov's theorem
for signed measures.
\begin{lemme}[Prohorov's theorem for signed measures, {\cite[Theorems 2.3.4 and  3.1.9]{Bogachev}}]
	\label{lemme_Bogachev_Prohorov_theorem}
Let $\Space$ be a Polish space, and let $\cm \subset \SignedMeas$
be a subset of signed measures on $\Space$.
Then the following conditions are equivalent:
\begin{enumerate}[label=(\roman*)]
\item
$\cm$ is relatively sequentially compact, that is
every sequence $(\mu_n)_{n\in\N}$ in $\cm$ contains a subsequence
	which weakly converges in $\SignedMeas$.
\item
$\cM$ is relatively compact for the weak topology,
	that is the closure of $\cM$ is compact for the weak topology.
\item
The family $\cm$ is tight and bounded.
\end{enumerate}
\end{lemme}

\begin{remark}[On the weak sequential topology]
  \label{rem:topo2}
  When the space $\Space$ is infinite, the weak topology does not coincide with the weak sequential topology
  	on $\SignedMeas$ (but recall from \Cref{rem:topo1} that their compact sets are the same).
  Recall that  if the space $\Space$  is compact,  then  the unit norm ball  of $\SignedMeas$ is metrizable,
  	and thus the weak topology and the weak sequential topology coincide on it.
  However, if the space $\Space$ is non-compact, then the weak topology and the weak sequential topology
  	do not coincide on the unit norm ball of $\SignedMeas$.

        We give a short proof of those statements according to $\Space $ being
        compact or not.
        
\begin{enumerate}[label=(\roman*)]
   \item  Remind  that when  $\Space$ is  an infinite  compact space  (for
        instance  $\Space=[0,1]$),   the  Banach  space   $\CbFunct$  is
        infinite-dimensional  and   separable  (using  Stone-Weierstrass
        theorem),      and       its      topological       dual      is
        $(\CbFunct)^*     =      \SignedMeas$     (see     \cite[Theorem
        1.3.3]{Bogachev}).        Thus,        using       \cite[Theorem
        2.5]{humphreysSeparableBanachSpace1996}, we get the existence of
        a  countable subset  which is  weak sequentially  closed yet
        weak dense  in $\SignedMeas$.   In particular,  the weak
        sequential topology and the weak topology do not coincide on
        $\SignedMeas$.

      \item  Assume  that  the  space $\Space$  is  non-compact.   Thus,
        $\Space$ contains a countable closed subset $F$ whose points are
        at  mutual  distances  uniformly  bounded away  from  zero.   By
        \cite[Remark    3.1.7]{Bogachev},   the    weak   topology    on
        $\SignedMeasX{F}$  for a  closed subset  $F$ coincides  with the
        trace of the weak topology on the whole space.  By \cite[Section
        3.1,  p. 102]{Bogachev},  $\SignedMeasX{F}$  is homeomorphic  to
        $\ell^1$ both endowed with their weak topology, weak convergence
        on  $\ell^1$ is  equivalent to  norm convergence,  and the  weak
        topology on  $\ell^1$ is not  sequential, even on the  unit norm
        ball.   Hence,  the  weak   topology  on  $\SignedMeas$  is  not
        sequential, even on the unit norm ball.

      \end{enumerate}
    \end{remark}

We define the notion of a quasi-convex distance,
which generalizes the convexity of a norm.

 \begin{defin}[Quasi-convex distance]
 	\label{def_dist_quasi_convex}
Let $(X, d)$ be a metric space which is a convex subset of a vector space. 
The distance  $d$ is \emph{quasi-convex} if for all $x_1,x_2,y_1,y_2 \in X$
	and all $\alpha \in [0,1]$,
	we have: 
\[ d( \alpha x_1 + (1-\alpha) x_2, \alpha y_1 + (1-\alpha) y_2)
			\leq \max( d(x_1,y_1), d(x_2,y_2) ) . \]
\end{defin}

In particular, any distance (on a convex subset of a vector space) which
derive from a norm  is quasi-convex.

\begin{lemme}\label{lemma_unif_cont_TM}
Let $\dmeas$ be distance on $\MeasEps$ with $\epsilon\in\{  +, \pm\}$ 
	which is quasi-convex and sequentially continuous with respect to the weak topology.
Then, $\dmeas$ is uniformly continuous with respect to  $\TM{\cdot}$ on  $\cM_\epsilon(\Space)$.
\end{lemme}

\begin{proof}
  We  shall simply consider the case   $\cm=\cm_{+}(\Space)$, the
  other case being simpler. We  first check
  that for  all $\mu\in\cM$ and $\eps  > 0$, there exists  $\eta>0$ such
  that for all  $\nu\in\cM$, we have that $\TM{\mu-\nu}  < \eta$ implies
  $\dmeas(\mu,\nu)<\eps$.   As   $\dmeas$   is   sequentially
  continuous  \wrt   the  weak  topology,  it   is  also  (sequentially)
  continuous \wrt  the strong  topology.  Let $\mu\in\cM$  and $\eps>0$.
  Then, the set $\{ \nu \in \cM \ : \  \dmeas(\mu,\nu) < \eps \}$ is an open set
  of  $\cM$  containing $\mu$  both  for  $\dmeas$  and for  the  strong
  topology.  Thus,  it contains a  neighborhood of $\mu$ for  the strong
  topology  $\{ \nu \in  \cM \ : \ \TM{\mu-\nu}  < \eta \}$ for 
   $\eta>0$ small enough. This proves the claim.

As $\dmeas$ is quasi-convex and $\cM$ is a cone, for $\mu,\nu \in \cM$ we have:
\begin{align*}
\dmeas(\mu,\mu+\nu)
= \dmeas\left( \frac{1}{2} \cdot (2 \mu + 0), \frac{1}{2} \cdot (2\mu + 2 \nu) \right)
\leq \max( \dmeas(2\mu, 2\mu), \dmeas(0, 2\nu))
= \dmeas(0, 2\nu) .
\end{align*}

Let $\eps>0$ be  fixed. We choose $\eta\in (0,1)$ such that  $\TM{\nu} < \eta$,
with $\nu\in\cM$,  implies $\dmeas(0,\nu)<\eps$.  Let $\mu,  \nu\in \cM$
be  such  that  $\TM{\mu  -  \nu}  <  \eta  /2$.
Let $\lambda'=\mu+\nu$ and $f$ (resp. $g$) the density of $\mu$
(resp. $\nu$) with respect to $\lambda'$. We set $\pi= \min(f,g) \,
\lambda'$, $\mu'=(f-g)_+ \, \lambda'$ and $\nu'=(f-g)_-\, \lambda'$ so
that $\pi,      \mu',       \nu'\in\cM$, $\mu   =  \pi   +  \mu'$ and
$\nu   =  \pi   +  \nu'$. Since
$\mu'-\nu'=\mu-\nu$ and $\mu'$ and $\nu'$ are mutually singular, we
deduce that  $\TM{\mu'} + \TM{\nu'} < \eta/2$. 
We get:
\begin{align*}
\dmeas(\mu,\nu) = \dmeas(\pi + \mu', \pi + \nu')
& \leq \dmeas(\pi, \pi+\mu') + \dmeas(\pi, \pi+\nu') \\
& \leq \dmeas(0, 2\mu') + \dmeas(0, 2\nu') \\
& \leq 2 \eps.
\end{align*}
Hence, the distance $\dmeas$ is uniformly continuous with respect to  $\TM{\cdot}$ on $\cM$.
\end{proof}

\section{Measured-valued graphons and the cut distance}\label{section_cut_distance}

In Section~\ref{section_def_graphons},  we introduce  the measure-valued
graphons,  which   are a  generalization of
real-valued graphons (\ie $[0, 1]$-valued measurable
  functions defined on $[0, 1]^2$).  We refer to the monography
  \cite{Lovasz}  on real-valued graphons for more details.
In Sections~\ref{section_def_dcut}, \ref{subsection_graphon_relabeling} 
	and \ref{subsection_unlabeled_cut_distance}, 
we introduce  the cut distance, and its unlabeled variant, on the space of
measure-valued graphons  which are analogous  to the ones  for real-valued
graphons       (see       \cite[Chapter 8]{Lovasz}).        In
Section~\ref{section_weak_isomorphism},  we  define a  weak  isomorphism
relation for measure-valued  graphons based on this  distance.  Then, in
Section~\ref{Section_cut_dist_combi},    we    give    an    alternative
combinatorial formulation of the cut distance for stepfunctions.

\subsection{Definition of measure-valued graphons}\label{section_def_graphons}

We start by defining measure-valued kernels and graphons
which are a generalization of real-valued kernels and graphons.
Recall that $\Space$ is a Polish space
and $\SignedMeas$ is the space of \emph{finite} signed measures.

\begin{defin}[Signed measure-valued kernels]
A \emph{signed measure-valued kernel} or \emph{$\SignedMeas$-valued kernel}
is a map $W$ from $[0,1]^2$  to $\SignedMeas$, 
such that:
\begin{enumerate}[label=(\roman*)]
\item   $W$  is   a  \emph{signed-measure}   in  $\rd   z$:  for   every
  $(x,y) \in [0,1]^2$, $W(x,y;\cdot)$ belongs to $\SignedMeas$.
  
\item $W$ is \emph{measurable} in $(x,y)$: 
for every measurable set $A\subset \Space$, 
the function $(x,y)\mapsto W(x,y;A)$ defined on $[0,1]^2$ is measurable.

\item $W$ is \emph{bounded}:
\begin{equation}
   \label{eq:def:TM}
\TM{W}:=    \sup_{x,y\in [0, 1] }\, \TotalMass{W(x, y; \cdot)} <+\infty .
  \end{equation}  
\end{enumerate}
\end{defin}

We denote by  $\Graphon$ (resp. $\SubGraphon$, resp. $\Kernelp$, resp.   $\Kernel$) the space
of  probability  measure-valued  kernels or  simply  probability-graphons
(resp. sub-probability measure-valued kernels,
	resp.  measure-valued  kernels, resp.  signed  measure-valued kernels),
where we identify kernels that are  equal \aE on $[0,1]^2$, with respect
to  the Lebesgue  measure.  Then,  \eqref{eq:def:TM} should  be
read with an  essential supremum instead of a supremum.  In what follows,
we  always  assume for  simplicity  that  we choose  representatives  of
measure-valued kernels such that $\TM{W}$ is also the essential supremum
of $(x,y)\mapsto\TotalMass{W(x, y; \cdot)}$.

For  $\cm\subset\SignedMeas$, we  denote  by $\cW  _\cm$  the subset  of
signed  measure-valued kernel  $W\in  \Kernel$  which are  $\cm$-valued:
$W(x,y; \cdot)\in \cm$ for every $(x,y)\in [0, 1]^2$.

\begin{remark}[On real-valued kernels]\label{rem:real-valued-kernels}
  Let $\Space =\{ 0, 1\}$ be equipped with the discrete topology.  Every
  real-valued  graphon   $w$ can  be     represented  using  a
  probability-graphon  $W$   defined  for  every   $x,y\in[0,1]$  by
  $W(x,y;\drv z)  = w(x,y)  \delta_1(\drv z) +  (1-w(x,y)) \delta_0(\drv
  z)$, where $\delta_z$ is the Dirac mass located at $z$. In
  particular we have that $w(x,y)=W(x,y; \{1\})$ for $x,y \in [0, 1]$. 
\end{remark}

Let $W \in \Kernel$ be a signed measure-valued kernel.
Define the map $W^+ : [0,1]^2 \to \Meas$ to be the positive part of $W$,
	\ie for every $(x,y)\in [0,1]^2$, $W^+(x,y;\cdot)$ is 
	the positive part of the measure $W(x,y;\cdot)$.
Similarly define $W^- : [0,1]^2 \to \Meas$ the negative part of $W$;
and then define $\vert W \vert = W^+ + W^-$ the total variation of $W$
and $\Vert W \Vert = \vert W \vert (\Space)$ the total mass of $W$.

\begin{lemme}[The positive part $W^+$ of a kernel]\label{lem:mesurability_W_+}
The maps $W^+$, $W^-$ and $\vert W \vert$ are all measure-valued kernels,
and the map $\Vert W \Vert : (x,y) \mapsto \TM{W(x,y;\cdot)}$ is measurable.
\end{lemme}

\begin{proof}
The statements for $\vert W \vert$ and $\Vert W \Vert$ are immediate consequences
of the statements for $W^+$ and $W^-$;
and as the proof for $W^+$ and $W^-$ are similar,
we only need to prove that $W^+$ is a measure-valued kernel.
It is immediate that $W^+$ is bounded and that for every $(x,y)\in [0,1]^2$,
$W^+(x,y;\cdot)$ is a measure in $\Meas$.
Thus, we are left to prove the measurability of $W^+$ in $(x,y)$.
By \cite[Proposition 2.1]{MeasurableSetsMeasures} and Remark~\ref{rem:borel_algebra_measures},
a signed measure-valued kernel $U$ is measurable in $(x,y)$
(\ie for every $A\in\Borel$, the map $(x,y)\mapsto U(x,y;A)$ is measurable)
if and only if the map $(x,y)\mapsto U(x,y;\cdot)$ is measurable
from $[0,1]^2$ (with its Borel $\sigma$-field) to $\SignedMeas$ 
equipped with the Borel $\sigma$-field generated by the weak topology.
By \cite[Theorem 2.8]{MeasurableSetsMeasures},
the map $\mu \mapsto \mu^+$, 
that associate to a signed measure the positive part 
	of its Hahn-Jordan decomposition,
is measurable from $\SignedMeas$ to $\Meas$ both endowed
with the Borel $\sigma$-field generated by the weak topology.
Considering the composition of $W$ and $\mu\mapsto\mu^+$, 
we get that $W^+$ is measurable in $(x,y)$
and is thus a measure-valued kernel.
\end{proof}

\begin{remark}[Probability-graphons $W : \Omega\times \Omega \to \Proba$]
	\label{rem_vertex_type_omega}
Similarly to the case of real-valued graphons,
it is possible to replace the vertex-type space $[0,1]$
by any standard probability space $(\Omega,\mathcal{A},\pi)$
	that might be more appropriate to represent vertex-types for some applications,
and to consider probability-graphons of the form
$W : \Omega\times \Omega \to \Proba$.
We recall that a standard probability space $(\Omega,\mathcal{A},\pi)$
	is a probability space such that there exists 
	a measure-preserving map $\varphi : [0,1] \to \Omega$,
	where $[0,1]$ is endowed with the Borel $\sigma$-field
		and  the Lebesgue measure.
In particular, every Polish space endowed with its Borel $\sigma$-field is a standard probability space.
As an example, the space $[0,1]^2$ equipped with the Borel $\sigma$-field
		and the  Lebesgue measure $\lambda_2$ 
	is a standard probability space;
	we will reuse this fact later.

Using the measure preserving map $\varphi$, it is then possible to 
	consider an unlabeled version $W^\varphi$  of $W$ constructed on $\Omega' = [0,1]$,
and to modify the definition of the cut distance $\ddcut$ 
	similarly as in \cite[Theorem 6.9]{jansonGraphonsCutNorm2013}
	to allow each probability-graphons to be constructed on different standard probability spaces.
For simplicity, in this article we only consider the equivalent case where all probability-graphons are constructed on $\Omega = [0,1]$.
\end{remark}

\begin{remark}[Symmetric kernels]
We shall  consider non-symmetric measure-valued kernels and
probability-graphons in order to handle  directed graphs whose
adjacency matrices are thus \emph{a priori} non-symmetric.
We say that a measure-valued kernel or graphon $W$ is symmetric 
if for \aE $x,y \in [0,1]$, $W(x,y;\cdot)=W(y,x;\cdot)$.
\end{remark}

We define stepfunctions  measure-valued kernel which are  often used for
approximation.
 
\begin{defin}[Signed measure-valued stepfunctions]
  \label{def:stepfunction}
  A signed measure-valued kernel $W\in \Kernel$ is a \emph{stepfunction} if there exists a
  finite partition of $[0,1]$ into measurable (possibly empty) sets, say
  $\mathcal{P}=\{S_1,\cdots,S_k\}$, such that $W$ is constant on the sets
  $S_i \times  S_j$, for $1\leq i,j\leq  k$.  We say that  $W$ and the partition
  $\cP$  are \emph{adapted}  to each other.  
  We write  $|\cp|=k$ the  number of  elements of  the partition
  $\cp$.
 \end{defin}

\subsection{The cut distance}\label{section_def_dcut}

We define a distance and a  norm on signed measure-valued  graphons and kernels,
called  the \emph{cut  distance}  and the  \emph{cut norm}  respectively
which are analogous to the cut norm for real-valued graphons and
kernels, see  \cite[Chapter 8]{Lovasz}.
For a signed measure-valued kernel $W\in\Kernel$
and a measurable subsets $A\subset [0,1]^2$,
we denote by $W(A;\cdot)$  the signed measure on $\Space$ defined by:
\[
  W(A;\cdot) = \int_{A} W(x,y;\cdot)\ \drv x \drv y.
\]

\begin{defin}[The cut distance $\dcut$]
Let $\dmeas$ be a quasi-convex distance on $\cM$ a convex subset of
$\SignedMeas$ containing the zero measure. The associated \emph{cut
  distance} $\dcut$   is the
function  defined on $\KernelM^2$
by:
\begin{equation}
  \label{def_dcut}
\dcut(U,W) = \sup_{S,T\subset [0,1]} 
\dmeas \Bigl(U(S\times T;\cdot), W(S\times T;\cdot) \Bigr)  ,
\end{equation}
where the supremum is taken over all  measurable subsets $S$ and $T$ of $[0,1]$.
\end{defin}
Notice that the right-hand side of~\eqref{def_dcut} is well defined as $\cM$ contains the zero measure (and thus if $U$ belongs to $ \KernelM$ then $U(A; \cdot)$ belongs to $\cM$). 

\begin{defin}[The cut norm $\NcutSymbol$]
The \emph{cut norm} $\NcutSymbol$ associated with a norm $\NmeasSymbol$ on
$\SignedMeas$ is the function  defined on $\Kernel$ by:
\begin{equation*}
\Ncut{W} = \sup_{S,T\subset [0,1]} 
\NmeasLLarge{W(S\times T;\cdot)}   ,
\end{equation*}
where the supremum is taken over all  measurable subsets $S$ and $T$ of $[0,1]$.
\end{defin}

The next proposition states that the cut distance (resp. norm) is indeed
a distance (resp. norm); its extension to distances on $\Meas$ and
$\SignedMeas$ is immediate.

\begin{prop}[$\dcut$ is a distance, $\NcutSymbol$ is a norm]
  The cut distance $\dcut$ associated with a distance $\dmeas$ on
  $\SubProba$ (resp. $\Meas$) is a distance on $\Graphon$ (resp. $\Kernelp$).
  The cut norm $\NcutSymbol$ associated with a
  norm $\NmeasSymbol$ on  $\SignedMeas$ is a norm on $\Kernel$.

Moreover, when the distance $\dmeas$ on $\SubProba$ (resp. $\Meas$)
derives from a norm $\NmeasSymbol$ on $\SignedMeas$,
then the distance $\dcut$ derives also from the norm $\NcutSymbol$.
\end{prop}

\begin{proof}
  Let $\dmeas$ be  a distance on $\SubProba$
  	(the proof for the case $\Meas$ is similar).
  It is  clear that $\dcut$ is symmetric and  satisfies the triangular inequality.   
  Thus, we only need to  prove that  $\dcut$ is  separating.  
  Let $U$  and $W$  be two probability-graphons  
  	such that $\dcut(U,W)  = 0$.  Then,  for every
  	measurable     subsets     $S,     T\subset    [0,1]$,     
	we     have $U(S\times T;\cdot)  =   W(S\times T;\cdot)$.   
  Let  $\F=(f_k  )_{k\in\N}$  be  a separating  sequence.  
  For  every $k\in\N$, and for     every    measurable     subsets    
  	$S,     T\subset    [0,1]$, we have that $U(S\times T;f_k) = W(S\times T;f_k)$. This
        implies that  $U(x,y,f_k)\, \rd x\rd y= W(x,y,f_k)\, \rd x\rd y$ for all $k\in \N$. 
  Hence, we  deduce  that for all $k\in\N$, $U(x,y;f_k) =  W(x,y;f_k)$ 
  	for  almost every $(x,y)\in [0,1]^2$. 
  Thus, $U(x,y;\cdot) =  W(x,y;\cdot)$ 
  	for  almost every $(x,y)\in [0,1]^2$.
  This implies that $\dcut$ is separating on $\Graphon$, 
  	and thus a distance on $\Graphon$.

  The proof for the cut norm is similar. 
  The proof of the last part of the proposition is clear. 
\end{proof}

\subsection{Graphon relabeling, invariance and smoothness properties}
	\label{subsection_graphon_relabeling}

The analogue  of graph  relabelings for graphons  are measure-preserving
maps.   Recall   the  definition   of  a  measure-preserving   map  from
Section~\ref{section_notations},        and         in        particular~\eqref{eq:re-label}.  
Recall $\Relabel$ denotes the set of measure-preserving (measurable)
maps from $[0, 1]$
to $[0, 1]$ endowed with the Lebesgue measure, 
and $\InvRelabel$ denotes its subset of bijective  maps.

The relabeling of a signed measure-valued kernel $W$ 
by a measure-preserving map $\varphi$,
is the signed measure-valued kernel $W^\varphi$ defined
for every $x,y\in [0,1]$ and every measurable set $A\subset \Space$
by:
\[
  W^\varphi(x,y;A) = W(\varphi(x),\varphi(y);A)
  \quad \text{for $x,y\in [0,1]$ and $ A\subset \Space$ measurable}.
\]

We say that a subset  $\ck\subset \Kernel$ is \emph{uniformly bounded} if:
\begin{equation}\label{eq_def_TM_graphon}
\sup_{W\in \ck}\,\, \TM{W}  < +\infty .
\end{equation}

\begin{defin}[Invariance and smoothness of a distance on kernels]
  \label{def:inv-smooth}
Let $d$ be a distance on $\Graphon$ (resp. $\Kernelp$ or $\Kernel$). 
We say that the distance $d$ is:
\begin{enumerate}[label=(\roman*)]
\item  \textbf{Invariant}:  if $d(U,W)=d(U^\varphi,W^\varphi)$  for  every 
bijective  measure-preserving map $\varphi  \in
 \InvRelabel$ and $U,V\in\Graphon$ 
  (resp. $U,V$ belongs to $\Kernelp$ or $\Kernel$).
  
\item \textbf{Smooth}: if \aE  weak convergence implies convergence for $d$,
  that is, if $(W_n)_{n\in\N}$  and $W$
  are  kernels  from $\Graphon$  (resp.  kernels from $\Kernelp$ or $\Kernel$ 
  that are  uniformly bounded and)  such that for \aE $(x,y)\in [0,1]^2$, $W_n(x,y;\cdot)$
  weakly  converges to  $W(x,y;\cdot)$  as $n\to\infty$,  then
  $\lim_{n\rightarrow\infty }d(W_n,W)= 0$.
\end{enumerate}
We say that a norm $N$ on $\Kernel$ is invariant (resp. smooth)
if its associated distance $d$ on $\Kernel$ is invariant (resp. smooth).
\end{defin}

We shall see in Section~\ref{section_examples_distance}
some examples of distances $\dmeas$
for which the associated cut distance $\dcut$ is invariant and smooth. 
The invariance  property from Definition~\ref{def:inv-smooth}  is always
satisfied by the cut distance, and thus also by the cut norm.

\begin{lemme}[$\dcut$ is invariant]\label{lemma_dcut_invariant}
Let $\dmeas$ be a distance on $\SubProba$ (resp.  $\Meas$, resp. $\SignedMeas$). 
Then the cut distance $\dcut$ on $\Graphon$ (resp. $\Kernelp$, resp. $\Kernel$) is invariant.
\end{lemme}

\begin{proof}
For a signed measure-valued kernel $W$, a
bijective measure-preserving map
$\varphi \in \InvRelabel$, 
and  measurable sets $S,T\subset [0,1]$, we have thanks to~\eqref{eq:re-label}:
\[ \int_{S\times T} W^\varphi(x,y;\cdot)\ \drv x \drv y
= \int_{S\times T} W(\varphi(x),\varphi(y);\cdot)\ \drv x \drv y
= \int_{\varphi(S)\times \varphi(T)} W(x,y;\cdot)\ \drv x \drv y
.
\]
Hence, taking the supremum over every measurable sets $S,T\subset [0,1]$,
we get that the cut distance $\dcut$ is invariant.
\end{proof}

When a smooth  distance on $\Graphon$ or  $\Kernelp$ derives from a distance on
$\Proba$ or $\Meas$, we have the following result. 

\begin{lem}[Smoothness and the weak topology]
  \label{lem:dm-smooth-d-cont}
  Let $\dmeas$ be a distance on $\SubProba$ (resp. $\Meas$ or $\SignedMeas$) such that the
  distance $\dcut$ on $\Graphon$ (resp. $\Kernelp$ or $\Kernel$) is smooth. Then, the
  distance $\dmeas$ is continuous \wrt the weak topology on
  $\Proba$ (resp. $\Meas$). 
\end{lem}

\begin{proof}
  Let  $(\mu_n)_{n\in\N}$,  and  $\mu$  be   measures  from  $\Proba$
  (resp. $\Meas$) such that $(\mu_n)_{n\in\N}$ weakly converges to
  $\mu$.  Consider the constant  measure-valued graphons (resp. kernels)
  $W_n  \equiv \mu_n$,  $n\in\N$, and  $W\equiv \mu$.   Then, for  every
  $x,y\in [0,1]$,  $W_n(x,y;\cdot)$ weakly converges  to $W(x,y;\cdot)$
  	as $n\to \infty$.
  As    the    distance    $\dcut$    is    smooth,    we    get    that
  $\lim_{n\rightarrow\infty } \dcut(W_n,W)= 0$.  Considering $S=T=[0,1]$
  in    the     cut    distance,     we    deduce that
  $\lim_{n\rightarrow\infty } \dmeas(\mu_n,\mu) = 0$.
\end{proof}

The next lemma is a partial converse of Lemma~\ref{lem:dm-smooth-d-cont},
it gives sufficient conditions for $\dcut$ to be smooth.
Remind the definition of a quasi-convex distance in \Cref{def_dist_quasi_convex}.

\begin{prop}[$\dcut$ is smooth]
	\label{prop_dcut_smooth}
Let $\dmeas$ be distance on $\MeasEps$ with $\epsilon\in\{ +, \pm\}$
	which is quasi-convex and sequentially continuous \wrt the weak topology (on $\MeasEps$).
Then, the cut distance $\dcut$ is smooth.

Moreover, for all $U,W \in \KernelEps$, and for all measurable $A\subset [0,1]^2$, we have:
\begin{equation}\label{eq_dcut_bound_esssup}
\dmeas(U(A;\cdot),W(A;\cdot)) \leq \esssup_{(x,y)\in A} \dmeas(U(x,y;\cdot), W(x,y;\cdot)).
\end{equation}
\end{prop}

To prove \Cref{prop_dcut_smooth}, we first need to prove the following lemma
	for approximation by $\cM$-valued kernels taking finitely many values.

\begin{lemme}\label{lemme_finite_values_approx}
Let $W\in\Kernel$ and a subset $A\subset [0,1]^2$. 
There exists a sequence $(W_n)_{n\in\N}$ in $\Kernel$
such that $(W_n(A;\cdot))_{n\in\N}$ weakly converges to $W(A;\cdot)$
and for all $n\in\N$, $W_n$ is finitely valued and takes its values in
	$\{ W(x,y;\cdot) \, : \, (x,y)\in A \}$.
\end{lemme}

\begin{proof}
By scaling, we may assume that $\TM{W} \leq 1$.
Let $(f_k)_{k\in\N}$ be a convergence determining sequence
	with $f_0=\un$ and $f_k$ takes values in $[0,1]$.
Thus, for all $(x,y)\in [0,1]^2$, $\epsilon\in\{\pm 1\}$ and $k\in\N$, 
	we have $W_\epsilon(x,y;f_k) \in [0,1]$.
For all $n\in\N$, let $(C_{n,i})_{1\leq i\leq d_n}$ be a partition of $[0,1]^{2(n+1)}$
	into $d_n = n^{2(n+1)}$ hypercubes of edge-length $r_n = 1/n$.
Then, for all $n\in\N$ and $i\in [d_n]$,
	define $B_{n,i} = A \cap ( W_+(\cdot; (f_i)_{0\leq i\leq n},  W_-(\cdot; (f_i)_{0\leq i\leq n})^{-1}(C_{n,i})$;
	thus we get a partition $(B_{n,i})_{1\leq i\leq d_n}$ of $A$.
If $B_{n,i} \neq \emptyset$, fix some $\mu_{n,i} \in \{ W(x,y;\cdot) \, : \, (x,y)\in B_{n,i} \}$.
If $A\neq [0,1]^2$, fix some $\mu_{\partial} \in \{ W(x,y;\cdot) \, : \, (x,y)\in [0,1]^2 \setminus A \}$.
For $n\in\N$, we define 
	$W_n = \ind_{A^c}\; \mu_{\partial} + \sum_{i=1}^{d_n} \ind_{B_{n,i}}\; \mu_{n,i}$,
which is finitely valued and takes its values in $\{ W(x,y;\cdot) \, : \, (x,y)\in A \}$.

Let $k\in\N$ and $\epsilon\in\{\pm\}$. For all $n\geq k$, we have:
\begin{equation*}
\vert W_\epsilon(A;f_k) - (W_n)_\epsilon(A;f_k) \vert
\leq \sum_{i=1}^{d_n} \int_{B_{n,i}} \vert W_\epsilon(x,y;f_k) - (\mu_{n,i})_\epsilon \vert \ \drv x\drv y 
\leq \frac{1}{n} \cdot
\end{equation*}
As $(f_k)_{k\in\N}$ is convergence determining,
this implies that $((W_n)_\epsilon(A;\cdot))_{n\in\N}$ weakly converges to $W_\epsilon(A;\cdot)$ for $\epsilon\in\{\pm\}$.
Hence, $(W_n(A;\cdot))_{n\in\N}$ weakly converges to $W(A;\cdot)$.
\end{proof}

\begin{proof}[Proof of \Cref{prop_dcut_smooth}]
As $\dmeas$ is quasi-convex, \eqref{eq_dcut_bound_esssup} is immediate 
when $U$ and $W$ take only finitely many values.
Now, assume that $U$ and $W$ are arbitrary $\MeasEps$-valued kernels.
Let $\eps>0$. As $\dmeas$ is sequentially continuous \wrt the weak topology, 
using \Cref{lemme_finite_values_approx},
	there exist two $\MeasEps$-valued kernel $U'$ and $W'$
	such that $\dmeas(U'(A;\cdot), U(A\cdot)) < \eps$ and
		$U'$ is finitely valued and takes its values in
	$\{ U(x,y;\cdot) \, : \, (x,y)\in A \}$,
	and similarly for $W'$ and $W$.
Thus, we have:
\begin{equation*}
\dmeas(U(A;\cdot),W(A;\cdot)) \leq 2\eps + \esssup_{(x,y)\in A} \dmeas(U(x,y;\cdot), W(x,y;\cdot)),
\end{equation*}
and this being true for all $\eps>0$, we get \eqref{eq_dcut_bound_esssup}.
\medskip

Let $(W_n)_{n\in\N}$ and $W$ be $\MeasEps$-valued kernels
	which are uniformly bounded by some constant $C<\infty$
	and such that for \aE $(x,y)\in[0,1]^2$, the sequence 
	$((W_n(x,y;\cdot))_{n\in\N}$ converges to $W(x,y;\cdot)$
	for the weak topology, and thus also for $\dmeas$.
Let $\eps>0$ and $S,T\subset [0,1]$.
As $\dmeas$ is quasi-convex and sequentially continuous \wrt the weak topology,
using  \Cref{lemma_unif_cont_TM},
there exists  $\eta>0$  such that for all $\mu,\nu\in\MeasEps$, we have that
	$\TM{\mu-\nu} < \eta$ implies $\dmeas(\mu,\nu)<\eps$.
For all $n\in\N$, define the measurable set:
	\[ A_n = \{ (x,y)\in S\times T \ : \ \dmeas(W_n(x,y;\cdot), W(x,y;\cdot)) < \eps \} . \]
By assumption, we have that $\lim_{n\to\infty} \lambda(A_n) = \lambda(S\times T)$.
Let $N\in\N$ be such that for $n\geq N$, we have $\lambda((S\times T) \setminus A_n) < \eta/C$.
Let $n\geq N$. 
Remark that $W_n((S\times T) \setminus A_n; \cdot)$ and $W((S\times T) \setminus A_n;\cdot)$
	have total mass at most $C \lambda(A_n^c) < \eta$. 
Thus, we have that $\dmeas(W_n(A_n;\cdot), W_n(S\times T;\cdot))<\eps$
	and $\dmeas(W(A_n;\cdot), W(S\times T;\cdot))<\eps$.
Hence, using~\eqref{eq_dcut_bound_esssup} we get that:
\begin{align*}
\dmeas(W_n(S\times T;\cdot), W(S\times T;\cdot)) 
& \leq 2\eps + \dmeas(W_n(A_n;\cdot), W(A_n;\cdot)) \\
& \leq 2\eps + \esssup_{(x,y)\in A_n} \dmeas(W_n(x,y;\cdot), W(x,y;\cdot)) \\
& \leq 3\eps .
\end{align*}
Taking the supremum over $S,T\subset [0,1]$, we get $\dcut(W_n,W) \leq 3\eps$.
This being true for all $\eps>0$, we conclude that $(W_n)_{n\in\N}$ converges to $W$ for $\dcut$,
and thus $\dcut$ is smooth.
\end{proof}

\subsection{The unlabeled  cut distance}
	\label{subsection_unlabeled_cut_distance}

We can now define the cut distance for unlabeled graphons.

\begin{defin}[The unlabeled cut distance $\ddcut$]\label{def:ddcut}
  Set  $\ck\in\{\Graphon,  \Kernelp,  \Kernel\}$.   Let  $d$  be  an
  invariant distance on the kernel space $\ck$.  The premetric $\dd$
  on $\ck$, also called the \emph{cut distance}, is defined by:
\begin{equation}\label{def_ddcut}
\dd(U,W) = \inf_{\varphi\in\InvRelabel} d(U,W^\varphi)	
= \inf_{\varphi\in\InvRelabel} d\left(U^{\varphi},W\right)		 .
\end{equation}
\end{defin}

Notice that  $\dd$ satisfies  the symmetry  property (as  $d$  is
invariant)  and  the triangular
inequality.  Hence, $\dd$   induces a
distance  (that we  still  denote  by $\dd $)  on  the quotient  space
$\cKd = \ck / \simd$
of   kernels in $\ck$  associated
with the equivalence relation $\simd$ defined by
$U\simd W$  if and only if  $\dd(U,W)=0$.

 When  the metric $d=\dcut$ on $\cK=\Graphon$ (resp.
$\Kernelp$, resp. $\Kernel$)
derives from a metric $\dmeas$ on $\SubProba$ (resp. $\Meas$, resp. $\SignedMeas$), and
is thus  invariant thanks to Lemma~\ref{lemma_dcut_invariant},  we write
$\ddcut$ for $\dd$ and $\cKm$ for $\cKdm$. 
We    shall    see    in
Theorem~\ref{theo_equiv_topo} and Corollary~\ref{cor:equiv-topo}
that  under some conditions,  different choices
of distance $\dmeas$, which induces the weak topology on $\SubProba$, lead
to the same quotient space, then simply denoted by $\UGraphon$, with the same topology.

\subsection{Weak isomorphism}
\label{section_weak_isomorphism}

Similarly to Theorem~8.13 in~\cite{Lovasz},
when the distance $\dmeas$ is such that $\dcut$ is invariant and smooth,
we can rewrite the cut distance $\ddcut$ as a minimum instead of an infimum
using measure-preserving maps, see the last equality in~\eqref{eq_premetric}.  

We introduce a weak isomorphism relation that allows to ``un-label'' probability-graphons.

\begin{defin}[Weak isomorphism]\label{def_weak_isomorphism}
We say that two signed measure-valued kernels $U$ and $W$ are \emph{weakly isomorphic}
(and we note $U\sim W$)
if there exists two measure-preserving maps $\varphi, \psi\in\Relabel$
such that $U^\varphi(x,y; \cdot) = W^\psi(x,y;\cdot)$ for \aE
$x,y\in [0,1]$.

We denote by $\UKernel=\Kernel / \sim$ (resp. $\UGraphon=\Graphon / \sim$)
the space of unlabeled signed measure-valued kernels (resp. probability-graphons)
\ie the space of signed measure-valued kernels (resp. probability-graphons) where
we identify signed measure-valued kernels (resp. probability-graphons) that are weakly
isomorphic. 
\end{defin}

Notice that $U\sim W$ implies that $\TM{U}=\TM{W}$
(we recall that signed measure-valued kernels are only defined 
	for \aE $x,y\in [0,1]$
	and that $\TM{W}$ in \eqref{eq:def:TM} is an $\esssup$ in general).
In particular, the notion of uniformly bounded subset 
	defined in \eqref{eq_def_TM_graphon} naturally extends
	to $\UKernel$.
The last part of this section is devoted to the proof of the following
key result. 

\begin{theo}[Weak isomorphism and $\dd$]
   \label{theo:Wm=W}
   Let $d$ be a distance defined on $\Graphon$ (resp.  $\Kernelp$ or $\Kernel$)
   which is invariant  and smooth.  Then, two
   kernels are weakly isomorphic, \ie $U \sim W$, if and only if
   $U \simd W$, \emph{\ie}  $\dd(U,W) = 0$.

   Furthermore, the map $\dd$ is a distance on $\UGraphon=\UGraphond$
   (resp. $\UKernelp=\UKernelpd$ or $\UKernel=\UKerneld$).
\end{theo}

As a first step in the proof of Theorem~\ref{theo:Wm=W}, 
following \cite{Lovasz},
we give a nice description of $\dd$ using couplings. 
We say  that a  measure $\mu$  on $[0,  1]^2$ is  a coupling  measure on
$[0,  1]^2$  (between two  copies  of  $[0,1]$  each equipped  with  the
Lebesgue   measure)  if   the   projection  maps   on  each   components
$\tau, \rho :  [0,1]^2 \to [0,1]$ (where $[0,1]^2$ is  equipped with the
measure  $\mu$  and  $[0,1]$  with the  Lebesgue  measure $\lambda$)  are  measure-preserving.       
Thus      for       every      kernel      $W$      on
$([0,1],\mathcal{B}([0,1]),\lambda)$,  
the  function   $W^{\tau}$  is  a
kernel  on the  probability space  $([0,1]^2,\mathcal{B}([0,1]^2),\mu)$,
and similarly for the projection $\rho$.

Let $\varphi$
be a given measure-preserving map
from $[0,1]$ with the Lebesgue measure to $[0,1]^2$ with a coupling  measure
$\mu$. For an invariant  distance $d$ on $\Graphon$ (resp. $\Kernel$), we
define  a distance, say $d^\mu$,  on kernels on
$([0,1]^2,\mathcal{B}([0,1]^2),\mu)$ by:
\[
  d^\mu(U',W') = d(U'^{\varphi}, W'^{\varphi}). 
\]
It is easy to see that, for $U$ and $W$ kernels on $[0,1]$, 
we have
$d ^ \mu(U^\tau,W^\tau) = d(U,W)$ as $d$ is invariant  and $\tau \circ \varphi$ is a
measure-preserving  map from $[0, 1]$ to itself; and similarly $d ^ \mu(U^\rho,W^\rho) = d(U,W)$.

A straightforward adaptation of 
the proof of \cite[Theorem~8.13]{Lovasz}
gives the next  result.

\begin{prop}[Minima in the cut distance $\dd$]\label{thm_min_dist}
  Let $d$ be a distance defined on $\Graphon$ (resp.  $\Kernelp$ or $\Kernel$)
   which is invariant  and smooth.   Then,  we  have  the
  following alternative  formulations for  the cut distance  $\dd$ on
  $\Graphon$ (resp. $\Kernelp$ or $\Kernel$):
\begin{equation}
  \label{eq_premetric}
  \begin{aligned}
\dd(U,W) 
& = \underset{\varphi\in \InvRelabel}{\inf} d(U,W^\varphi) 
  = \underset{\varphi\in \Relabel}{\inf} d(U,W^\varphi)\\
& = \underset{\psi\in \InvRelabel}{\inf} d(U^\psi,W) 
= \underset{\psi\in \Relabel}{\inf} d(U^\psi,W) 	\\
& = \underset{\varphi, \psi\in \InvRelabel}{\inf} d(U^\psi,W^\varphi) 
= \underset{\varphi,\psi\in \Relabel}{\min} d(U^\psi,W^\varphi) , 
\end{aligned}
\end{equation}
and 
\begin{equation*}
\dd(U,W) 
= \underset{\mu}{\min}\ d^{\mu} \, (U^\tau, W^\rho)
\end{equation*}
where $\mu$ range over all coupling measures on $[0,1]^2$.
\end{prop}

\begin{proof}[Proof of  Theorem~\ref{theo:Wm=W}]
 We deduce from the last equality in~\eqref{eq_premetric}
that $\dd(U,W) =  0$ if and  only if there  exist measure-preserving    maps    
	$\varphi, \psi \in \Relabel$    such    that  
$U^\psi(x,y;\cdot) =  W^\varphi(x,y;\cdot)$ for \aE $x,y\in[0,1]$.  
This gives that the equivalence relations $\simd$ and $\sim$ are the
same. 
\end{proof}

\subsection{The cut norm for stepfunctions}\label{Section_cut_dist_combi}

For a quasi-convex distance $\dmeas$,
the cut distance $\dmeas$ for stepfunctions
can be reformulated using a finite combinatorial optimization. 
For a collection of subsets $\cp$, denote by $\sigma(\cp)$ the $\sigma$-field
generated by $\cp$. 

\begin{lemme}[Combinatorial optimization of quasi-convex $\dmeas$ for stepfunctions]\label{lemma_cut_dist_combi_2}
Let $\dmeas$ be a quasi-convex distance on $\cM$ a convex subset of $\SignedMeas$ containing the zero measure.
Let $U, W\in\KernelM$ be $\cM$-valued stepfunctions adapted to the same finite partition $\cp$. 
Then, there exists $S, T \in \sigma(\cp)$ such that:
\begin{equation*}
  \dcut(U,W) =\dmeas(U(S\times T;\cdot), W(S\times T;\cdot)).
\end{equation*}
\end{lemme}

\begin{proof}
Let $\cp=\{S_1, \ldots, S_k\}$ with $k=\vert\cp\vert$ the size of the partition $\cp$.
  First, remark that the quantity
 $\dcut(U,W) =\dmeas(U(S'\times T';\cdot), W(S'\times T';\cdot))$
depends on $S'$ and $T'$ only through the values of $\lambda(S'\cap  S_i)$ 
and $\lambda(T'\cap S_i)$ for $1\leq i\leq k$.
Thus, the cut distance between $U$ and $W$ can be reformulated as:
\[
\dcut(U,W) = \underset{0\leq \alpha_i, \beta_i  \leq \lambda(S_i) ;\,  1
  \leq i \leq k}{\sup} \, \, \dmeas\left(  
  	\sum_{1\leq i,j \leq k} \alpha_i \beta_j\, \mu_{i,j}(\cdot) ,
	\sum_{1\leq i,j \leq k} \alpha_i \beta_j\, \nu_{i,j}(\cdot) 	\right)  ,
\]
where $\mu_{i, j}$ (resp. $\nu_{i,j}$) is the constant value of $U(x,y; \cdot)$ (resp. $W(x,y; \cdot)$) 
	when $x\in S_i$ and $y\in S_j$. 
Moreover, when we fix the value of  $\beta = (\beta_i)_{1\leq i\leq k}$,
the quantity
\[ \dmeas\left(  
  	\sum_{1\leq i,j \leq k} \alpha_i \beta_j\, \mu_{i,j}(\cdot) ,
	\sum_{1\leq i,j \leq k} \alpha_i \beta_j\, \nu_{i,j}(\cdot) 	\right) \]
	is a  quasi-convex function of $\alpha = (\alpha_i)_{1\leq i\leq k}$, and
thus  realizes its  maximum  on  the extremal  points  of the  hypercube
$\prod_{i=1}^k [0,\lambda(S_i)]$,  \ie  when  $\alpha_i$ equals  $0$ or
$\lambda(S_i)$  for  every $1\leq  i\leq  k$.   By symmetry,  a  similar
argument holds for  $\beta$.  The cut distance can thus  be reformulated as
the combinatorial optimization:
\[
 \dcut(U,W) = \max_{\mathrm{I},\mathrm{J}\subset [ k]} \dmeas\left(  
  	\sum_{i\in \mathrm{I},j\in \mathrm{J}}  \mu_{i,j}(\cdot) ,
	\sum_{i\in \mathrm{I},j\in \mathrm{J}}  \nu_{i,j}(\cdot) 	\right)	 . 
\]
Let $I,J\subset[k]$ that maximizes this combinatorial optimization,
and take $S = \cup_{i\in I} S_i$ and $T = \cup_{j\in J} S_j$
	to conclude.
\end{proof}

\subsection{The supremum in \texorpdfstring{$S$}{} and \texorpdfstring{$T$}{} 
		in the cut distance \texorpdfstring{$\dcut$}{}}

In this section, we prove that the supremum in the cut distance $\dcut$
	is achieved by some subsets $S,T\subset [0,1]$.

For $W\in\SignedMeas$ and $f,g : [0,1] \to [0,1]$  measurable, 
we	define the signed measure:
\begin{equation*}
W(f \otimes g;\cdot) = \int_{[0,1]^2} W(x,y;\cdot) f(x)g(y) \ \drv x\drv y.
\end{equation*}
Remark that if we have $W\in \KernelEps$ with $\epsilon\in\{ 1, \leq 1, +, \pm\}$,
then we have $W(f \otimes g;\cdot) \in \MeasEps$.

\begin{lemme}[The supremum in the cut distance $\dcut$ for quasi-convex distance $\dmeas$]
	\label{lemma_sup_SxT}
Let $\dmeas$ be a quasi-convex distance on $\MeasEps$ with $\epsilon\in\{  +, \pm\}$
	that is sequentially continuous \wrt the weak topology.
Let $U,W\in\KernelEps$.
Then, there exist measurable subsets $S,T \subset [0,1]$ such that $f = \ind_S$ and $g = \ind_T$
	achieve the supremum in:
\begin{equation*}
\sup_{f,g} \dmeas \Bigl( U(f \otimes g;\cdot), W(f \otimes g;\cdot) \Bigr)
\end{equation*}
where the supremum is taken over measurable functions $f,g$ from $[0,1]$ to itself.
\end{lemme}

\begin{proof}
Define the map $\Psi : (f,g) \mapsto \dmeas(U(f \otimes g;\cdot),W(f \otimes g;\cdot))$,
	and denote $C = \sup_{f,g} \Psi(f,g)$,
	where the supremum is taken over measurable functions $f,g$ from $[0,1]$ to itself.
Let $(f_n)_{n\in\N}$ and $(g_n)_{n\in\N}$ be sequences of measurable functions from $[0,1]$ to itself
	such that $\lim_{n\to\infty} \Psi(f_n, g_n)  = C$.
As the unit ball of $L^\infty([0,1],\lambda)$
	is compact for the weak-$*$ topology (with primal space $L^1([0,1],\lambda)$),
	upon taking subsequences, we may assume that $(f_n)_{n\in\N}$ (resp. $(g_n)_{n\in\N}$)
	weak-$*$ converges to some $f$ (resp. $g$) which take values in $[0,1]$.
Thus, $(f_n \otimes g_n)_{n\in\N}$ weak-$*$ converges to $f \otimes g$ in $L^\infty([0,1]^2,\lambda_2)$.
In particular, for every $h\in\CbFunct$, as $W[h]$ is a real-valued kernel, 
	this implies that $\lim_{n\to\infty} W(f_n \otimes g_n; h) = W(f \otimes g ; h)$.
This being true for every $h\in\CbFunct$, we get that the sequence
	$( W(f_n \otimes g_n;\cdot) )_{n\in\N}$ in $\MeasEps$ weakly converges 
	to $W(f \otimes g;\cdot)\in\MeasEps$; and similarly for $U$.
As $\dmeas$ is sequentially continuous \wrt the weak topology on $\MeasEps$,
	we get that $C = \lim_{n\to\infty} \Psi(f_n, g_n)
		= \Psi(f,g)$.

Now, we show that we can replace the functions $f$ and $g$ by functions that only take 
	the values $0$ and $1$ (\ie indicator functions).
We first fix $g$ and do this for $f$.
Let $X$ be a random variable uniformly distributed over $[0,1]$,
	and consider the random function $\ind_{X \leq f}$.
Remark that we have $\Esp[ W( \ind_{X\leq f} \otimes  g;\cdot) ] = W( f \otimes g;\cdot)$,
	and similarly for $U$.
As $\dmeas$ is quasi-convex and sequentially continuous \wrt the weak topology,
we have:
\begin{align*}
C & \geq
\sup_{x\in [0,1]} \dmeas ( U(\ind_{x\leq f} \otimes g;\cdot), W(\ind_{x\leq f} \otimes g;\cdot) )  \\
& \geq   \dmeas \Bigl( \Esp[ U(\ind_{X\leq f} \otimes g;\cdot) ], \Esp[ W(\ind_{X\leq f} \otimes g;\cdot)] \Bigr)   \\
& = \dmeas ( U(f \otimes g;\cdot), W(f \otimes g;\cdot) ) \\
& = C,
\end{align*}
where in the second equality we used the quasi-convex supremum inequality from \eqref{eq_dcut_bound_esssup} 
	with the $\MeasEps$-valued kernels $U'(x,y;\cdot) = U(\ind_{x\leq f} \otimes g;\cdot)$
	and $W'(x,y;\cdot) = W(\ind_{x\leq f} \otimes g;\cdot)$, and $A = [0,1]^2$.
All inequalities being equalities, this imposes:
\begin{equation*}
C  = \sup_{x\in [0,1]} \dmeas ( U(\ind_{x\leq f} \otimes g;\cdot), W(\ind_{x\leq f} \otimes g;\cdot) )
= \lim_{n\to\infty} \dmeas ( U(\ind_{r_n\leq f} \otimes g;\cdot), W(\ind_{r_n\leq f} \otimes g;\cdot) ),
\end{equation*}
for some sequence $(x_n)_{n\in\N}$ in $[0,1]$.
Upon taking a subsequence, we may assume that the sequence $(x_n)_{n\in\N}$
	monotonically converges to some $x\in[0,1]$.
In particular, the sequence of functions $(\ind_{x_n\leq f})_{n\in\N}$ (monotonically) converges
	to the function $f' = \ind_{x \leq f}$ (resp. $f' = \ind_{x < f}$) if $(x_n)_{n\in\N}$
	is non-decreasing (resp. decreasing),
	and thus also weak-$*$ converges in $L^\infty([0,1],\lambda)$.
Using, as in the first part of the proof, the sequential continuity of the function $\Psi$
	\wrt the weak-$*$ topology on $L^\infty([0,1],\lambda)$,
	we get that $\Psi(f', g) = \dmeas ( U(f' \otimes g;\cdot), W(f' \otimes g;\cdot) ) = C$,
	that is we can replace $f$ by the indicator function $f'$.
The same argument allows to replace $g$ by an indicator function.
\end{proof}

\subsection{Examples of distance \texorpdfstring{$\dmeas$}{}}
\label{section_examples_distance}

We consider usual distances and norms on $\Meas$ or $\SignedMeas$ that
induce the weak topology on $\Meas$. 
All the distances we consider are quasi-convex,
and all the norms we consider are sequentially continuous \wrt the weak topology on $\SignedMeas$.
Thus their associated cut distances are invariant and smooth 
	by \Cref{lemma_dcut_invariant} and \Cref{prop_dcut_smooth}.
Properties for the cut distances associated with those distances and norms
	are summarized in Corollaries~\ref{cor:reg-dist-usuel} and
	\ref{cor:equiv-topo}.

In this section, we assume that $(\Space, \dspace)$ is a Polish metric
space, and remind that $\Borel$ denotes its Borel $\sigma$-field. 

\subsubsection{The Prohorov distance \texorpdfstring{$\dmeasP$}{}}

The Prohorov distance $\dmeasP$ is a complete distance defined on the
set of finite  measures $\Meas$ 
that induces the weak topology (see \cite[Theorem 6.8]{Billingsley}). 
It is defined for $\mu, \nu\in \Meas$ as:
\begin{equation}\label{eq_def_d_P}
\dmeasP(\mu,\nu) = \inf\{ \eps > 0 : \forall A\in \mathcal{B}(\Space),\ 
\mu(A) \leq \nu(A^\eps) + \eps 
\quad\text{and}\quad
\nu(A) \leq \mu(A^\eps) + \eps 		\} ,
\end{equation}
where $A^\eps = \{ x\in\Space : \exists y\in A, \dspace(x,y) < \eps \}$.
For probability measures, we only need one inequality in \eqref{eq_def_d_P}
to define the Prohorov distance; however for positive measures
we need both inequalities as two arbitrary positive measures might not have the same total mass.
For $\dmeas=\dmeasP$, we use the subscript $\m= \mathcal{P}$.
We now prove that  the Prohorov distance is quasi-convex.

\begin{lemme}
The Prohorov distance $\dmeasP$ is quasi-convex on $\Meas$.
\end{lemme}

\begin{proof}
Let $\mu_1,\mu_2,\nu_1,\nu_2 \in \Meas$ and let $\alpha \in [0,1]$.
Let $\eps > \max( \dmeasP(\mu_1,\nu_1), \dmeasP(\mu_2,\nu_2) )$,
	then for all $i\in\{1,2\}$ and $B\in\Borel$, we have that
	$\mu_i(B) \leq \nu_i(B^\eps) + \eps$
	and $\nu_i(B) \leq \mu_i(B^\eps) + \eps$.
Taking a linear combination of those inequalities, we get that
	for all $B\in\Borel$, we have that
	$\alpha \mu_1(B) + (1-\alpha) \mu_2(B) \leq \alpha \nu_1(B^\eps) + (1-\alpha) \nu_2(B^\eps) + \eps$,
	and similarly when swapping the role $(\mu_1,\mu_2)$ and $(\nu_1,\nu_2)$.
Hence, we get that $\dmeasP( \alpha \mu_1 + (1-\alpha) \mu_2, \alpha \nu_1 + (1-\alpha) \nu_2) \leq \eps$,
	and taking the infimum over $\eps$, we get that
	$\dmeas( \alpha \mu_1 + (1-\alpha) \mu_2, \alpha \nu_1 + (1-\alpha) \nu_2)
		\leq \max( \dmeas(\mu_1,\nu_1), \dmeas(\mu_2,\nu_2) )$.
\end{proof}

\subsubsection{The Kantorovitch-Rubinshtein and Fortet-Mourier norms}
	\label{sec:def_KR_FM_norms}

The Kantorovitch-Rubinshtein norm $\NmeasKRSymbol$
(sometimes also called the bounded Lipschitz distance)
and the Fortet-Mourier norm $\NmeasFMSymbol$
are two norms defined on $\SignedMeas$ 
that induce the weak topology on $\Meas$
(see Section 3.2 in~\cite{Bogachev} for definition and properties of those norms). 
They are defined for $\mu\in\SignedMeas$ by:

\begin{align*}
  \NmeasKR{\mu}
  &= \sup\left\{ 
\int_{\Space} f\ \drv \mu : \text{$f$ is $1$-Lipschitz and }
\norm{f}_\infty \leq 1\right\}	 ,\\
  \NmeasFM{\mu}
  &= \sup\left\{ 
\int_{\Space} f\ \drv \mu : \text{$f$ is Lipschitz and }
\norm{f}_\infty + \text{Lip}(f) \leq 1\right\}	 ,
\end{align*}
where $\norm{f}_\infty = \sup_{x\in \Space} \vert f(x) \vert$ is the infinite norm
and $\text{Lip}(f)$ is the smallest constant $L>0$ such that $f$ is $L$-Lipschitz.
Those two norms are metrically equivalent, see beginning of Section 3.2 in~\cite{Bogachev}:
\begin{equation}
  \label{eq:met-equiv}
\NmeasFM{\mu} \leq \NmeasKR{\mu} \leq 2 \NmeasFM{\mu}.
\end{equation}
Note that we have $\NmeasKR{\mu} \leq \TotalMass{\mu}$,
	and thus those two norms are
	sequentially continuous \wrt the weak topology on $\SignedMeas$.

An easy adaptation of the proof for Theorem 3.2.2 in~\cite{Bogachev}
gives the following comparison between $\dmeasP$, $\NmeasKRSymbol$ and $\NmeasFMSymbol$.

\begin{lemme}[Comparison of $\dmeasP$, $\NmeasKRSymbol$ and $\NmeasFMSymbol$]
  \label{lemma_comp_Prohorov_KR_FM}
Let $\mu, \nu \in \Meas$. Then, we have:
\begin{equation*}
	\frac{\dmeasP(\mu,\nu)^2}{1+\dmeasP(\mu,\nu)}
	\leq \NmeasFM{\mu-\nu} \leq \NmeasKR{\mu-\nu}
	\leq \bigl(2+ \min(\mu(\Space), \nu(\Space))\bigr) \ \dmeasP(\mu,\nu).
\end{equation*}
\end{lemme}
In particular, the Prohorov distance $\dmeasP$ is uniformly continuous
\wrt 
$\NmeasKRSymbol$ and $\NmeasFMSymbol$ on $\Meas$;
and  $\NmeasKRSymbol$ and $\NmeasFMSymbol$
are uniformly continuous \wrt  $\dmeasP$ on $\SubProba$.

For the special choice $\NmeasSymbol=\NmeasKRSymbol$ (resp. $\NmeasSymbol=\NmeasFMSymbol$), 
we use the subscript $\m= \mKR$ (resp. $\m=\mFM$).

\subsubsection{A norm  based on a convergence determining sequence}
	\label{subsection_norm_F}

From  a  convergence determining  sequence  $\F=(f_k  )_{k\in\N}$,  where
$f_0=\un$  and $f_k \in\CbFunct$  takes values  in  $[0, 1]$,  we define  a norm  on
$\SignedMeas$ metrizing the weak topology on $\Meas$, for
$\mu\in \SignedMeas$, by:
\begin{equation}\label{eq_def_NmeasF}
\NmeasF{\mu}\ = \sum_{k\in \N} 2^{-k} |\mu(f_k)|.
\end{equation}
Note that we have $\NmeasF{\mu} \leq 2 \TotalMass{\mu}$,
	and thus $\NmeasFSymbol$ is
	sequentially continuous \wrt the weak topology on $\SignedMeas$.
For the special choice $\NmeasSymbol=\NmeasFSymbol$, 
we use the subscript $\m= \F$.

Even though  the norm $\NmeasFSymbol$  is not complete when  $\Space$ is
not compact (see Lemma~\ref{d_F_complete} below),  the cut norm $\NcutFSymbol$
and  the cut  distance $\ddcutF$  will  turn out  to be  very useful  in
Sections~\ref{section_sampling}  and   \ref{section_counting_lemmas}  to
link the topology of the cut distance to the homomorphism densities.
Recall $\dmeasF$ is the distance derived from the norm $\NmeasFSymbol$.

\begin{lemme}[$\dmeasF$ is not complete in general]\label{d_F_complete}
Let $\F$ be a convergence determining sequence.
Then, the distance $\dmeasF$ is  complete over $\Proba$
if and only if $\Proba$ is a compact space,
\emph{\ie}, if and only if $\Space$ is  compact.
\end{lemme}

\begin{proof}
Theorem 3.4 in \cite{Varadarajan} states that
$\Space$ is compact if and only if $\Proba$ is compact.
When this is the case, any distance metrizing the weak topology
on $\Proba$ is complete.

Reciprocally, assume that $\dmeasF$ is a complete metric over $\Proba$
	and write $\F=(f_m)_{m\in\N}$.
Let $(\mu_n)_{n\in\N}$ be an arbitrary sequence of probability measures
from $\Proba$.
For every $m\in\N$, as $f_m$ takes values in $[0,1]$,
we have for every $n\in\N$ that $\mu_n(f_m) \in [0,1]$.
Hence, using a diagonal extraction, there exists a subsequence $(\mu_{n_k} )_{k\in\N}$ of
the sequence $(\mu_n)_{n\in\N}$ such that for every $m\in\N$,
the sequence $(\mu_{n_k}(f_m) )_{k\in\N}$ converges, 
that is,  $(\mu_{n_k} )_{k\in\N}$ is a Cauchy sequence
for the distance $\dmeasF$. 
As we assumed the distance $\dmeasF$ to be complete,
this implies that the sequence $(\mu_n)_{n\in\N}$
has a convergent subsequence.
The sequence $(\mu_n)_{n\in\N}$ being arbitrary,
we conclude that the space $\Proba$ is sequentially compact, and thus compact by \Cref{rem:topo1}.
\end{proof}

\medskip

For  $W\in\Kernel$ and  $f\in\CbFunct$, 
we denote by $W[f]$ the real-valued kernel defined by:
\begin{equation}
   \label{eq:notation_graphon_function}
  W[f](x,y) = W(x,y;f) = \int_{\Space} f(z)\ W(x,y;\drv z).
\end{equation}

We denote by $\NcutRSymbol$ (resp. $\NcutRpos{\cdot}$)
	the cut norm (resp. one-sided version of the cut norm) for real-valued kernels
	defined as:
\begin{equation}
	\label{eq_def_NcutR}
\NcutR{w} = \sup_{S,T \subset [0,1]} \left\vert \int_{S\times T} w(x,y)\ \drv x\drv y \right\vert
\quad \text{and} \quad
\NcutRpos{w} = \sup_{S,T \subset [0,1]} \int_{S\times T} w(x,y)\ \drv x\drv y ,
\end{equation}
where $w$ is a real-valued kernel $w$
(see \cite[Section 8.2]{Lovasz}, resp. \cite[Section 10.3]{Lovasz},
	for definition and properties of those objects).

The following two remarks link the cut norm $\NcutFSymbol$
of a signed measure-valued kernel $W$
with the cut norm $\NcutRSymbol$
of the real-valued kernels $W[f]$
for some particular choices of functions $f\in\CbFunct$.
We will reuse those facts in Section~\ref{section_sampling}.

\begin{remark}[Link between $\NcutFSymbol$ and $\NcutRpos{\cdot}$]
	\label{rem:N_F_and_N_R_+} 
For $\mu\in\SignedMeas$ we have:
\begin{equation}\label{eq:def_norm_F_with_eps}
\NmeasF{\mu}
= \sup_{\eps \in \{ \pm 1\}^\N}
\sum_{n\in\N} 2^{-n} \eps_n \mu(f_n)
= \sup_{\eps \in \{ \pm 1\}^\N}
\mu \left( \sum_{n\in\N} 2^{-n} \eps_n f_n \right),
\end{equation}
with $\varepsilon=(\varepsilon_n)_{n\in\N}$. 
Hence, for a signed measure-valued kernel $W\in\Kernel$, we have:
\begin{equation}\label{eq:def_cut_norm_with_eps}
\NcutF{W} 
= \sup_{\eps \in \{ \pm 1\}^\N}  \sup_{S, T\subset [0,1]}
	W\left(S\times T; \sum_{n\in\N} 2^{-n} \eps_n f_n \right)
= \sup_{\eps \in \{ \pm 1\}^\N}	
		\NcutRposLarge{W \left[ \sum_{n\in\N} 2^{-n} \eps_n f_n \right]}.
\end{equation}
\end{remark}

\begin{remark}[Inequality with $\NcutFSymbol$ and $\NcutR{\cdot}$]
For a signed measure-valued kernel $W$, we have:
\begin{align}
\NcutF{W}
& = \sup_{S,T \subset [0,1]} \sum_{n=0}^\infty 2^{-n} \left| \int_{S\times T}  W(x,y,f_n) \ \drv x \drv y \right|
\nonumber \\
& \leq \sum_{n=0}^\infty 2^{-n} \sup_{S,T \subset [0,1]} \left| \int_{S\times T}  W(x,y,f_n) \ \drv x \drv y \right|
\nonumber \\
& = \sum_{n=0}^\infty 2^{-n} \NcutR{W[f_n]}	 
. \label{major_dist_graph}
\end{align}
\end{remark}

\section{Tightness and weak regularity}
\label{sec:tight-reg}

In  this  section,  using  a  conditional  expectation  approach  as  in
\cite[Chapter  9]{Lovasz}, we  provide approximations  of signed measure-valued
kernels and probability-graphons by stepfunctions  with an explicit bound on the  quality of the
approximation.   This procedure  takes into  account that  signed measure-valued
kernels are infinite-dimensional valued.

\subsection{Approximation  by stepfunctions}\label{section_approx}

We start by introducing the partitioning of a signed measure-valued kernel.

\begin{defin}[The stepping operator]
	\label{def_stepping_operator}
Let $W\in \Kernel$ be a signed measure-valued kernel
and  $\mathcal{P}=\{S_1,\cdots,S_k\}$ be a finite partition of $[0,1]$.
We define the kernel stepfunction $W_\mathcal{P}$ adapted to the
partition  $\mathcal{P}$ by averaging $W$ over the partition subsets:
\[
  W_\mathcal{P}(x,y;\cdot) = \frac{1}{\lambda(S_i)\lambda(S_j)}\,
  W(S_i \times S_j;\cdot)
  \qquad \text{for $ x\in S_i, y\in S_j$,}
\]
when $\lambda(S_i)\neq 0$ and $\lambda(S_j)\neq 0$,
and $W_\mathcal{P}(x,y;\cdot) = 0$ the null measure otherwise.
We call the map $W\mapsto W_\cP$ defined on $\Kernel$ the stepping operator 
(associated with the finite partition $\mathcal{P}$).
\end{defin}

Since the signed measure-valued kernel are defined up to an
a.e.\ equivalence, the value of $ W_\mathcal{P}(x,y;\cdot)$ for  $ x\in
S_i, y\in S_j$ when $\lambda(S_i)\lambda(S_j)$ is unimportant.

\begin{remark}[Link with conditional expectation]
	\label{rem:steping_cond_expectation}
The stepfunction $W_\mathcal{P}$ can be viewed as the conditional expectation
of $W$ \wrt  the (finite) sigma-field $\sigma(\cP\times\cP)$ 
on $[0,1]^2$,
where $W:[0,1]^2\to \SignedMeas$ is seen as 
a random signed measure in $\SignedMeas$ and the probability measure on $[0,
1]^2$ is the Lebesgue measure. 
\end{remark}

\begin{remark}[Steppings are convex stable]
Let $\cm\subset \SignedMeas$ be a convex subset of measures,
for instance $\cm$ is $\Proba$, $\SubProba$, $\Meas$ or $\SignedMeas$.
Whenever $W\in \Kernel$ is a $\cm$-valued kernel, 
then by simple computation its stepping $W_\mathcal{P}$
is also a $\cm$-valued kernel. 
\end{remark}

In the following remark, we give a characterization of refining partitions
	that generate the  Borel  $\sigma$-field of $[0,1]$.
\begin{remark}[On refining partitions that generates the Borel $\sigma$-field]
	\label{rem_refining_partitions}
  Let  $(\cP_{k})_{k\in\N}$ be  a sequence  of refining  partitions of
  $[0, 1]$. 
  It  generates   the  Borel
  $\sigma$-field of $[0,1]$ (that is, $\{  S   \,  :  \,   S\in\cP_k,\,  k\in\N  \}$  
  generates   the  Borel $\sigma$-field of $[0,1]$)   
  if and only if  $(\cP_{k})_{k\in\N}$  separates points (that is,   
  for every  distinct $x,y\in [0,1]$, there exists  $k\in\N$  
  such  that  $x$  and  $y$ belong  to  different  classes  of $\cP_k$). 
	
  Indeed, assume that  $(\cP_{k})_{k\in\N}$ separates points,
  and consider   the   countable   family  of   Borel-measurable   functions
  $\F  = \{  \ind_S  \, :  \, S\in\cP_k,\,  k\in\N  \}$ which  separates
  points.   Thus, by  \cite[Theorem~6.8.9]{BogachevMT2}  (remark that  a
  Polish       space      is       a       Souslin      space,       see
  \cite[Definition~6.6.1]{BogachevMT2}), the  family $\F$  generates the
  Borel  $\sigma$-field of  $[0,1]$.  This  implies that  the family  of
  Borel sets $\{  S \, : \, S\in\cP_k,\, k\in\N  \}$ generates the Borel
  $\sigma$-field of $[0,1]$.
  
  Conversely, assume there exist $x,y\in [0,1]$
  which are not separated by $(\cP_{k})_{k\in\N}$,
  \ie for all $k\in \N$, $x$ and $y$ belong to the same class of $\cP_k$.
  This implies that the set $\{x\}$ does not belong to the $\sigma$-field generated
  by $(\cP_{k})_{k\in\N}$, and thus $(\cP_{k})_{k\in\N}$ does not
  generate   the  Borel $\sigma$-field of $[0,1]$.
\end{remark}

Recall the definition of the norm  $\TM{\cdot}$ on $\Kernel$ defined in~\eqref{eq:def:TM}. 
The following lemma allows to approximate any signed measure-valued kernel
	by its steppings.

\begin{lemme}[Approximation using the stepping operator]
\label{lem:approx-W-p}
Let $W\in \Kernel$ be a signed measure-valued kernel
 (which is bounded by definition). 
Let  $(\mathcal{P}_n  )_{n\in\N}$  be  a  refining sequence  of   finite
partitions  of  $[0,1]$  that  generates the Borel $\sigma$-field on $[0,1]$.   
Then,  the  sequence $(W_{\mathcal{P}_n}  )_{n\in\N}$  
is uniformly  bounded by  $\TM{W}$, and
weakly converges to $W$ almost everywhere (on $[0,1]^2$).
\end{lemme}

\begin{proof}
Set $W_n= W_{\mathcal{P}_n}$ for $n\in \N$. 
  By definition of the stepping operator,
we have for every $n\in\N$ and every $(x,y)\in[0,1]^2$ 
that the total mass of $W_n(x,y; \cdot)$ is upper bounded by $\TM{W}$.

Recall  that for  $W\in\Kernel$ and  $f\in\CbFunct$, 
the real-valued kernel $W[f]$ is defined by~\eqref{eq:notation_graphon_function}.
First assume that $W\in\Kernelp$.
Let $\F = (f_k )_{k\in\N}$ be a convergence determining sequence, with by
convention  $f_0=\un$.  For  every $k\in\N$  and $n\in\N$,  an immediate
computation                                                        gives
$W_n[f_k]   = (W[f_k])_{\mathcal{P}_n}$. 
For every  $k\in\N$, as
$W[f_k]$  is   a  real-valued  kernel,  we   can  apply  
the closed martingale theorem 
	(as $(W[f_k])_{\mathcal{P}_n}$ can be viewed as a conditional expectation,
		see Remark~\ref{rem:steping_cond_expectation}),
and we get that
$\lim_{n\rightarrow \infty } W_n[f_k] 
=  W[f_k]$ 
almost everywhere, since $(\cP_n)_{n\in\N}$ generates the Borel $\sigma$-field. 
Hence, as the sequence $(f_k)_{k\in\N}$ is convergence determining,
the sequence  $(W_n)_{n\in\N}$
weakly converges to $W$ almost everywhere.

Now, for $W\in\Kernel$, write $W=W^+ - W^-$ where $W^+, W^- \in\Kernelp$
	(see Lemma~\ref{lem:mesurability_W_+}).
By linearity of the stepping operator, 
remark that we have $W_n = (W^+)_{\cP_n} - (W^-)_{\cP_n}$ for all $n\in \N$.
By the first case, we have that the sequence $((W^+)_{\cP_n})_{n\in\N}$ 
	weakly converges \aE to $W^+$,
and similarly for $((W^-)_{\cP_n})_{n\in\N}$ and $W^-$.
Hence, the sequence $(W_n)_{n\in\N}$
weakly converges to $W$ almost everywhere.
\end{proof}

We first provide a separability result on the space of probability-graphons. 

\begin{prop}[Separability of $\Graphon$ and $\UGraphon$]
  \label{theo_separable}
  Let $d$  be a smooth distance  on $\Graphon$ (resp. $\Kernelp$  or $\Kernel$).
  Then, the  space $(\Graphon,  d)$
  (resp.  $(\Kernelp, d)$  or $(\Kernel, d)$) is separable.
  
  If furthermore $d$ is invariant (which implies that $\dd$ is a distance), 
   then the  space $(\UGraphon,  \dd)$ (resp.  $(\UKernelp, \dd)$  or $(\UKernel, \dd)$) is separable.
\end{prop}

\begin{proof}
We shall consider the space of probability-graphons $\Graphon$, as the proofs for
$\Kernelp$ and $\Kernel$ are similar. 
Applying Lemma~\ref{lem:approx-W-p} with the sequence of dyadic partitions,
for every probability-graphon $W$, we can find a sequence of probability-graphon stepfunctions
adapted to finite dyadic partitions and
converging to $W$ almost everywhere on $[0,1]^2$.

As the space $\Space$ is separable, 
the space of probability measures $\Proba$ is also separable
for the weak topology (see \cite[Theorem~6.8]{Billingsley}).
Let $\ca\subset \Proba$ be an at most countable dense (for the weak
topology) subset. Then, for any stepfunction $W\in \Graphon$ adapted 
to a finite dyadic partition, we can approach it everywhere on $[0,1]^2$
by a sequence of  $\ca$-valued stepfunctions adapted 
to the same finite dyadic partition.  

Hence,   for   every   $W\in   \Graphon$,   there   exists   a   sequence
$(W_n )_{n\in\N}$  in the countable  set of  $\ca$-valued stepfunctions
adapted to a finite dyadic partition 
that converges to $W$ almost everywhere on $[0,1]^2$.
As $d$ is smooth, we get that this convergence also holds for $d$.
Thus, the space $(\Graphon, d)$ is complete.

Remind that by Theorem~\ref{theo:Wm=W},
	when the distance $d$ is invariant and smooth,
	then the premetric $\dd$ is a distance on $\UGraphon$.
In that case, convergence for $d$ implies convergence for $\dd$,
and thus the space $(\UGraphon,  \dd)$ is also separable.
\end{proof}

\subsection{Tightness}

Similarly to  the case of  signed measures (remind Lemma~\ref{lemme_Bogachev_Prohorov_theorem}), 
we introduce a  tightness criterion
for signed measure-valued kernels that characterizes relative
compactness, see  Proposition~\ref{prop:tight} below. 
For a signed measure-valued kernel $W\in \Kernel$,
we define the measure $M_W \in \Meas$ by:
\begin{equation}
  \label{eq:def-MW}
M_W(\drv z) 
= \vert W \vert ([0,1]^2; \drv z)
= \int_{[0,1]^2} \vert W\vert (x,y;\drv z) \ \drv x \drv y,
\end{equation}
where for every $x,y\in[0,1]$, $\vert W\vert (x,y;\cdot)$
	is the total variation of $W(x,y;\cdot)$
	(see Lemma~\ref{lem:mesurability_W_+}).
In particular, if $W$ is a probability-graphon
then $M_W$ is a probability measure from $\Proba$. Notice also that if
$W$ and $U$ are weakly isomorphic, then $M_W=M_U$, so that the
application $W \mapsto M_W$ can be seen as a map from $\UGraphon$
(resp. 
$\UKernel$) to $\Proba$ (resp. $\Meas$). 

\begin{defi}[Tightness criterion]
  \label{defi:tight}
  A subset  $\mathcal{K}\subset \Kernel$ (resp. $\mathcal{K}\subset \UKernel$)
   is said  to be
  tight      if     the      subset     of  measures
  $\{ M_W : W\in\mathcal{K} \} \subset \Meas$ is tight.
\end{defi}

The following proposition shows the equivalence between a global tightness criterion
and a local tightness criterion.
Recall that uniformly bounded subsets of $\UKernel$ are discussed after  \Cref{def_weak_isomorphism}. Recall also $\lambda_2$ is the Lebesgue measure on $[0, 1]^2$.

\begin{prop}[Alternative tightness criterion]
  \label{prop:tight}
  Let $\ck\subset \Kernel$ (or  $\ck\subset \UKernel$) be a uniformly bounded subset  of  signed measure-valued  kernels.
  The set $\ck$ is tight if and only if
  for every $\eps>0$, there exists a compact set $K \subset \Space$, 
  such that for every $W\in \mathcal{K}$ we have:
\begin{equation}
   \label{eq:lK>1}
    \lambda_2 \Bigl( \{ (x,y)\in [0,1]^2 : \vert W\vert(x,y; K^c) \leq \eps \} \Bigr) >  1-\eps.
  \end{equation} 
\end{prop}

\begin{proof}
  As the left hand side of \eqref{eq:lK>1} is invariant by relabeling,
  	it is enough to do the proof for $\Kernel$.
  Let   $\ck\subset    \Kernel$   be    uniformly   bounded    and   set
  $C= \sup_{W\in  \ck} \TM{W} <\infty$.  Assume that  for every  $\eps>0$, there
  exists  a compact  set $K  \subset \Space$,  such that~\eqref{eq:lK>1}
  holds for every  $W\in \mathcal{K}$.  Let $1 > \eps  > 0$. Thus, there
  exists  a  compact  subset  $K\subset  \Space$  such  that  for  every
  $W\in \mathcal{K}$  there exists a  subset $A_W  \subset [0,1]^2$
  with  (Lebesgue)  measure  at  least $1-\eps$,  such  that  for  every
  $(x,y)\in A_W$, we have $\vert W\vert(x,y;K^c) \leq \eps$.
We have that for all $W\in \ck$:
\[
M_W(K^c) = \int_{[0,1]^2} \vert W\vert(x,y;K^c)\ \drv x\drv y
\leq  \TM{W} \lambda_2 (A_W^c )+ \varepsilon \lambda_2(A_W)
\leq  (C+1) \varepsilon. 
\]
Hence, the subset of   measures 
$\{ M_W : W\in\mathcal{K} \} \subset \Meas$ is tight,
that is  $\mathcal{K}$ is tight.

\medskip

Conversely, suppose that $\mathcal{K}$ is  tight.  Let $\eps >0$.
There  exists  a compact  set  $K\subset  \Space$  such that  for  every
$W\in    \mathcal{K}$,   we    have    $M_W(K^c)   < \eps^2$.     For
$W\in                        \mathcal{K}$,                        define
$A_W = \{ (x,y) \in [0,1]^2 : \vert W\vert(x,y;K^c) \leq \eps \}$.  We have:
\[
  \eps^2 >  M_W(K^c)
  = \int_{[0,1]^2} \vert W\vert(x,y;K^c)\ \drv x\drv y 
  \geq \varepsilon \lambda_2(A_W^c).
\]
Hence, $\lambda_2(A_W) >  1 - \eps$,
and consequently Equation~\eqref{eq:lK>1}
holds. 
\end{proof}

We end this section on a continuity result of the map $W \mapsto M_W$.

\begin{lemme}[Regularity of the map $W \mapsto M_W$]
  \label{conv_graphon_conv_measure}
  Let  $\dmeas$  be  a  distance on  $\SubProba$ (resp. $\Meas$). Then
  the map $W\mapsto M_W$ is  1-Lipschitz, and thus continuous,  from $(\UGraphonm,
  \ddcut)$  (resp. $(\UKernelpm, \ddcut)$) to $(\Proba, \dmeas)$ (resp. $(\Meas, \dmeas)$). 
\end{lemme}
\begin{proof}
Taking $S=T=[0,1]$
in Definition~\eqref{def_dcut} of $\dcut$, we get that $\dmeas(M_U,M_W) \leq \dcut(W,U) $. 
As  $M_{U^\varphi}=M_U$ for any measure-preserving map $\varphi$
thanks to~\eqref{eq:re-label}, we deduce from Definition~\eqref{def_ddcut} of $\ddcut$ that 
$\dmeas(M_U,M_W) \leq \ddcut(U,W)$.
\end{proof}

\subsection{Weak regularity}
We shall consider the following extra regularities  of distances on the set of
signed measure-valued kernels \wrt the stepping operator. 
For a finite partition $\cP$, denote by $\vert\cP\vert$ the size of the partition $\cP$,
\ie the number of sets composing $\cP$.

\begin{defi}[Regularities of distances]
  \label{defi:extra-prop}
Let $d$ be a distance  on $\Graphon$ (resp. $\Kernelp$ or $\Kernel$). 
\begin{enumerate}[label=(\roman*)]

\item\label{hypo_dist_WRL}  \textbf{Weak  regularity.} The distance
  $d$ is weakly  regular
  if  whenever  the
  subset  $\ck$ of $ \Graphon$  (resp. $\Kernelp$ or  $\Kernel$) is  tight (resp. tight and
  uniformly bounded),  then  for  every
  $\eps   >  0$,   there  exists   $m\in\N^*$,  such   that  for   every
   kernel    $W\in\ck$,  and  for   every  finite
  partition $\mathcal{Q}$  of $[0,1]$,  there exists a  finite partition
  $\mathcal{P}$  of   $[0,1]$  that  refines  $\mathcal{Q}$   such
  that:
  \[
    \vert\mathcal{P}\vert  \leq m  \vert\mathcal{Q}\vert
    \quad\text{and}\quad
    d(W, W_{\mathcal{P}}) < \eps.
  \]

\item\label{hypo_dist_stepping_optimal} \textbf{Regularity \wrt the
    stepping  operator.} The distance
  $d$ is regular \wrt the
    stepping  operator if
 (resp. for any finite  constant $C\geq 0$) there exists a finite constant $C_0>0$ 
  such  that for every $W, U$ in
  $\Graphon$ (resp. in $\Kernelp$ or $\Kernel$, with  $\TM{W}\leq  C$ and
  $\TM{U}\leq C$) and
  every finite partition $\cp$ of $[0, 1]$, then we have:
\begin{equation}
   \label{eq:opt-step}
    d(W,W_\cp) \leq C_0\,  d(W,U_\cp).
\end{equation}
\end{enumerate}
We say that a norm $N$ on $\Kernel$ is weakly regular (resp. regular \wrt the
    stepping  operator)
if its associated distance $d$ on $\Kernel$ is weakly regular (resp. regular \wrt the
    stepping  operator).
\end{defi}

The weak regularity property is an analogue to the weak regularity lemma for
real-valued graphons (see  \cite[Lemma~9.15]{Lovasz}).
If a distance $d$ is weakly regular, then for a subset $\ck \subset \SignedMeas$
which is tight and uniformly bounded, every $\ck$-valued kernel
can be approximated by a stepfunction with a uniform bound.
The regularity \wrt the stepping  operator states that the 
stepping operator gives an almost optimal way to approximate a signed measure-valued kernel
using stepfunctions adapted to a given partition.

\subsubsection{An example of cut distance regular \wrt  the
    stepping  operator}

Remind the definition of a quasi-convex distance in \Cref{def_dist_quasi_convex}.
We first show that the stepping operator is 1-Lipschitz  for the cut distance $\dcut$
when the distance $\dmeas$ is quasi-convex.

\begin{lemme}[The stepping operator is $1$-Lipschitz]
	\label{contractivity_2}
 Let $\dmeas$ be a quasi-convex distance on $\cM$ a convex subset of $\SignedMeas$ containing the zero measure.
Then, the stepping operator associated with a given finite partition of
$[0, 1]$ is 1-Lipschitz  on $\KernelM$ for the cut distance $\dcut$.
\end{lemme}

\begin{proof}
Let $U,W\in\KernelM$ be $\cM$-valued kernels, and let $\mathcal{P}$ 
be a finite measurable partition of $[0,1]$.
As $U_\mathcal{P}$ and $W_\mathcal{P}$ are stepfunctions adapted to the same partition,
and as $\dmeas$ is quasi-convex,
we can use  Lemma~\ref{lemma_cut_dist_combi_2} 
to get for some $S, T\in\sigma(\cp)$ that:
\[
\dcut(U_\mathcal{P}, W_\mathcal{P})  = \dmeas(U_\mathcal{P}(S\times T;\cdot), W_\mathcal{P}(S\times T;\cdot))
= \dmeas(U(S\times T;\cdot), W(S\times T;\cdot))
\leq \dcut(U,W),
\]
where the  second equality comes  from the  fact that the  integrals are
equals as  $S, T  \in \sigma(\mathcal{P})$ and  thus the  integration is
over full  steps of the partition.  
Hence, the stepping  operator is 1-Lipschitz  on $\KernelM$ for the cut distance $\dmeas$.
\end{proof}

For a quasi-convex distance $\dmeas$, the cut distance $\dcut$
is regular \wrt the stepping operator with  $C_0=2$
in~\eqref{eq:opt-step} (and one can take $C=+\infty$ 
	in Definition~\ref{defi:extra-prop}~\ref{hypo_dist_stepping_optimal}).
	
\begin{lemme}[$\dcut$ is regular \wrt the stepping operator]\label{lemma_step_opti_2}
Let $\dmeas$ be a quasi-convex distance on $\cM$ a convex subset of $\SignedMeas$ containing the zero measure.
Let $W, U\in \KernelEps$ be $\KernelM$-valued kernels,
and let  $\mathcal{P}$ be a finite  partition of $[0,1]$. 
Then, we have:
\[
  \dcut(W, W_\cP) \ \leq\ 2 \dcut(W, U_\cp) .
\]
\end{lemme}

\begin{proof}
The proof is similar to the proof of \cite[Lemma~9.12]{Lovasz}.
As $\dmeas$ is quasi-convex, using Lemma~\ref{contractivity_2},
we get:
\begin{equation*}
\dcut(W, W_\mathcal{P})
 \leq\ \dcut(W , U_\cp) + \dcut(U_\cp , W_\mathcal{P})
 \leq\  2 \dcut(W ,U_\cp) .
\end{equation*}
\end{proof}

\subsubsection{An example of weakly regular cut distance}

We have the following general result. 
Recall Definitions~\ref{def:inv-smooth} and~\ref{defi:extra-prop} on distances and norms
on $\KernelEps$, with $\eps\in\{ +, \pm \}$, being  invariant, smooth, weakly
  regular and regular \wrt the stepping operator. 
  
\begin{prop}[Weak regularity of $\dcut$]\label{prop:weak_regular_2}
Let $\dmeas$ be a quasi-convex distance  on $\MeasEps$, with $\epsilon\in\{  +, \pm\}$,
  which is sequentially continuous   \wrt  the  weak topology.
  Then, the cut distance $\dcut$ on $\KernelEps$ is
  invariant, smooth, weakly regular 
  and regular \wrt the stepping operator. 
\end{prop}

Using  results  from  Section~\ref{section_examples_distance},  we directly get  the
following weak regularity of the cut distance $\dcutP$
and the cut norms $\NcutFSymbol$, $\NcutKRSymbol$ and $\NcutFMSymbol$.

\begin{cor}[Weak regularity of usual distances and norms]
  \label{cor:reg-dist-usuel}
  The cut norms $\NcutFSymbol$, $\NcutKRSymbol$ and $\NcutFMSymbol$
  (resp. the cut distance $\dcutP$)
 on $\Kernel$ (resp. $\Kernelp$)
 are invariant, smooth, weakly
  regular and regular \wrt the stepping operator.
\end{cor}

\begin{proof}[Proof of \Cref{prop:weak_regular_2}]
  We deduce from Lemmas~\ref{lemma_dcut_invariant} and~\ref{lemma_step_opti_2},  Proposition~\ref{prop_dcut_smooth} and
that  the cut distance $\dcut$ 
  on $\KernelEps$ is invariant, smooth and regular \wrt the stepping operator. 
 We  are left to prove  that  $\dcut$  is weakly  regular  on $\KernelEps$. We prove it by considering in the first step the case $\Space$ compact and in a second step the general case $\Space$ Polish. 
 
  \medskip
  
  \textbf{Step 1.} We assume $\Space$ compact.   As in the definition of weak regularity,
let $\ck \subset \KernelEps$ be a subset of $\MeasEps$-valued kernels
  that is tight and uniformly bounded by some finite constant $C$.
Let $\cM\subset\MeasEps$ be the subset of elements of $\MeasEps$ with total mass at most $C$;
	in particular $\cM$ is a convex set containing $0$ 
	and  $\ck\subset \KernelM$. 
As $\Space$ is compact, from Remarks~\ref{rem:topo1} and \ref{rem_Proba_compact},
	we know that the weak topology is metrizable on $\cM$
	and that $\cM$ is compact, and thus sequentially weakly compact. 
Hence, as $\dmeas$ is sequentially continuous \wrt the weak topology on $\MeasEps$,
we  have  that  $(\cM,   \dmeas)$  is sequentially compact, and thus compact.

Denote by $B(\mu,r) = \{ \nu\in\cM  : \dmeas(\mu, \nu) < r  \}$ the open ball
centered at $\mu\in \cM$ with radius $r>0$.
Let    $\eps>0$.     As    $\cM$    is     compact,     there    exist
$\mu_1,   \dots,    \mu_n   \in    \cM$,   $n\in\N^*$,    such   that
$\cM  = \cup_{i=1}^n  B(\mu_i,\eps)$.   For $1\leq  i\leq n$,  define
$A_i  = B(\mu_i,  \eps) \setminus  \cup_{j<i} B(\mu_j,  \eps)$, so  that
$\{A_1,  \ldots,  A_n\}$ is  a  finite  partition (with  possibly  some empty
sets) of $\cM$.

Every $\cM$-valued kernel $W$
can be approximated by a $\{\mu_1, \dots, \mu_n\}$-valued kernel $U$
defined for every $(x,y)\in[0,1]^2$ by 
$U(x,y;\cdot) = \mu_i$ for $i$ such that $W(x,y;\cdot)\in A_i$.
Thus, by construction, we have that for every $(x,y)\in[0,1]^2$, 
$\dmeas(W(x,y;\cdot) , U(x,y;\cdot)) < \eps$.
Applying the quasi-convex supremum inequality from \eqref{eq_dcut_bound_esssup} 
to $W$ and $U$, we get that:
\begin{equation*}
  \dcut(W,U) \leq \esssup_{(x,y)\in [0,1]^2} \dmeas(W(x,y;\cdot) , U(x,y;\cdot)) \leq \eps  .
\end{equation*}

Then, as  the stepping operator is  1-Lipschitz  for the  cut norm, see
Lemma~\ref{contractivity_2}, we  have for  any finite partition  $\mathcal{P}$ of
$[0,1]$ that:
  \begin{align}
    \dcut(W , W_\mathcal{P})
    & \leq \dcut(W , U) + \dcut(U , U_\mathcal{P}) + \dcut(U_\mathcal{P}
      , W_\mathcal{P}) \nonumber \\
    &\leq  2\eps + \dcut(U , U_\mathcal{P})  .
      \label{eq_ineg_wrl_eps_2}
\end{align}

Hence, to get the weak regularity property for $\cM$-valued kernels,
we are left to prove it for the much smaller set of 
$\mathcal{V}$-valued kernels, 
where $\mathcal{V}$ is the convex hull of $\{\mu_1, \ldots, \mu_n\}$.

\medskip

As $\dmeas$ is quasi-convex and sequentially continuous \wrt the weak topology,
using \Cref{lemma_unif_cont_TM}, there exists $\eta>0$ such that
for all $\mu,\nu\in \MeasEps$, we have that $\TM{\mu - \nu} < \eta$ 
implies that $\dmeas(\mu,\nu) \leq \eps$.

As $\mathcal{V}$  is a subset of a vector  space with  finite dimension $n$,  
the norm $\TM{\cdot}$  seen  over  $\mathcal{V}$  is  equivalent  to  the  $L_1$-norm
$\mu  =   \sum_{i=1}^n  \alpha_i   \mu_i  \mapsto   \norm{\alpha}_1  =
\sum_{i=1}^n \vert \alpha_i\vert$.  We can now see $\mathcal{V}$-valued
kernel  as $\R^n$-valued  graphon  with  a cut  norm  derived from  the
$L_1$-norm  $\norm{\cdot}_1$,   and  in  this  case
the proof for the weak regularity Lemma~9.9 in~\cite{Lovasz}
can easily be adapted. 
Hence,  we have  the weak  regularity property  for $\mathcal{V}$-valued
kernels:    
there   exists    $m\in\N^*$,   such    that   for    every
$\mathcal{V}$-valued   kernel   $U'$,   and for every finite partition $\cQ$ of $[0,1]$
 there   exists   a   finite partition
$\mathcal{P}$  of  $[0,1]$  that  refines $\mathcal{Q}$,
and such that $\vert\cP\vert \leq m \vert\cQ\vert$
and $\sup_{S,T\subset [0,1]}\TM{(U' - U'_{\mathcal{P}})(S\times T;\cdot)} < \eta$,
and thus $\dcut(U', U'_{\mathcal{P}}) \leq \eps$.

Taking    $U'=U$ in \eqref{eq_ineg_wrl_eps_2}, we get   that
$ \dcut(W  , W_\mathcal{P})\leq 3\varepsilon$  and $|\cp|\leq  m |\mathcal{Q}|$. This
concludes the proof of the lemma when $\Space$ is compact.
\medskip

\textbf{Step 2.} We consider the general case $\Space $  Polish. 
We  now prove  that  $\dcut$  is weakly  regular  on $\KernelEps$.   
  Let  $\ck  \subset \KernelEps$  be  a subset  of
  $\MeasEps$-valued  kernels that  is  tight and uniformly bounded,
  and denote by $C=\sup_{W\in\ck} \TM{W}<\infty$.
  
  Let $\eps > 0$. 
 As $\dmeas$ is quasi-convex and sequentially continuous \wrt the weak topology,
using \Cref{lemma_unif_cont_TM}, there exists $\eta>0$ such that
for all $\mu,\nu\in \MeasEps$, we have that $\TM{\mu - \nu} < \eta$ 
implies that $\dmeas(\mu,\nu) < \eps$.
Without loss of generality, we assume that $\eta \leq \eps$.
Let $\eta_C = \min(\eta, \eta / C)$.

  As  $\ck$ is  tight,
  using  Proposition~\ref{prop:tight},   there  exists  a   compact  set
  $K \subset \Space$, such that  for every $W\in \mathcal{K}$ the subset
  $A_W  = \{  (x,y)\in [0,1]^2  : \vert W\vert(x,y;  K^c) \leq  \eta_C / 2 \}$  has Lebesgue
  measure at  least $1-\eta_C / 2$.  
  Let $W\in  \ck$, and  define the
  signed measure-valued kernel  $U$ by:  $U(x,y;\cdot) = W(x,y;\cdot  \cap K)$
  for  every $(x,y)\in  A_W$,  and $U(x,y;\cdot)  =  0$
  otherwise. 
Let $S,T \subset [0,1]$. We have:
\begin{align*}
\TM{(W-U)(S\times T; \cdot)}
& \leq \int_{S\times T}\TM{W(x,y;\cdot)-U(x,y;\cdot)} \ \drv x\drv y \\
& \leq   \int_{A_W \cap (S\times T)} \vert W\vert(x,y;K^c) \ \drv x\drv y
	+ \int_{A_W^c \cap (S\times T)} \TM{W(x,y;\cdot)} \ \drv x\drv y \\
& \leq \eta_C / 2 + C \cdot \eta_C / 2 \\ 
& \leq \eta. 
\end{align*}
Thus, we have that $\dmeas(W(S\times T; \cdot), U(S\times T; \cdot)) < \eps$.
Since this holds for all $S,T \subset [0,1]$, we get that
	$\dcut(W,U) \leq \eps$.

Notice that the $\SignedMeas$-valued kernel  $U$ is also  a
$\SignedMeasK$-valued kernel, where $K\subset \Space$ is a compact set,
and that $\TM{U} \leq \TM{W} \leq C$.
Further remark that, using Lemma~\ref{contractivity_2},
	for every $W\in\ck$ and every finite partition $\cP$ of $[0,1]$,
	we have that:
\begin{align*}
\dcut(W , W_\cP)
& \leq \dcut(W , U) + \dcut(U , U_\cP) + \dcut(U_\cP , W_\cP) \\
& \leq 2 \eps + \dcut(U , U_\cP) .
\end{align*}
Hence, to get the weak regularity property for $\dcut$ on $\ck$ 
	(see Definition~\ref{defi:extra-prop}~\ref{hypo_dist_WRL}),
	it is enough to prove that $\dcut$ restricted to 
	$\MeasEpsK$-valued kernels is weakly regular,
	which is true by Step~1.
As a consequence, we get that $\dcut$ on $\KernelEps$
	is weakly regular.
\end{proof}

\subsection{A stronger weak regularity lemma for \texorpdfstring{$\dcutF$}{}}\label{section_stronger_weak_regularity_lemma}

In this subsection, we prove a stronger version of the weak regularity lemma
for the special case of the cut distance $\dcutF$.
We shall use this result for the proof of the second sampling Lemma~\ref{SecondSamplingLemma}.

Let $\F = (f_n )_{n\in\N}$, with $f_0=\un$ and $f_n$  takes values  in  $[0, 1]$, 
be a convergence determining
sequence, which is assumed fixed in  this section.

\subsubsection{Comparison between  \texorpdfstring{$\NcutFSymbol$}{} and an euclidian norm}

To better understand the stepping operator,
we introduce a scalar product over signed measure-valued kernels.
The link between this scalar product and the norm $\NcutFSymbol$
is given by Lemma~\ref{lemma_ineg_norm_cut_two}.
We define the scalar product $\langle \cdot, \cdot\rangle_{\F}$ on signed measure-valued kernels
for $U,W\in \Kernel$ by:
\begin{equation*}
\langle U,W \rangle_{\F} = \sum_{n\geq 0} 2^{-n} \langle U[f_n], W[f_n] \rangle  , 
\end{equation*}
where for all $n$ the scalar product taken for $U[f_n]$ and $W[f_n]$
is the usual scalar product in $L^2([0,1]^2, \lambda_2)$ for real-valued kernels:
\begin{equation*}
\langle U[f_n], W[f_n] \rangle 
= \int_{[0,1]^2} U[f_n](x,y) W[f_n](x,y) \ \drv x \drv y.
\end{equation*}
The scalar product $\langle \cdot, \cdot\rangle_{\F}$ 
induces a norm on $\Kernel$ which we denote by $\norm{\cdot}_{2,\F}$.

Let $\cp$ be a finite partition of $[0, 1]$. As the stepping operator for 
measurable real-valued $L^2$ functions on $[0,1]^2$
is a linear projection, and is idempotent and symmetric,
and by definition of the scalar product $\langle \cdot, \cdot\rangle_{\F}$
for signed measure-valued kernels,
we have that the stepping operator for signed measure-valued kernels
is linear, idempotent and symmetric for $\langle \cdot, \cdot\rangle_{\F}$.
Moreover, the stepping operator is the orthogonal projection for $\langle \cdot, \cdot\rangle_{\F}$
onto the space of stepfunctions with steps in $\mathcal{P}$.

Note that for a probability-graphon $W\in \Graphon$,
we have $\norm{W}_{2,\F}  \leq \sqrt{2}$ as  each $f_n$ takes values in $[0,1]$.
The following technical lemma gives a comparison between $\NcutFSymbol$ and $\norm{\cdot}_{2,\F}$. 

\begin{lemme}[Comparison between $\NcutFSymbol$ and $\norm{\cdot}_{2,\F}$]
	\label{lemma_ineg_norm_cut_two}
For a signed measure-valued kernel $W\in \Kernel$,
we have $\norm{W}_{\square,\F} \leq \sqrt{2} \norm{W}_{2,\F}$.
\end{lemme}

\begin{proof}
Let $S,T \subset [0,1]$ be measurable subsets.
By the Cauchy-Schwarz inequality, we have
$\vert \langle W[f_n], \mathds{1}_{S\times T} \rangle \vert^2
\leq \norm{W[f_n]}_2^2 = \langle W[f_n], W[f_n] \rangle $
for every $n\geq 0$.
Using this inequality along with Jensen's inequality, 
	we get for every $S,T\subset [0,1]$ that:
\begin{align*} 
\left( \sum_{n\geq 0} 2^{-n} \vert  W(S\times T,f_n) \vert \right)^2
& = 
 \left( \sum_{n\geq 0} 2^{-n} \vert \langle W[f_n], \mathds{1}_{S\times T} \rangle \vert \right)^2 \\
& \leq \sum_{n\geq 0} 2^{-n+1} \vert \langle W[f_n], \mathds{1}_{S\times T} \rangle \vert^2 \\
& \leq \sum_{n\geq 0} 2^{-n+1} \langle W[f_n], W[f_n] \rangle  \\
& = 2 (\norm{W}_{2,\F})^2		.
\end{align*}
Taking the supremum over every measurable subsets $S,T \subset [0,1]$
gives the desired inequality.
\end{proof}

\subsubsection{The weak regularity lemma for \texorpdfstring{$\NcutFSymbol$}{}}
\label{subsection_WRL_dcutF}

The following lemma gives an explicit bound on the approximation of a signed measure-valued
kernel, say $W$, 
by its steppings $W_\cp$, with $\cp$ a finite partition on $[0, 1]$.
Its proof  is a straightforward adaptation of the proof of the weak regularity lemma 
for real-valued graphons in \cite[Lemma~9.9]{Lovasz}.

\begin{lemme}[Weak regularity lemma for $\NcutFSymbol$, simple formulation]\label{Regularity_Lemma}
For every signed measure-valued kernel $W\in \Kernel$ and $k\geq 1$,
there exists a finite partition $\mathcal{P}$ of $[0,1]$
such that $|\cp|=k$ and:
\[
  \norm{W - W_\mathcal{P}}_{\square,\F} 
\leq \frac{\sqrt{8}}{\sqrt{\log(k)}} \norm{W}_{2,\F}
 . \]

In particular, if $W\in \Graphon$ is a probability-graphon, 
(as $\norm{W}_{2,\F}\leq \sqrt{2}$) we have:
\[
  \norm{W - W_\mathcal{P}}_{\square,\F} 
\leq \frac{4}{\sqrt{\log(k)}}
 \cdot \]
\end{lemme}

It is possible in the weak regularity lemma to ask for extra requirements,
for instance to start from an already existing partition,
or to ask the partition to be balanced,
as stated in the following lemma.
The proof is a straightforward adaptation of the proof of \cite[Lemma~9.15]{Lovasz}.

\begin{lemme}[Weak regularity lemma for $\NcutFSymbol$, with extra requirements]\label{Regularity_Lemma_2}
Let $W\in\Graphon$ be a probability-graphon, and let $1\leq m < k$.
\begin{enumerate}[label=(\roman*)]
\item\label{Regularity_Lemma_2a} For every partition $\mathcal{Q}$ of $[0,1]$ into $m$ classes,
there is a partition $\mathcal{P}$ with $k$ classes refining $\mathcal{Q}$
and such that:
\[  \norm{W - W_\mathcal{P}}_{\square,\F} \leq \frac{4}{\sqrt{\log (k/m)}}
\cdot \]
\item\label{Regularity_Lemma_2b} For every partition $\mathcal{Q}$ of $[0,1]$ into $m$ classes,
there is an equipartition (\ie a finite partition into classes with the same measure)
$\mathcal{P}$ of $[0,1]$ into $k$ classes and such that:
\[  \norm{W - W_\mathcal{P}}_{\square,\F} \leq 2 \norm{W - W_\mathcal{Q}}_{\square,\F} + \frac{2m}{k}
\cdot \]
\end{enumerate}
\end{lemme}

\section{Compactness and completeness of \texorpdfstring{$\Graphon$}{}}
\label{section_tension}

In Section~\ref{subsec:tightness},  we link the tightness  criterion for
measure-valued kernels  with the relative compactness \wrt  the cut distance
$\ddcut$.  In Section~\ref{subsec:equiv_topo}, we compare the topologies
induced  by  the cut  distance  $\ddcut$  for  different choice  of  the
distance $\dmeas$,  and state  that under  some conditions  on $\dmeas$,
those  topologies  coincide.  In  Section~\ref{subsec:completeness},  we
investigate the completeness of $\Graphon$ endowed with the cut distance
$\ddcut$ and prove that the space of probability-graphons $\UGraphon$ is
a Polish space (Theorem~\ref{theo_complete}), and  that it is compact if
and only  if $\Space$  is compact  (Corollary~\ref{cor:W1compact}).  The
technical        proofs        are         postponed        to        
Section~\ref{section_proof_theo_tension_conv}.

\subsection{Tightness criterion and compactness}
	\label{subsec:tightness}

Let $\cm \subset \SignedMeas$ be  a subset of signed measures on
$\Space$.   Recall that  $\cW_\cm\subset\Kernel$  denote  the subset  of
signed measure-valued  kernels which are $\cm$-valued.   In this section,
we shall denote by $\widetilde{\mathcal{W}}_\cm$ the quotient of
$\mathcal{W}_\cm$ identifying  signed measure-valued kernels  that are
weakly isomorphic.

Remind from Definition~\ref{def:ddcut} and Theorem~\ref{theo:Wm=W}
that for an invariant, smooth and weakly regular distance $d$ 
on $\Graphon$ (resp. $\Kernelp,\Kernel$),
$\dd$ is defined as $\dd(U,W) = \inf_{\varphi \in \InvRelabel} d(U, W^\varphi)$,
and is a distance on $\UGraphon$ (resp. $\UKernelp,\UKernel$).

We are now ready to formulate the important following theorem,
which relates tightness with compactness and convergence
for signed measure-valued kernels.
We prove this theorem in Section~\ref{section_proof_theo_tension_conv}.

  \begin{theo}[Compactness theorem for $\UGraphon$]
    \label{theo_tension_conv}
    Let  $d$ be  an invariant,  smooth  and weakly  regular distance  on
    $\Graphon$ (resp. $\Kernel$). 
    \begin{enumerate}[label=(\roman*)]
    \item\label{it:tension-cv} If a sequence of elements of
      $\Graphon$ or $\UGraphon$ (resp. $\Kernel$ or $\UKernel$)
      is  tight (resp.  tight  and  uniformly bounded),  then  it has  a
        subsequence converging for $\dd$.
   \item\label{it:M-compact}  If  $\cm\subset \Proba$  (resp. $\cm\subset  \SignedMeas$) 
      is convex and compact (resp. sequentially compact)
      for   the   weak   topology,   then   the   space
      $(\widetilde{\mathcal{W}}_\cm, \dd)$ is convex and
      compact. 
    \item\label{it:WG=compact}  If $\Space$  is compact, then the space  $(\UGraphon, \dd)$ is
      compact.
      
 \end{enumerate}
\end{theo}

We deduce from this theorem a characterization of relative compactness
for subsets of probability-graphons.

\begin{prop}[Characterization of relative compactness]\label{prop_equiv_tight_compact}
  Let  $\dmeas$  be   a  distance  on  $\SubProba$   (resp.  $\Meas$  or
  $\SignedMeas$) that  induces the weak topology on  $\SubProba$ (resp.
  $\Meas$).   Assume   that   the   distance   $\dcut$   on   $\Graphon$
  (resp.  $\Kernelp$  or $\Kernel$)  is  (invariant)  smooth and  weakly
  regular.

    \begin{enumerate}[label=(\roman*)]
    \item\label{it:cv-tight-recip}
If a sequence of elements of
    $\Graphon$ or $\UGraphon$ (resp. $\Kernelp$ or $\UKernelp$) is
    converging for $\ddcut$, then it is tight.
    \item     \label{it:compact}
  Let $\ck$ be a subset of $\UGraphon$ (resp. a uniformly bounded
  subset of $\UKernelp$).  Then, the set $\ck$ is relatively compact for
  $\ddcut$ if and only if it is tight.
\item \label{it:cvx-closed}  Let $\cm$ be  a subset of $\Meas$  which is
   bounded, convex   and   closed   for   the  weak   topology.   
  Then   the   set
  $\widetilde{\cW} _\cm$ is convex and closed in $\UKernelp$.
 \end{enumerate}
\end{prop}

Remark that convergence for $\ddcut$ does not necessarily imply tightness on $\Kernel$ or on $\UKernel$.

\begin{proof}
   We consider the case where $\dmeas$ is a distance on $\Meas$ or $\SignedMeas$,
   	the case with $\SubProba$ is similar.

  We prove  Point~\ref{it:cv-tight-recip}. Let  $(W_n )_{n\in\N}$  be a
  convergent  sequence  of  $\Kernelp$  (and thus  of  $\UKernelp$)  for
  $\ddcut$. We  deduce from the continuity  of the map $W  \mapsto M_W$,
  see    Lemma~\ref{conv_graphon_conv_measure},   that    the   sequence
  $(M_{W_n} )_{n\in\N}$ is converging for $\dmeas$, and thus is tight
  as $\dmeas$ induces the weak  topology on $\Meas$. Then, by definition
  the sequence $(W_n )_{n\in\N}$ is tight.

  \medskip

  We prove  Point~\ref{it:compact}.  If $\ck\subset \UKernelp$  is tight
  and           uniformly           bounded,           then           by
  Theorem~\ref{theo_tension_conv}~\ref{it:tension-cv}
  every sequence in $\ck$ has a subsequence converging for $\ddcut$,
  which implies that $\ck$ is relatively compact in the metric space
  	$(\UKernelp,\ddcut)$ (see Remark~\ref{rem_metrizable_topo_seq}).

  Conversely, assume  that $\ck\subset  \UKernelp$ is  uniformly bounded
  and      relatively      compact       for      $\ddcut$.       Define
  $\cm  =   \{  M_W  :  W   \in  \mathcal{K}  \}  \subset   \Meas$.   By
  Lemma~\ref{conv_graphon_conv_measure}, the mapping  $W \mapsto M_W$ is
  continuous from $(\UGraphon, \ddcut)$ to $(\Meas, \dmeas)$.  Hence, as
  $\dmeas$ induces the weak topology on $\Meas$,  the set $\cm$ is also relatively
  compact in  $\Meas$ for the weak  topology.  As the space  $\Space$ is
  Polish,  applying Lemma~\ref{lemme_Bogachev_Prohorov_theorem},  
  we  get  that  $\cm \subset \Meas$  is  tight,  and  by
  Definition~\ref{defi:tight}, the set $\ck \subset \UKernelp$ is tight.

\medskip
We postpone the proof of Point~\ref{it:cvx-closed}
	to Section~\ref{section_proof_theo_tension_conv} on page \pageref{page_proof_point_iii}.
\end{proof}

\subsection{Equivalence of topologies induced by \texorpdfstring{$\ddcut$}{}}
	\label{subsec:equiv_topo}
The following lemma allows to show a first result
on equivalence of the topologies induced by the cut distance $\ddcut$
for different distances $\dmeas$, where the sub-script $\mathrm{m}$ is used to
distinguish different distances.  
Its proof is given below.
Remind from \Cref{theo:Wm=W} that $\dcut$ must be smooth for $\ddcut$ to be a distance.

\begin{lemme}[Comparison of topologies induced by $\dcut$ and $\ddcut$]
  \label{lem:unif-cont}
  Let $\dmeas$  and $\dmeasPrime$ be  two distances on  $\SubProba$ such
  that $\dmeasPrime$ is uniformly continuous \wrt $\dmeas$
  (in particular, $\dmeas$ induces a finer topology than $\dmeasPrime$ on $\SubProba$).
    Then, we
  have the following properties. 
\begin{enumerate}[label=(\roman*)]
\item\label{it:dcut-cont}
The distance $\dcutPrime$ is uniformly continuous \wrt $\dcut$ on
$\Graphon$. 
In particular $\dcut$ induces a finer topology than $\dcutPrime$
on $\Graphon$. 
\item\label{it:dcut-smooth}  If the  distance $\dcut$  on $\Graphon$ is
  smooth,  then  the  distance  $\dcutPrime$   is  also  smooth and
  $\ddcutPrime$ is uniformly continuous \wrt $\ddcut$.
  In particular,   $\ddcut$ induces  a finer  topology 
  	than  $\ddcutPrime$ on $\UGraphon$.

\item\label{it:dcut-weak-reg}
If the distance $\dcut$ on $\Graphon$ is
weakly regular, then the distance 
$\dcutPrime$ is also weakly regular.

\item\label{it:ddcut-topo}
  Assume that the distance $\dmeasPrime$ 
	induces the  weak topology on $\SubProba$, and
	that the distance $\dcut$ is smooth and weakly regular.
  In particular, the distance $\dmeas$ also induces the weak topology on $\SubProba$.
Then, the distances $\ddcut$ and $\ddcutPrime$ induce the same topology on $\UGraphon$.
\end{enumerate}
\end{lemme}

We will see some application of Lemma~\ref{lem:unif-cont}
	in Corollary~\ref{cor:equiv-topo} below.
	
\begin{rem}[Extension to $\Kernel$ for topology comparisons]
  \label{rem:extension-pm}
  In Lemma~\ref{lem:unif-cont}~\ref{it:dcut-cont}-\ref{it:dcut-weak-reg}, one can replace
$\Graphon$ and $\UGraphon$ by $\Kernelp$ and $\UKernelp$ or by
 $\Kernel$ and $\UKernel$
as soon as the distances $\dmeas$ and $\dmeasPrime$ are defined on 
$\Meas$ or  $\SignedMeas$;
in this case comparisons of topologies only apply
on uniformly bounded subsets.
In Lemma~\ref{lem:unif-cont}~\ref{it:ddcut-topo},
one can replace $\UGraphon$ by $\UKernelM$
with a bounded subset $\cm \subset \Meas$
as soon as the distances $\dmeas$ and $\dmeasPrime$ are defined on 
$\Meas$.
\end{rem}
 
\begin{proof}[Proof of Lemma~\ref{lem:unif-cont}]
  We prove Point~\ref{it:dcut-cont}.  Let  $\eps > 0$.  As $\dmeasPrime$
  is uniformly continuous \wrt $\dmeas$, there exists $\eta > 0$ such
  that for every $\mu, \nu\in \SubProba$, if $\dmeas(\mu,\nu) < \eta$,
  then $\dmeasPrime(\mu,\nu)  < \eps$.  Let $U,W\in  \Graphon$ such that
  $\dcut(U,W)  <  \eta$.   Then, for every
  subsets $S,T \subset [0,1]$, we have:
\[
  \dmeasPrime\left(  U(S\times T;\cdot), 
    W(S\times T;\cdot) \right)
  < \eps .
\]
Thus, $\dcutPrime(U,W) \leq \eps$.
Hence, $\dcutPrime$ is uniformly continuous \wrt $\dcut$.

\medskip

We prove  Point~\ref{it:dcut-smooth}.  Assume  that $\dcut$  is smooth.
Let  $(W_n )_{n\in\N}$  and $W$  be probability-graphons such  that
$W_n(x,y;\cdot)$  weakly converges  to $W(x,y;\cdot)$  for almost  every
$x,y\in [0,1]$.  Since  the cut distance $\dcut$ is smooth,  we get that
$\dcut(W_n, W)  \to 0$.  As  $\dcutPrime$ is uniformly  continuous 
(and     thus    also     continuous)  \wrt  $\dcut$,    we     have    that
$\dcutPrime(W_n, W) \to 0$.  Hence, $\dcutPrime$ is smooth.

Furthermore, let $\eps > 0$.
Let $\eta > 0$ be such that for every $\mu,\nu\in\SubProba$,
	$\dmeas(\mu,\nu) < \eta$ implies $\dmeasPrime(\mu,\nu) < \eps$.
For every $U,W\in\Graphon$ such that $\ddcut(U,W) < \eta$,
	there exists $\varphi\in\InvRelabel$ such that $\dcut(U, W^\varphi) < \eta$,
	which implies that  $\dcutPrime(U, W^\varphi) < \eps$,
	which then implies that $\ddcutPrime(U,W) < \eps$.
That is, $\ddcutPrime$ is uniformly continuous \wrt $\ddcut$.
\medskip

We prove  Point~\ref{it:dcut-weak-reg}.  Assume that $\dcut$  is weakly
regular.  Let $\mathcal{K}\subset  \Graphon$ be tight.  Let  $\eps > 0$.
As  $\dcutPrime$ is  uniformly  continuous \wrt  $\dcut$, there  exists
$\eta   >   0$   such   that   for   every   $U,W\in   \Graphon$,   if
$\dcut(U,W) < \eta$, then $\dcutPrime(U,W) < \eps$.  Since $\dcut$ is
weakly   regular,  there   exists  $m\in\N^*$,   such  that   for  every
probability-graphon  $W\in\mathcal{K}$,   and  for  every  finite partition
$\mathcal{Q}$  of $[0,1]$,  there  exists a  finite partition $\mathcal{P}$  of
$[0,1]$ that refines $\mathcal{Q}$ and such that
$\vert \cP \vert \leq m \vert\cQ\vert$ and
$\dcut(W,   W_{\mathcal{P}})  <   \eta$;   and  thus   we  also   have
$\dcutPrime(W, W_{\mathcal{P}}) < \eps$.   Hence, $\dcutPrime$ is weakly
regular.

\medskip

We prove  Point~\ref{it:ddcut-topo}.  Assume that  $\dmeasPrime$ induces
the weak topology on $\SubProba$ and that $\dcut$ is smooth and weakly regular.
In particular, the topology induced by $\dmeas$ if finer than the topology
	induced by $\dmeasPrime$, \ie finer than the weak topology.
As $\dcut$ is smooth, by Lemma~\ref{lem:dm-smooth-d-cont},
	$\dmeas$ is continuous \wrt the weak topology
	(\ie the weak topology if finer than the topology induced by $\dmeas$),
	and thus $\dmeas$ induces the weak topology on $\SubProba$.
By  Points~\ref{it:dcut-smooth}
and~\ref{it:dcut-weak-reg}, we  get that $\dcutPrime$ is  also smooth
and weakly regular. By Point~\ref{it:dcut-smooth}, the distance $\ddcut$
induces a finer topology than $\ddcutPrime$ on $\UGraphon$.

We  now prove  that  the topology  of $\ddcutPrime$  is  finer than  the
topology of  $\ddcut$. Let  $(W_n )_{n\in\N}$  and $W$  be probability-graphons  
in  $\UGraphon$,   such  that  $W_n$  converges   to  $W$  for
$\ddcutPrime$.  By
Proposition~\ref{prop_equiv_tight_compact}~\ref{it:cv-tight-recip}, we
deduce that  the sequence $(W_n)_{n\in \N}$ is tight.  
As $\dcut$ is smooth  and weakly regular,
Theorem~\ref{theo_tension_conv}    gives    that    every    subsequence
$(W_{n_k}  )_{k\in\N}$  of  the sequence $(W_n  )_{n\in\N}$ 
has  a  further subsequence $(W_{n'_k}  )_{k\in\N}$  that converges   for
$\ddcut$ to a limit, say  $U\in\UGraphon$.  Since  $\ddcut$ is finer  than $\ddcutPrime$, we  deduce that
$(W_{n'_k}  )_{k\in\N}$ converges also  to $U$ for $\ddcutPrime$;  but, as
a  subsequence, it   also  converges  to  $W$  for  $\ddcutPrime$.   As
$\ddcutPrime$    is    a    distance    on    $\Graphon$    thanks    to
Theorem~\ref{theo:Wm=W},  we get  $U=W$.   Hence,  every subsequence  of
$(W_n  )_{n\in\N}$  has a  further subsequence that  converges  to $W$  for
$\ddcut$, therefore the whole sequence itself converges to  $W$ for $\ddcut$.
Consequently,  $\ddcutPrime$ is  finer  than $\ddcut$, and thus those  two
distances induce the same topology on $\UGraphon$.
\end{proof}

The following theorem states that  the topology induced by $\ddcut$ does
not depend on $\dmeas$ under some hypothesis.  We prove this theorem
in  Section~\ref{section_proof_theo_tension_conv}.   Recall  that  under
suitable  conditions satisfied  in the  next theorem,  the  quotient
space  $\UGraphon$  does  not  depend  on the  choice  of  the  distance
$\dmeas$, see Theorem~\ref{theo:Wm=W}.

\begin{theo}[Equivalence of topologies induced by $\ddcut$ on $\UGraphon$]\label{theo_equiv_topo}
  The topology on  the space probability-graphon $\UGraphon$ induced by
  the distance $\ddcut$  does not depend on the choice  of the distance
  $\dmeas$ on $\SubProba$, as long
  as $\dmeas$ induces the weak  topology on $\SubProba$
  and the cut distance $\dcut$ on $\Graphon$
  is (invariant) smooth,  weakly regular and regular \wrt
  the stepping operator.
\end{theo}

Remind from \Cref{prop:weak_regular_2} that when the distance
$\dmeas$ is quasi-convex and continuous \wrt the weak topology on $\Meas$ or $\SignedMeas$,
then the cut distance $\dcut$ is invariant, smooth,  weakly regular and regular \wrt
  the stepping operator.
This is in particular the case of $\dmeasP$, $\NmeasFSymbol$, $\NmeasKRSymbol$
	and $\NmeasFMSymbol$.

The next corollary is an immediate consequence of Lemma~\ref{lemma_comp_Prohorov_KR_FM},
Corollary~\ref{cor:reg-dist-usuel}, Lemma~\ref{lem:unif-cont} and Theorem~\ref{theo_equiv_topo}.
This corollary gathers results comparing the topology induced
	by the cut distances associated with the distances 
	introduced in Section~\ref{section_examples_distance}.
It is yet unclear if the  distances $\dcutF$ induces the same topology
on the space of labeled probability-graphons $\Graphon$ 
as the one induced by $\dcutP$, $\dcutFM$ or $\dcutKR$.

\begin{cor}[Topological equivalence of the cut distances associated to
	$\dmeasP$, $\NmeasFMSymbol$, $\NmeasKRSymbol$ and $\NmeasFSymbol$]
	\label{cor:equiv-topo}
The cut distances $\dcutP$ on $\Kernelp$ and $\dcutKR, \dcutFM$ and $\dcutF$ on $\Kernel$
	are invariant, smooth, weakly regular and regular \wrt the stepping operator.
Moreover, we have the following comparison between the distances
introduced in Section~\ref{section_examples_distance}.
\begin{enumerate}[label=(\roman*)]
\item\label{item:equiv-topo-norm}
The cut norms $\NcutFMSymbol$ and $\NcutKRSymbol$
(resp. the cut distances $\ddcutFM$ and $\ddcutKR$)
are metrically equivalent on $\Kernel$ (resp. $\UKernel$).

\item\label{item:equiv-topo-dP}
The cut distances $\ddcutFM$, $\ddcutKR$ and $\ddcutP$
(resp. $\dcutFM$, $\dcutKR$ and $\dcutP$)
are uniformly continuous \wrt one another,
and thus induce the same topology on $\UGraphon$ (resp. $\Graphon$)
and on every uniformly bounded subset of $\UKernelp$ (resp. $\Kernelp$).

\item\label{item:equiv-topo-dF}
The cut distances  $\ddcutFM$, $\ddcutKR$, $\ddcutP$ and $\ddcutF$, for every
 choice of the convergence determining sequence $\F$, induce the same
  topology  on  $\UGraphon$. 
\end{enumerate}
\end{cor}

\begin{proof}
The first part of the corollary is a re-statement of \Cref{cor:reg-dist-usuel}.
  Point~\ref{item:equiv-topo-norm} is an immediate consequence of~\eqref{eq:met-equiv}.
  
  We now prove Point~\ref{item:equiv-topo-dP}.
  Thanks to~\eqref{eq:met-equiv} and Point~\ref{item:equiv-topo-norm}, it is  enough to consider only 
  the Prohorov and the Kantorovitch-Rubinshtein distances.  As $\dmeasP$ is uniformly
  continuous  \wrt  $\dmeasKR$ (see Lemma~\ref{lemma_comp_Prohorov_KR_FM}),  
  applying Lemma~\ref{lem:unif-cont}  (remind Corollary~\ref{cor:reg-dist-usuel})
  	with Remark~\ref{rem:extension-pm} in mind,
  we  get that  $\ddcutP$ (resp. $\dcutP$)  
  	is  uniformly continuous \wrt  $\ddcutKR$ (resp. $\dcutKR$)
	on every uniformly bounded subset of $\UKernelp$ (resp. $\Kernelp$)
  As $\dmeasKR$ is  also uniformly continuous \wrt $\dmeasP$ 
  	(see Lemma~\ref{lemma_comp_Prohorov_KR_FM}),
	applying again Lemma~\ref{lem:unif-cont}, 
   we have that $\ddcutKR$ (resp. $\dcutKR$)  
  	is  uniformly continuous \wrt  $\ddcutP$ (resp. $\dcutP$)
	on every uniformly bounded subset of $\UKernelp$ (resp. $\Kernelp$).
  
  Point~\ref{item:equiv-topo-dF} is an immediate consequence of
  Corollary~\ref{cor:reg-dist-usuel} and Theorem~\ref{theo_equiv_topo},
  together with Point~\ref{item:equiv-topo-dP}.
\end{proof}

\begin{rem}[Extension to uniformly bounded subsets of $\UKernelp$]
  \label{rem:extension-pm-2}
  In Theorem~\ref{theo_equiv_topo} and also in Corollary~\ref{cor:equiv-topo}~\ref{item:equiv-topo-dF},
  one can replace $\UGraphon$ 
  by $\UKernelM$ with a bounded subset $\cm \subset \Meas$
  as soon as the distance $\dmeas$ is defined on $\Meas$.
  (One has in mind the case $\cM=\SubProba$.) This can be seen by an easy modification in the proof of Theorem~\ref{theo_equiv_topo}.
  Alternatively, this can be seen
  using scaling to reduce the case of general $\cm$
  to the case of $\SubProba$, and then adding a cemetery point (for missing mass of measures)
  to $\Space$ to further reduce to the case of $\Proba$.
\end{rem}

\subsection{Completeness}
	\label{subsec:completeness}

Let $\dmeas$ be a distance on $\SubProba$ or $\Meas$. 
We shall consider a slight modification of the cut
distances $\dcut$ and $\ddcut$ to  achieve completeness. 
Recall the measure $M_W\in \Meas$  defined by~\eqref{eq:def-MW}
associated to $W\in\Kernelp$.

\begin{defin}[The cut distances $\dcutC$ and $\ddcutC$]\label{ref_dist_complete}
Let $\dmeas$ and  $\dmeasC$ be two distances on $\MeasEps$ with $\epsilon\in\{\leq 1, +\}$.
We define the cut distance $\dcutC$ on the space of $\MeasEps$-valued kernels $\KernelEps$ as:
\begin{equation*}
\dcutC(U,W) = \dcut(U,W) + \dmeasC(M_U,M_W)  ,
\end{equation*}
and the cut (pseudo-)distance $\ddcutC$  on the space of 
	unlabeled $\MeasEps$-valued kernels $\UKernelEps$ as:
\begin{equation*}
\ddcutC(U,W) = \inf_{\varphi \in \InvRelabel} \dcutC(U,W^\varphi)	= 
\ddcut(U,W) + \dmeasC(M_U,M_W)		 .
\end{equation*}
\end{defin}
Notice that by Lemma~\ref{lemma_dcut_invariant}
and the definition of $M_W$, the distance $\dcutC$ is invariant.

\begin{lemme}[Topological equivalence of $\ddcut$ and $\ddcutC$]
\label{lem:ddC-metrique}
Let $\dmeas$  and $\dmeasC$  be two distances  on $\MeasEps$, with $\epsilon\in\{\leq 1, +\}$,
such that $\dmeasC$  is    continuous  \wrt $\dmeas$  
and  that $\dcut$  is
(invariant and) smooth  on $\KernelEps$.  Then, the  cut distance $\dcutC$
is  invariant and  smooth  and  $\ddcutC$  is  a  distance  on  $\UKernelEps$.
Moreover, the distances $\dcut$ and $\dcutC$ (resp. $\ddcut$ and  $\ddcutC$) induce the same topology
on the space $\KernelEps$ (resp. $\UKernelEps$).
\end{lemme}

\begin{proof}
  Let $(W_n  )_{n\in\N}$ and  $W$ be elements  of $\SubGraphon$  such that
  $(W_n(x,y;\cdot)  )_{n\in\N}$  weakly converges  to $W(x,y;\cdot)$  for
  almost every $x,y\in[0,1]$.  Since the  distance $\dcut$ is smooth, we
  have that $\lim_{n\rightarrow \infty } \dcut(W_n,W)= 0$.  
  Using
  Lemma~\ref{conv_graphon_conv_measure}  on the  continuity  of the  map
  $W \mapsto  M_W$ and  that $\dmeasC$  is   continuous \wrt
  $\dmeas$,                we                 obtain                that
  $\lim_{n\rightarrow \infty  } \dcutC(W_n,W)=  0$. This gives  that the
  distance  $\dcutC$  is  smooth.  Since  we  have  already  seen  that $\dcutC$  is
  invariant, we deduce from  Theorem~\ref{theo:Wm=W} that $\ddcutC$ is a
  distance on $\UGraphon$.

  \medskip

  We now  prove that  the two  distances $\dcut$ and $\ddcutC$ induce  the same  topology
  (which implies that this is also true for $\ddcut$ and $\ddcutC$).
  As  $\dcut \leq  \dcutC$, convergence for $\dcutC$  implies convergence
  for $\dcut$.   Conversely, let $(W_n )_{n\in\N}$ be a  sequence in $\KernelEps$  
  that converges for  $\dcut$ to a  limit, say $W\in\KernelEps$.   
  Using   again  Lemma~\ref{conv_graphon_conv_measure}   and  the
     continuity     of     $\dmeasC$ \wrt $\dmeas$,    we     obtain     that
  $\lim_{n\rightarrow \infty} \dmeasC(M_{W_n},M_W)= 0$. 
  This  clearly implies
  that the sequence $(W_n )_{n\in\N}$ converges to $W$ for $\dcutC$.
  Then, the two distances have the same convergent sequences
  	and thus induce the same topology (see Remark~\ref{rem_metrizable_topo_seq}).
\end{proof}

Recall $\Space$ is a Polish space.
We already proved in Proposition~\ref{theo_separable}
that the space $(\UGraphon, \ddcut)$ is separable;
and we now investigate completeness of this space.

\begin{theo}[$\UGraphon$ is a Polish space]\label{theo_complete}
  Let $\dmeas$ and  $\dmeasC$ be two distances on  $\SubProba$ such that
  $\dmeasC$  induces  the  weak topology  on $\SubProba$, $\dmeasC$  is
  complete  and  continuous \wrt  $\dmeas$, and  $\dcut$ is
  (invariant) smooth and weakly regular  on $\Graphon$.  Then, the space
  $(\UGraphon,  \ddcutC)$  is  a  Polish metric  space. 
\end{theo}

Note that the assumptions in Theorem~\ref{theo_complete} imply
that   $\dmeas$  also induces  the  weak topology on $\SubProba$.
Indeed, as $\dmeasC$ is continuous \wrt  $\dmeas$, 
	the topology induced by $\dmeas$ if finer than the topology
	induced by $\dmeasC$, \ie finer than the weak topology.
As $\dcut$ is smooth, by Lemma~\ref{lem:dm-smooth-d-cont},
	$\dmeas$ is continuous \wrt the weak topology
	(\ie the weak topology if finer than the topology induced by $\dmeas$),
	and thus $\dmeas$ induces the weak topology on $\SubProba$.

Also note that \Cref{theo_complete} can easily be extended to $\SubGraphon$
or the space of unlabeled $\cM$-valued kernels $\UKernelM$
when $\cM$ is a bounded convex closed subset of $\Meas$.

\begin{proof}
  From Lemma~\ref{lem:ddC-metrique},  we have that $\ddcutC$  is a distance
  on $\UGraphon$ which  induces the same topology as  $\ddcut$, and from
  Proposition~\ref{theo_separable}, we have that $(\UGraphon,  \ddcut)$, and thus
  $(\UGraphon, \ddcutC)$, is separable.  To get that this latter space is
  Polish, we are left to prove  that the distance $\ddcutC$ is complete.
  \medskip

Let $(W_n )_{n\in\N}$ be a sequence of  probability-graphons that is
Cauchy for $\ddcutC$.  By definition  of the cut distance $\ddcutC$, the
sequence  of probability  measures $(M_{W_n}  )_{n\in\N}$ is  Cauchy in
$\Proba$ for  the complete distance  $\dmeasC$.  
Thus, the  sequence $(M_{W_n}  )_{n\in\N}$ is
weakly convergent as $\dmeasC$ induces the weak topology,
which implies that it is tight
	(see Lemma~\ref{lemme_Bogachev_Prohorov_theorem}).  
Hence, by definition,  the sequence of  probability-graphons $(W_n )_{n\in\N}$ is tight.  
By Theorem~\ref{theo_tension_conv}~\ref{it:tension-cv}, there exists a
subsequence $(W_{n_k} )_{k\in\N}$ that converges for $\ddcut$ to a limit,
say $W\in\UGraphon$.  This subsequence also  converges for $\ddcutC$  to 
$W$  as $\ddcut$  and $\ddcutC$  induce the  same topology.
Finally, because the sequence $(W_n )_{n\in\N}$ is Cauchy  for $\ddcutC$ and  has a
subsequence converging  to $W$  for $\ddcutC$,  the whole  sequence must
also  converge  to  $W$   for  $\ddcutC$.   Consequently,  the  distance
$\ddcutC$ is complete.
\end{proof}

The following lemma shows that every probability measure 
can be represented as a constant probability-graphon.

\begin{lemme}[$\Proba$ seen as a closed subset of $\Graphon$]
	\label{constant_graphons_closed}
  Let $\dmeas$ be a distance on  $\SubProba$ such that $\dcut$  is
  (invariant and) smooth  on  $\Graphon$. Then, the
  map $\mu \mapsto W_\mu \equiv \mu $ is an  injection from $(\Proba,
  \dmeas)$ to $(\UGraphon, \ddcut)$ 
with a closed range and continuous inverse. 
\end{lemme}

\begin{proof}
  For any $\mu\in \Proba$  consider the constant probability-graphon
  $W_\mu\equiv     \mu$,     and    notice     that     $M_{W_\mu}=\mu$,
  that $W_\mu(S\times T;  \cdot)=\lambda(S)\lambda(T)\,  \mu$  for  all  measurable
  $S,  T  \subset  [0,1]$,  and that $W_\mu^\varphi=W_\mu$  for  any
  measure-preserving  map   $\varphi$.
  This   readily  implies   that  for
  $\mu\in \Proba$ and $W\in \Graphon$:
 \begin{equation}
    \label{eq:W-mu-W}
   \ddcut(W_\mu, W) = \dcut(W_\mu, W) =\sup_{S, T\subset [0, 1]}
  \dmeas(\lambda(S)\lambda(T) \mu, W(S\times T; \cdot)) \geq  \dmeas(\mu, M_W)  .
 \end{equation}

In particular, taking $W=W_\nu$ for $\nu\in \Proba$ we get
that $\ddcut(W_\mu, W_\nu)\geq 
  \dmeas( \mu, \nu) $. This implies that 
the map $\cI: \mu \mapsto W_\mu \equiv \mu $ is an injection,
and its inverse, given by the map $W_\mu \mapsto \mu$, is 1-Lipschitz.

\medskip

Let $(\mu_n )_{n\in\N}$ be a sequence in $\Proba$ such that the sequence
$(W_{\mu_n} )_{n\in\N}$ converges for $\ddcut$ to a limit, say
$W$. We deduce from~\eqref{eq:W-mu-W}  that $(\mu_n)_{n\in\N}$ 
converges for $\dmeas$ to $\mu=M_W$ and that for all measurable $S,
T\subset [0,1]$, $(\lambda(S)\lambda(T) \mu_n)_{n\in\N}$ 
converges for $\dmeas$ to $W(S\times T; \cdot)$. This implies that $W(S\times T; \cdot)
=\lambda(S)\lambda(T) \mu(\cdot)$ for all measurable $S,
T\subset [0,1]$, that is,  $W=W_\mu$. This implies that the
image by $\cI$ of any closed subset of $\Proba$ is a closed subset of
$\Graphon$, and thus the 
range of $\cI$ is closed. 
\end{proof}

\begin{rem}[Extension to isometric representation of $\Proba$]
  \label{rem:isometrie}
  If  the   distance  $\dmeas$,   in  addition   to  the   hypothesis  of
  Lemma~\ref{constant_graphons_closed}, is sub-homogeneous,  that is, for
  all         $\mu,        \nu\in         \Proba$        we         have
  $\dmeas(\mu, \nu)=\sup_{r\in  [0, 1]}  \dmeas(r \mu, r\nu)$  
  (which is the case if $\dmeas$ is quasi-convex), 
  then we deduce   from~\eqref{eq:W-mu-W}      that        the        map
  $\mu \mapsto W_\mu  \equiv \mu $ is isometric  from $(\Proba, \dmeas)$
  to $(\UGraphon, \ddcut)$.
\end{rem}

We now state characterization of compactness and completeness
for the space of probability-graphons. Recall $\Space$ is a Polish
space. 

\begin{corol}[Characterization of compactness and completeness for
  $\UGraphon$]
  \label{cor:W1compact}
  Let $\dmeas$ be  a distance on $\SubProba$, which  induces 
  	the weak topology on $\SubProba$,
  and  such that  $\dcut$ is (invariant) smooth  and weakly
  regular on $\Graphon$.  We have the following properties.
\begin{enumerate}[label=(\roman*)]
\item \label{it:topo_compact} 
  $\Space$ is compact
  $\iff$ $(\SubProba, \dmeas)$ is compact
  $\iff$   $(\UGraphon, \ddcut)$ is compact.
\item \label{it:topo_complete} 
If $(\SubProba, \dmeas)$ is complete then  $(\UGraphon, \ddcut)$ is
complete. 
\item  \label{it:topo_complete_2}  Assume  furthermore that  $\dmeas$  is
  sub-homogeneous  (see  Remark~\ref{rem:isometrie}).
If $(\UGraphon, \ddcut)$  is
  complete, then  $(\Proba,  \dmeas)$ is  complete.
\end{enumerate}
\end{corol}

\begin{proof}
  We prove Point~\ref{it:topo_compact}. 
rom \Cref{rem_Proba_compact}, we already know that
$\Space$ is compact if and only if $\SubProba$ is weakly compact,
\ie compact for $\dmeas$ as $\dmeas$ 
induces the weak topology on $\SubProba$.

Now, assume that $(\SubProba, \dmeas)$ is compact.
Applying Theorem~\ref{theo_tension_conv}~\ref{it:WG=compact},
we get that the space $(\UGraphon, \ddcut)$ is also compact.

Conversely, assume that $(\UGraphon, \ddcut)$ is compact.
By Lemma~\ref{conv_graphon_conv_measure},
the mapping $W \mapsto M_W$ is continuous
from $(\UGraphon, \ddcut)$ to $(\Proba, \dmeas)$,
and as $(\UGraphon, \ddcut)$ is compact its image through
this mapping is also compact. To conclude, it is enough to check that
this mapping is  surjective. But this is clear as the image of  the
constant probability-graphon $W_\mu \equiv \mu$ is  $M_{W_\mu} = \mu$.
Hence, $(\Proba, \dmeas)$ (and thus $(\SubProba, \dmeas)$) is compact.
\medskip

We prove Point~\ref{it:topo_complete}. 
Assume that $(\SubProba, \dmeas)$ is complete.
Thus, we can choose $\dmeasC = \dmeas$ in Definition~\ref{ref_dist_complete},
and apply Theorem~\ref{theo_complete} to get that $(\UGraphon, \ddcutC)$ is complete.
As $\dmeasC = \dmeas$, we have $\ddcut \leq \ddcutC \leq 2 \ddcut$.
Hence, $(\UGraphon, \ddcut)$ is also complete.
\medskip

We prove  Point~\ref{it:topo_complete_2}.
Assume  that   $(\UGraphon,  \ddcut)$   is  complete.    Let
$(\mu_n )_{n\in\N}$  be a  Cauchy sequence  of probability  measures in
$(\Proba, \dmeas)$.   By Remark~\ref{rem:isometrie}, the  sequence of
constant probability-graphons $(W_{\mu_n} )_{n\in\N}$ 
is also Cauchy  for  $\ddcut$.    As  $(\UGraphon,  \ddcut)$  is  complete,   there  exists  a
probability-graphon $W\in\UGraphon$ such  that $(W_{\mu_n} )_{n\in\N}$ converges to $W$
for      the      cut       distance      $\ddcut$.       Thanks      to
Lemma~\ref{constant_graphons_closed},  $W$  is  constant equal  to  some
$\mu\in\Proba$, and $(\mu_n)_{n\in\N}$ converges  to $\mu$ for $\dmeas$.  Hence,
$(\Proba, \dmeas)$ is complete. 
\end{proof}

\section{Sampling from probability-graphons}\label{section_sampling}

Measured-valued graphons allow to define models for generating random weighted graphs
that are more general than the models based on real-valued graphons.
We prove that the weighted graphs sampled from probability-graphons
are close to their original model for the cut distance $\ddcutF$,
where $\F = (f_k)_{k\in\N}$ (with $f_0 = \un$) is a convergence determining sequence.

It would have been more natural to work in 
	Sections~\ref{section_sampling} 
	and \ref{section_counting_lemmas}
with the Kantorovitch-Rubinshtein norm
	or the Fortet-Mourier norm
that both treats all test functions in a uniform manner.
Unfortunately, the supremum in the definition of both
of this norms does not behave well regarding
the probabilities and expectations of graphs
sampled from probability-graphons.
We need in our proofs
(and in particular that of the First Sampling Lemma~\ref{FirstSamplingLemma} below)
to consider simultaneously only a finite number of test functions
in order to control the probability of failure
for our stochastic bounds.

\subsection{\texorpdfstring{$\Proba$}{}-Graphs and weighted graphs}

A graph $G=(V,E)$ is composed of a finite set of vertices $V(G)=V$,
and a set of edges $E(G)=E$ which is a subset of $V\times V$
	avoiding the diagonal.
When its set of edges $E(G)$ is symmetric, we say that $G$ is \emph{symmetric}
or \emph{non-oriented}.
We denote by $v(G)=|V(G)|$ the number of vertices of this graph,
and by $e(G)=|E(G)|$ its number of edges.

\begin{defin}[$\mathcal{X}$-graphs]
Let $\mathcal{X}$ be a non-empty set.
A  \emph{$\mathcal{X}$-graph} 
is a triplet $G=(V,E,\Phi)$ where $(V,E)$ is a graph
and $\Phi: E\to \mathcal{X}$ is a map that associates a decoration $x=\Phi(e)\in \mathcal{X}$ to each edge $e\in E$.
When $\mathcal{X}=\Space$, we say that G is a \emph{weighted graph}.

Furthermore, the graph $G$
is said to be \emph{symmetric}
if $(V,E)$ is a symmetric graph and if $\Phi$ 
is a symmetric  function,
\ie for every edge $(x,y)\in E$, we have $(y,x)\in E$ and $\Phi(x,y) = \Phi(y,x)$.
\end{defin}

\begin{remark}[$\Proba$-Graphs as probability-graphons]
	\label{rem_def_W_G}
Any labeled $\Proba$-graph $G$
can be naturally represented as an $\Proba$-valued graphon,
which we denote by $W_G$, in the following way.
Let $G = (V,E,M)$ be a $\Proba$-graph,
with $v(G)=n\in\N^*$.
Denote by $V = [n] = \{1,\dots,n\}$ the vertices of $G$.
Consider intervals of length $1/n$:
for $1\leq i\leq n$, let $J_i =\ ( (i-1)/n , i/n]$.
We then define the $\Proba$-valued graphon stepfunction $W_G$ 
associated with the $\Proba$-graph $G$ by:
\[ \forall (i,j)\in E,\quad 	\forall (x,y) \in J_i\times J_j,\quad 	W_G(x,y;\drv z) = 
\Phi(i,j)(\drv z) ;
\]
and $W_G(x,y;\drv z)$ equals the Dirac mass at $\partial$ otherwise,
where $\partial$ is an element of $\Space$ used as a cemetery point for missing edges in graphs.
\end{remark}

In this section, 
we investigate weighted graphs sampled from probability-graphons. Hence, using the cemetery point argument in the remark above,
we only consider complete graphs for the rest  of this section.

\medskip
Let $\dmeas$ be a distance on $\SubProba$.
If  $G$ and $H$ have the same vertex-set,
the cut distance between them is defined as the cut distance between their associated graphons:
\begin{equation*}
\dcut(G,H) = \dcut(W_G,W_H) . 
\end{equation*}
When $G$ and $H$ does not have the same vertex-sets, as the numbering of the vertices in \Cref{rem_def_W_G} is arbitrary, we must consider 
the unlabeled cut distance between them  defined as the cut distance between their associated graphons:
\begin{equation*}
\ddcut(G,H) = \ddcut(W_G,W_H).
\end{equation*}

Remind that 
when the distance $\dmeas$ derives from a norm $\NmeasSymbol$
on $\SignedMeas$, 
Lemma~\ref{lemma_cut_dist_combi_2} applies,
and the cut distance $\dcut(G,H)$ can be rewritten as
a combinatorial optimization over whole steps.

\begin{remark}[Weighted graphs as $\Proba$-graphs]
	\label{weighted_graphs_as_proba_graphs}
We will sometimes need to interpret
a weighted graph $G$ as 
a $\Proba$-graph where a weight $x$ on an edge is replaced
by $\delta_x$ the Dirac mass laocated at $x$.
\end{remark}

\begin{notation}[{The real-weighted graph $G[f]$}]
	\label{notation_graphe_function}
For a $\Proba$-graph (resp. weighted graph) $G$ 
and a function $f\in\CbFunct$,
we denote by $G[f]$ the real-weighted graph with the same
vertex set and edge set as $G$, 
and where the edge $(i,j)$ has weight 
$\Phi_{G[f]}(i,j)=\Phi _G(i,j;f) = \int_{\Space} f(z)\Phi_G(i,j;\drv z)$
(resp. $\Phi_{G[f]}(i,j)= f( \Phi_G(i,j) )$), where $\Phi_G$ is the decoration of the 
$\Proba$-graph  $G$.
\end{notation}

\subsection{\texorpdfstring{$W$}{}-random graphs}
	\label{subsection_W_random_graphs}

Let $W$ be a probability-graphon, and $x=(x_1,\dots,x_n)$, $n\in\N^*$,
	be a sequence of points from $[0,1]$.
We define the $\Proba$-graph $\mathbb{H}(x,W)$
as the complete graph whose vertex set is $[n] = \{1,\dots,n\}$,
and with each edge $(i,j)$ decorated by the probability measure 
$W(x_i,x_j;\drv z)$.

Let $H$ be any $\Proba$-graph.
We can define from $H$ a random weighted (directed) graph $\mathbb{G}(H)$
whose vertex set $V(H)$ and edge set $E(H)$ are the same as $H$,
and with each edge $(i,j)$ having a random weight $\beta_{i,j}$ 
distributed according to 
the probability distribution  decorating the edge $(i,j)$ in $H$,
all the weights being independent from each other.
For the special case where $H = \mathbb{H}(x,W)$,
we simply note  $\mathbb{G}(x,W) = \mathbb{G}(\mathbb{H}(x,W))$.

An important special case is when the sequence $X$ is chosen at random:
$X=(X_i)_{1\leq i\leq n}$ where the $X_i$ are independent 
and uniformly distributed on $[0,1]$.
For this special case, we simply note
$\mathbb{H}(n,W) =\mathbb{H}(X,W)$
and $\mathbb{G}(n,W) =\mathbb{G}(X,W)$,
that are conditionally on $X=x$,  distributed respectively as
$\mathbb{H}(x,W)$ and  $\mathbb{G}(x,W)$.
The random graphs $\mathbb{H}(n,W)$ and $\mathbb{G}(n,W)$
	are called $W$-random graphs.

\begin{remark}[The case of symmetric graphons]
In the special case where $W$ is a symmetric probability-graphon,
the $\Proba$-graph $\mathbb{H}(x,W)$ is also symmetric.
From a symmetric $\Proba$-graph $H$,
the random weighted graph $\mathbb{G}(H)$ is not necessarily symmetric,
but we can define a random symmetric weighted graph 
$\mathbb{G}^{\text{sym}}(H)$
whose vertex set $V(H)$ and $E(H)$ are the same as $H$,
and with independent weights $\beta_{i,j} = \beta_{j,i}$ on each edge $(i,j)=(j,i)$
distributed according to $\Phi_H(i,j;\cdot)$.
For $H = \mathbb{H}(x,W)$ we simply note $\mathbb{G}^{\text{sym}}(x,W)$
and $\mathbb{G}^{\text{sym}}(n,W)$.
\end{remark}

For a weighted graph $G$, and for $1\leq k \leq v(G)$,
we can define the random weighted graph $\mathbb{G}(k,G)$
as being the sub-graph of $G$ induced by a uniform random subset of $k$
distinct  vertices from $G$.
Then, upper bounding by the probability that a uniformly-chosen map $[k]\to V(G)$ is non-injective, 
we get the following bound on the total variation distance between the graphs
obtained from $G$ and its associated graphon $W_G$:
\begin{equation*}
	d_{\text{var}}(\mathbb{G}(k,G),\mathbb{G}(k,W_G)) \leq {k \choose 2} \frac{1}{v(G)} 	,
\end{equation*}
where $d_{\text{var}}$ is the total variation distance between probability measures.

\subsection{Estimation of the distance by sampling}

\subsubsection{The first sampling lemma}

In this subsection, we link
sampling from graphons with the cut distance.
This result is the equivalent of Lemma 10.6 in~\cite{Lovasz}.
The main consequence of the following lemma is that 
the cut distance $\dcutF$ between two probability-graphons
can be estimated by sampling.

\begin{notation}[The random stepfunction $W_X$]
For a measure-valued kernel $W$ (resp. a real-valued kernel $w$)
and a vector $X=(X_i)_{1\leq i \leq k}$ composed of 
$k$ independent random variables uniformly distributed over $[0,1]$,
we denote by $W_X=W_{\mathbb{H}(k,W)}$ (resp. $w_X$)
the random measure-valued (resp. real-valued) stepfunction 
with $k$ steps of size $1/k$,
and where the step $(i,j)$ has value $W(X_i,X_j;\cdot)$ (resp. $w(X_i,X_j)$).
\end{notation}

\begin{lemme}[First Sampling Lemma]\label{FirstSamplingLemma}
Let $\F$ be a convergence determining sequence.
Let $k\in\N^*$, and $U,W\in\Graphon$ be two probability-graphons,
and let $X$ be a random vector uniformly distributed over $[0,1]^k$.
Then with probability at least $1-4 k^{1/4} \e^{-\sqrt{k}/10}$, we have:
\[ -\frac{2}{k^{1/4}} \leq \NcutF{U_X - W_X} - \NcutF{U-W} \leq \frac{9}{k^{1/4}}  \cdot\]
\end{lemme}

An immediate consequence of Lemma~\ref{FirstSamplingLemma}
is that the decorated graphs with probability measures on their edges
$\mathbb{H}(k,U)$ and $\mathbb{H}(k,W)$
can be coupled in order that
$\dcutF(\mathbb{H}(k,U),\mathbb{H}(k,W))$ 
is close to $\dcutF(U,W)$ with high probability.

To prove the first sampling lemma,
we first need to prove the following lemma
which states that the cut norm $\NcutFSymbol$
can be approximated by the maximum
of the one-sided cut norm using a finite number of function.
Remind from Remark~\ref{rem:N_F_and_N_R_+}
the definition of  the one-sided version of the cut norm $\NcutRposSymbol$.

\begin{lemme}[Approximation bound with $\NcutFSymbol$ and $\NcutRposSymbol$]
	\label{lemma_approx_cut_norm_with_eps}
Let $U,W\in\Graphon$ and let $N\in\N$.
For every $\eps=(\eps_n)_{1\leq n \leq N}\in \{ \pm 1 \}^N$,
define $g_{N,\eps} = \sum_{n=1}^N 2^{-n} \eps_n f_n$.
Then, we have:
\begin{equation*}
\NcutF{U-W} - 2^{-N}
\leq
\max_{\eps\in \{\pm 1\}^N} 
	\NcutRposLarge{(U-W)\left[ g_{N,\eps} \right]}
\leq \NcutF{U-W} .
\end{equation*}
\end{lemme}

\begin{proof}
First remark that for $n\in\N$, $f_n$ takes values in $[0,1]$,
and thus $U[f_n]-W[f_n]$ takes values in $[-1,1]$.
Remind that $f_0 = \un$, and thus $U[f_0]-W[f_0] \equiv 0$.
Upper bounding integrals by $1$ for indices $n>N$, we get:
\begin{equation*}
\NcutF{U-W}
\leq
\sup_{S,T \subset [0,1]} \sum_{n=1}^N 2^{-n}
		\left\vert \int_{S\times T} (U-W)[f_n](x,y)\ \drv x\drv y \right\vert
+ 2^{-N}.
\end{equation*}
And adding the non-negative terms for $n>N$, we get:
\begin{equation*}
\sup_{S,T \subset [0,1]} \sum_{n=1}^N 2^{-n}
		\left\vert \int_{S\times T} (U-W)[f_n](x,y)\ \drv x\drv y \right\vert
\leq \NcutF{U-W}.
\end{equation*}
Using the same idea as in~\eqref{eq:def_norm_F_with_eps}
and~\eqref{eq:def_cut_norm_with_eps}, we get:
\begin{equation*}
\sup_{S,T \subset [0,1]} \sum_{n=1}^N 2^{-n}
		\left\vert \int_{S\times T} (U-W)[f_n](x,y)\ \drv x\drv y \right\vert
= \max_{\eps\in \{\pm 1\}^N} 
	\NcutRposLarge{(U-W)\left[ g_{N,\eps} \right]},
\end{equation*}
which concludes the proof.
\end{proof}

\begin{proof}[Proof of Lemma~\ref{FirstSamplingLemma}]
Remark that for $f\in\CbFunct$ and $W\in\Kernel$,
we have $(W_X)[f] = (W[f])_X$,
and we thus write $W[f]_X$ without any ambiguity.

Assume that $k\geq 2^4$
(otherwise the lower bound in the lemma is trivial).
Set $N=\lceil \log_2(k^{1/4}) \rceil$,
so that $2^{-1}k^{-1/4} < 2^{-N}  \leq k^{-1/4}$.
Let $\eps  \in  \{ \pm 1 \}^N$.
Remark that as the $f_n$ take values in $[0,1]$,
the real-valued kernels $(U-W)[f_n]$ take values in $[-1,1]$,
and thus 
the real-valued kernel $(U-W)[g_{N,\eps}]$
also take values in $[-1,1]$.
Applying Lemma~10.7 in \cite{Lovasz} to
the real-valued kernel $(U-W)\left[g_{N,\eps}\right]$,
we get with probability at least $1-2\e^{-\sqrt{k}/10}$ that:
\begin{equation}\label{eq_FSL_bound_W_eps}
- \frac{3}{k} \leq
\NcutRpos{(U-W)[g_{N,\eps}]_X}
- \NcutRpos{(U-W)[g_{N,\eps}] }
\leq \frac{8}{k^{1/4}} ,
\end{equation}
where remind that $\NcutRpos{\cdot}$ is the one-sided version of the
	cut norm for real-valued kernels defined in \eqref{eq_def_NcutR}.
Hence, with probability at least 
	$1-2^{N+1}\e^{-\sqrt{k}/10} \geq 1 - 4k^{1/4}\e^{-\sqrt{k}/10}$,
we have that the bounds in~\eqref{eq_FSL_bound_W_eps}
holds for every $\eps \in \{ \pm 1 \}^N$ simultaneously;
and when all of this holds, applying Lemma~\ref{lemma_approx_cut_norm_with_eps}
to $U, W$ and to $U_X, W_X$, we get:
\begin{align*}
\NcutF{U_X-W_X}
& \leq  \max_{\eps\in \{\pm 1\}^N} \NcutRpos{(U-W)[g_{N,\eps}]_X} 
		+ 2^{-N} \\
& \leq  \max_{\eps\in \{\pm 1\}^N} \NcutRpos{(U-W)[g_{N,\eps}] } 
		+ \frac{9}{k^{1/4}} \\
& \leq \NcutF{U-W} + \frac{9}{k^{1/4}},
\end{align*}
and similarly:
\begin{align*}
\NcutF{U-W}
& \leq  \max_{\eps\in \{\pm 1\}^N} \NcutRpos{(U-W)[g_{N,\eps}] } 
		+ 2^{-N} \\
& \leq  \max_{\eps\in \{\pm 1\}^N} \NcutRpos{(U-W)[g_{N,\eps}]_X} 
		+ \frac{1}{k^{1/4}} + \frac{3}{k} \\
& \leq \NcutF{U_X-W_X} + \frac{2}{k^{1/4}} \cdot
\end{align*}
This concludes the proof.
\end{proof}

\subsubsection{Approximation with random weighted graphs}
	\label{subsection:approx_M1_graphs}

As a consequence of the First Sampling Lemma~\ref{FirstSamplingLemma},
we get that the cut distance between the sampled graphs
$\mathbb{H}(k,U)$ and $\mathbb{H}(k,W)$
(with the proper coupling)
is close to the cut distance between the probability-graphons $U$ and $W$.
The following lemma states that if $k$ is large enough,
then $\mathbb{G}(k,W)$ is close to $\mathbb{H}(k,W)$
in the cut distance $\dcutF$,
and thus the cut distance between
the random weighted graphs
$\mathbb{G}(k,U)$ and $\mathbb{G}(k,W)$
is also close to $\dcutF(U,W)$.

Recall from Section~\ref{subsection_W_random_graphs} the definition of the random weighted graph 
	$\mathbb{G}(H)$ when $H$ is a $\Proba$-graph.
Following Remarks \ref{weighted_graphs_as_proba_graphs}
 and \ref{rem_def_W_G}, we shall see the weighted graph
$\mathbb{G}(H)$ as a $\Proba$-graph or even as a probability-graphon.

\begin{lemme}[Bound in probability for $\dcutF(\mathbb{G}(H),H)$]
	\label{lemma_graph_close}
For every $\Proba$-graph $H$ with $k$  vertices,
and for every $\eps \geq 10/\sqrt{k}$, we have:
\[ \Prb\Bigl(\dcutF(\mathbb{G}(H),H) > 2\eps\Bigr) \leq e^{-\eps^2k^2}  . \]
\end{lemme}

\begin{remark}[Bound in expectation for $\dcutF(\mathbb{G}(H),H)$]
	\label{graph_close_esp}
Remind that $\dcutF(\mathbb{G}(H), H) \leq 1$.
Applying Lemma~\ref{lemma_graph_close} with $\eps = 10 / \sqrt{k}$,
we get the following  bound on the expectation
of $\dcutF(\mathbb{G}(H),H)$:
\begin{equation*}
\Esp[\dcutF(\mathbb{G}(H),H)] \leq \frac{20}{\sqrt{k}} + \e^{-100 k} < \frac{21}{\sqrt{k}}  \cdot
\end{equation*}
\end{remark}

\begin{proof}[Proof of Lemma~\ref{lemma_graph_close}]
Let $H$ and $\eps$ be as in the lemma.
Assume that $\eps \leq 1/2$ (otherwise the probability to bound in the lemma is null).
To simplify the notations, denote by $G = \mathbb{G}(H)$
through this proof.
Define $N=\lceil \log_2(\eps^{-1}) \rceil$,
so that $\sum_{n = N+1}^\infty 2^{-n} \allowbreak \leq \eps$.
Upper bounding by $1$ the terms for $n > N$
in~\eqref{major_dist_graph}, we get for $U,W\in\Graphon$: 
\begin{equation*}
\dcutF(U,W)
\leq \sum_{n=1}^N 2^{-n} \NcutR{U[f_n] - W[f_n]}	+\ \eps	 ,
\end{equation*}
where remind that $\NcutRSymbol$ is the cut norm for real-valued kernels
	defined in \eqref{eq_def_NcutR}.
Using this equation with the graphs $G$ and $H$, we get:
\begin{align}
\Prb(\dcutF(G,H) > 2\eps )
& \leq \Prb\left(\sum_{n=1}^N 2^{-n} \dcutR(G[f_n], H[f_n]) \ > \eps \right) \nonumber\\
& \leq \sum_{n=1}^N \Prb\left( \dcutR(G[f_n], H[f_n]) \ > \eps \right)	 ,
\label{borne_proba_graph_2}
\end{align}
where $\dcutR$ denotes the cut distance associated to the cut norm $\NcutRSymbol$
	for real-valued graphons and kernels.
Remark that for every $n\in\N$,
$H[f_n]$ and $G[f_n]$ are real-weighted graphs with weights in $[0,1]$.
Thus, by a straightforward adaptation of the proof of \cite[Lemma 10.11]{Lovasz},
we get:
\begin{equation}\label{borne_proba_graph_3}
\forall n \in [ N ],\quad
\Prb( d_\square(G[f_n], H[f_n]) \ > \eps ) \leq 2 \cdot 4^k \e^{-2 \eps^2 k^2}	 .
\end{equation}
Combining \eqref{borne_proba_graph_2} and \eqref{borne_proba_graph_3},
we get for $\eps > 10 / \sqrt{k}$:
\begin{equation*}
\Prb(\dcutF(G,H) > 2\eps ) \leq 2 N  4^k e^{-2 \eps^2 k^2}	
\leq \e^{-\eps^2 k^2 },
\end{equation*}
where the last bound derives from simple calculus.
This concludes the proof.
\end{proof}

We can apply the First Sampling Lemma~\ref{FirstSamplingLemma} along with
Lemma~\ref{lemma_graph_close} to get the following lemma,
equivalent of the first sampling lemma for 
the random weighted graph $\mathbb{G}(k,W)$:

\begin{corol}[First Sampling Lemma for $\bbG(k,W)$]
Let $U,W\in\Graphon$ be two probability-graphons, and $k \in \N^*$.
Then, we can couple the random weighted graphs
$\mathbb{G}(k,U)$ and $\mathbb{G}(k,W)$ such that
with probability at least $1 - (4k^{1/4}+1)\e^{-\sqrt{k}/10}$, we have:
\[ \left| \dcutF(\mathbb{G}(k,U), \mathbb{G}(k,W)) - \dcutF(U,W) \right|
\leq \frac{13}{k^{1/4}}	 \cdot \]
\end{corol}

\begin{proof}
Assume that $k \geq 13^4$ (otherwise the bound in the corollary is trivial).
Then, we have with probability at least $1 - 4k^{1/4}\e^{-\sqrt{k}/10} - 2\e^{-100 k}
> 1 - (4k^{1/4}+1)\e^{-\sqrt{k}/10}$:
\begin{align}
\left| \dcutF(\mathbb{G}(k,U), \mathbb{G}(k,W)) - \dcutF(U,W) \right| 
\ \leq & \ \left| \dcutF(\mathbb{G}(k,U), \mathbb{G}(k,W)) - \dcutF(\mathbb{H}(k,U), \mathbb{H}(k,W)) \right| 
\nonumber \\
& + \left| \dcutF(\mathbb{H}(k,U), \mathbb{H}(k,W)) - \dcutF(U,W) \right|
\nonumber \\
\leq & \  \dcutF(\mathbb{G}(k,U), \mathbb{H}(k,U))  \nonumber \\
& \ + \dcutF(\mathbb{G}(k,W), \mathbb{H}(k,W)) + \frac{9}{k^{1/4}}
\nonumber \\
\leq & \  \frac{40}{\sqrt{k}} + \frac{9}{k^{1/4}}	  \nonumber \\
\leq & \  \frac{13}{k^{1/4}}	, \nonumber
\end{align}
where we used 
the upper bound from the First Sampling Lemma~\ref{FirstSamplingLemma} 
(which gives the coupling with the same random vector $X$ to define both graphs
	$U_X = \bbH(k,U)$ and $W_X = \bbH(k,W)$)
for the second inequality,
the upper bound from Lemma~\ref{lemma_graph_close} with $\eps = 10 / \sqrt{k}$
	with both $U$ and $W$
for the third inequality,
and that $\frac{1}{\sqrt{k}} \leq \frac{1}{13 k^{1/4}}$
for the last inequality.
\end{proof}

\subsection{The distance between a probability-graphon and its sample}

In this section, we present  the Second Sampling Lemma,
that shows that a sampled $\Proba$-graph is close to 
its original probability-graphon with high probability.
Note that we use the unlabeled cut distance $\ddcutF$ rather than $\dcutF$
as the sample points  are unordered.
The bound on the distance is much weaker than the one in the
First Sampling Lemma~\ref{FirstSamplingLemma}, 
but nevertheless goes to $0$ as the sample size increases.

The proof is a straightforward adaptation of the proof of \cite[Lemma 10.16]{Lovasz}
	(replacing the weak regularity lemma and the first sampling lemma
		by their counterparts for probability-graphons,
		that is Lemmas~\ref{Regularity_Lemma_2} and \ref{FirstSamplingLemma};
	the sample concentration theorem for real-valued graphons can easily be adapted
		to probability-graphons).

\begin{lemme}[Second Sampling Lemma]\label{SecondSamplingLemma}
Let $\F$ be a convergence determining sequence.
Let $W\in\UGraphon$ be a probability-graphon and $k\in\N^*$.
Then, with probability at least $1 - \exp( - k / (2\ln(k) ) )$ we have:
\[ \ddcutF( \mathbb{H}(k,W), W) \leq \frac{21}{\sqrt{\ln(k)}}  
\qquad \text{ and  } \qquad
\ddcutF( \mathbb{G}(k,W), W) \leq \frac{22}{\sqrt{\ln(k)}}  \cdot\]
\end{lemme}
In the above lemma, the asymmetric random graph $\mathbb{G}(k,W)$ 
can be replaced by the symmetric random graph $\mathbb{G}^{\text{sym}}(k,W)$
without changing the proof.
Similarly, the results in Section~\ref{subsection:approx_M1_graphs}
can be reformulated with symmetric random graphs $\mathbb{G}^{\text{sym}}(k,W)$ 
	and $\mathbb{G}^{\text{sym}}(H)$
(but with a slight modification of the proof for Lemma~\ref{lemma_graph_close}
	to symmetrize the random variable $X_{i,j}$
	and with the upper bound $e^{-\eps^2k^2/2}$,
	see also \cite[Lemma 10.11]{Lovasz}).

\medskip
As an immediate consequence of \Cref{SecondSamplingLemma} and of the Borel-Cantelli lemma,
we get the convergence of the sampled subgraphs for the cut distance $\ddcutF$.

\begin{theo}[Convergence of sampled subgraphs]
	\label{theo_conv_sampled_subgraphs}
Let $\F$ be a convergence determining sequence.
Let $W\in\UGraphon$ be a probability-graphon.
Then, \as the sequence of sampled subgraphs $(\bbG(k,W))_{k\in\N^*}$
	converges to $W$ for the cut distance $\ddcutF$,
	and thus for any cut distance $\ddcut$ from \Cref{theo_equiv_topo}.
\end{theo}

\section{The Counting Lemmas and the topology of probability-graphons}\label{section_counting_lemmas}

In this section, we introduce the homomorphism densities
for probability-graphons,
and then we link those to the cut distance $\ddcutF$
through the Counting Lemma
and the Inverse Counting Lemma.
Those results are analogous to the case of real-valued graphons,
see  \cite[Chapter 7]{Lovasz} for the definition of homomorphism densities
and \cite[Chapter 10]{Lovasz} for the Counting Lemma and Inverse Counting Lemma.
The main differences with \cite{Lovasz} are: the  decoration of  the edges of the graphs with functions from $\CbFunct$;
the Counting Lemma 
	for the decorations belonging only in the convergence determining sequence $\F$;
the more technical proof of the Inverse Counting Lemma. 
Note that we need to work with $\ddcutF$ here
as the proof of the Inverse Counting Lemma relies on
	the second sampling Lemma~\ref{SecondSamplingLemma}.

\subsection{The homomorphism densities}

In the case of non-weighted graphs, the homomorphism densities $t(F,G)$
allow to characterize a graph (up to twin-vertices expansion),
and also allow to define a topology for real-valued graphons.
In the case of weighted graphs and probability-graphons, 
we need to replace the absence/presence of edges
(which is $0$-$1$ valued) by test functions from $\CbFunct$ decorating each edge.

In this section, 
we often need to fix the underlying (directed) graph structure $F = (V,E)$ 
	(which may be incomplete)
	of a $\CbFunct$-graph
and to vary only the $\CbFunct$-decorating functions $g=(g_e)_{e\in E}$,
thus we will write $F^g=(V,E,g)$ for a $\CbFunct$-graph.
Moreover, when there exists a convergence determining
	sequence $\F$ such that 
	$g_e\in\F$ for every edge $e\in E$,
we say that $F^g$ is a $\F$-graph 
	and use the same notation conventions.

\begin{defin}[Homomorphism density]
	\label{def_hom_dens}
We define the \emph{homomorphism density} of a $\CbFunct$-graph $F^g$
in a signed measure-valued kernel $W\in\Kernel$ as:
\begin{equation}\label{eq_def_hom_dens}
t(F^g,W) = M_W^F(g) 
= \ \int_{[0,1]^{V(F)}}  \prod_{(i,j)\in E(F)} W(x_i,x_j; g_{i,j})		\	\prod_{i\in V(F)} \drv x_i	 .
\end{equation}
Moreover, $M_W^F$ defines a measure on $\Space^E$
(which we still denote by $M_W^F$) which is characterized by
$M_W^F(\otimes_{e\in E} g_e)=M_W^F(g) $
	for $g = (g_e)_{e\in E}$.
\end{defin}

\begin{remark}[Invariance under relabeling of homomorphism densities]
	\label{rem_dens_homo_inv}
Let $\varphi : [0,1] \to [0,1]$ be a measure-preserving map.
As $\varphi^{\otimes k} : (x_1, \dots, x_k) \mapsto (\varphi(x_1), \dots, \varphi(x_k))$ is a measure-preserving map on $[0,1]^k$, 
applying the transfer formula (see \eqref{eq:re-label}),
we get that for every $\CbFunct$-graph $F^g$
and every signed measure-valued kernel $W\in\Kernel$,
we have $t(F^g, W^\varphi) = t(F^g, W)$.
Thus $t(F^g, \cdot)$ can be extending to $\UKernel$.
\end{remark}

\begin{remark}
When $W\in\Kernelp$ is a measure-valued kernel,
and $F$ is the graph with two vertices and one edge, 
we get that $M_W^F = M_W$ the measure defined in \eqref{eq:def-MW}.
\end{remark}

\begin{remark}[Adding missing edges to $F$]
When we work with probability-graphons,
we can always assume the graph $F$ to be complete,
by adding the missing edges $(i,j)$ and decorating them
with the constant function $g_{(i,j)} = \un$.
\end{remark}

For a finite weighted graph $G$, we define the  
\emph{homomorphism density} of the $\CbFunct$-graph $F^g$ in $G$ as $t(F^g,G) = t(F^g,W_G)$
(remind from Remark~\ref{rem_def_W_G} the definition of $W_G$), that is:
\begin{equation*}
t(F^g, G) = \frac{1}{v(G)^{k}} \sum_{(x_1,\cdots, x_{k}) \in V(G)^{k}}\ \prod_{(i,j)\in E(F)} g_{(i,j)}( \Phi_G(x_i,x_j) )
 ,
\end{equation*}
where $k=v(F)$ and $\Phi_G(x_i,x_j)$ is the weight of the directed edge from $x_i$ to $x_j$.

\subsection{The Counting Lemma}

The following lemma links
the homomorphism densities with the cut distance $\ddcutF$
for some convergence determining sequence
	$\F = (f_n)_{n\in\N}$ (with $f_0 = \un$ and $f_n$  takes values  in  $[0,1]$).
This lemma is a generalization to probability-graphons
of the Counting Lemma for real-valued graphons
(see Lemmas~10.22 and 10.23 from~\cite{Lovasz}).
Recall that by \Cref{rem_dens_homo_inv}, $t(F^g,\cdot)$ is defined on $\UKernel$.

\begin{lemme}[Counting Lemma]\label{CountingLemma}
Let $\F= (f_n)_{n\in \N}$ be a convergence determining sequence (with $f_0=\un$ and $f_n$  takes values  in  $[0,1]$).
Let $F^g$ be a $\F$-graph,
and for every edge $e\in E(F)$,
	let $n_e\in\N$ be such that $g_e = f_{n_e}$.
Then, for every probability-graphons $W, W' \in \UGraphon$, we have:
\[ |t(F^g,W) - t(F^g,W')| \leq \left( \sum_{e\in E(F)} 2^{n_e} \right)   \ddcutF(W,W')
 . \]
\end{lemme}

\begin{remark}[$W\mapsto t(F^g,W)$ is Lipschitz]
The Lipschitz constant given by the lemma
is too large to be useful in practical cases.
Nevertheless, the homomorphism density function 
	$W \mapsto t(F^g,W)$ is Lipschitz 
on the space of unlabeled probability-graphons $\UGraphon$
equipped with the cut distance $\ddcutF$.
\end{remark}

\begin{proof}[Proof of Lemma~\ref{CountingLemma}]
To do this proof, we will apply Lemma~10.24 from~\cite{Lovasz},
which applies to graphs $F$ whose
edges are decorated with (possibly different)
real-valued graphons $w = (w_e : e\in E(F))$,
and the associated homomorphism density is defined as
\begin{equation}\label{eq_def_tFw}
t(F,w) = \int_{[0,1]^{V(F)}} \prod_{(i,j)\in E(F)} w_e(x_i,x_j)\ 
		\prod_{i\in V(F)} \drv x_i  .
\end{equation}
Remind from \eqref{eq:notation_graphon_function}
that for a probability-graphon $W\in\Graphon$
and a function $f\in\F$ (which is $[0,1]$-valued
	by our definition of convergence determining sequences),
we have that $W[f]$ is a real-valued graphon.
Define the collections of real-valued graphons
$w = (W[g_e] : e\in E(F) )$
and $w' = (W'[g_e] : e\in E(F) )$.
Notice from \eqref{eq_def_hom_dens} and \eqref{eq_def_tFw}
that we have $t(F,w) = t(F^g,W)$
	and $t(F,w') = t(F^g,W')$.
Applying \cite[Lemma~10.24]{Lovasz}
to the graph $F$ 
and edge-decorations $w$ and $w'$,
we get:
\[ 
\vert t(F^g,W) - t(F^g,W') \vert
= \vert t(F,w) - t(F,w') \vert 
\leq \sum_{e\in E(F)} \NcutR{W[g_e] - W'[g_e]}  ,
\]
where the norm $\NcutRSymbol$ in the upper bound
is the cut norm for real-valued graphons
	(see \eqref{eq_def_NcutR} for definition of this object). 
For $e\in E(F)$,
by definition of the cut distance $\dcutF$
and using~\eqref{eq_def_NmeasF},
we have:
\[ \NcutR{W[g_e] - W'[g_e]}\ \leq 2^{n_e}\  \dcutF(W, W')  . \]
Hence, combining all those upper bounds,
we get the bound in the lemma
but with $\dcutF$ instead of $\ddcutF$.
Since  $t(F^g,\cdot)$ is invariant under relabeling by \Cref{rem_dens_homo_inv}, 
taking the infimum other all relabelings allows
to replace $\dcutF$ by $\ddcutF$ and to get the bound in the lemma.
\end{proof}

We have just seen that homomorphism densities defined using only
functions from $\F$ are Lipschitz.
We are going to see that the other homomorphism densities
are nevertheless continuous.

\begin{lemme}[Weak Counting Lemma]\label{WeakCountingLemma} 
Let $\F$ be a convergence determining sequence
	(with $f_0 = \un$).
Let $(W_n )_{n\in\N}$ and $W$ be probability-graphons
such that $\lim_{n\to\infty}t(F^g,W_n) = t(F^g,W)$  for all $\F$-graph $F^g$ (which in particular the case if
$\lim_{n\to\infty} \ddcutF(W_n,W) = 0$ by the  Counting Lemma~\ref{CountingLemma}). 
Then, for every $\CbFunct$-graph $F^g$ 
we have:
\[ t(F^g,W_n) \underset{n\to\infty}{\longrightarrow} t(F^g,W)  . \]
\end{lemme}

\begin{proof}
Let $F=(V,E)$ be some fixed (directed) graph.
By assumption, we have for all edge-decorations $g=(g_e)_{e\in E}$ in $\F$ that
$\lim_{n\to\infty} M_{W_n}^F(\otimes_{e\in E} g_e)
	= M_W^F(\otimes_{e\in E} g_e)$
	(see Definition~\ref{def_hom_dens}).
By \cite[Chapter 3, Proposition 4.6]{Ethier},
$\F^{\otimes E}$ is a (countable) convergence determining family on $\MeasX{\Space^E}$.
Thus, the sequence of measures $(M_{W_n}^F)_{n\in\N}$
converges to $M_W^F$
for the weak topology on $\MeasX{\Space^E}$.
And in particular, for every edge-decoration function 
	$g = (g_e)_{e\in E}$
(here for every $e\in E$, $g_e\in \CbFunct$ is arbitrary)
we have $M_{W_n}^F(\otimes_{e\in E} g_e) = t(F^g,W_n) \to t(F^g,W) = M_W^F(\otimes_{e\in E} g_e)$ as $n\to\infty$.
This being true for all choices of the graph $F$,
it concludes the proof.
\end{proof}

\subsection{The Inverse Counting Lemma}

The goal of this subsection is to establish a converse  to the Counting Lemma:
if two probability-graphons are close in terms of homomorphism densities,
then they are close \wrt the cut distance $\ddcutF$.

\begin{lemme}[Inverse Counting Lemma]
	\label{InverseCountingLemma}
Let $\F= (f_n)_{n\in \N}$ be a convergence determining sequence
	(with $f_0 = \un$ and $f_n$  takes values  in  $[0,1]$).
Let $U,W \in \UGraphon$ be two probability-graphons, and let $k, n_0\in \N^*$.
Assume that we have $| t(F^g,U) - t(F^g,W) | \leq 2^{-k - n_0 k^2}$
for every (complete) $\CbFunct$-graph $F^g$ with $k$  vertices 
and such that the edge-decoration functions $g=(g_e)_{e\in E(F)}$
are  products (without repetition)
	of the functions $(f_n)_{1\leq n\leq n_0}$ and $(1-f_n)_{1\leq n\leq n_0}$.
Then, we have:
\[ \ddcutF(U,W) \leq \frac{44}{\sqrt{\log(k)}} + 2^{-n_0}    . \]
\end{lemme}

To prove Lemma~\ref{InverseCountingLemma}, we first need to prove
	the special case where the space $\Space$ is finite.

\begin{lemme}[Inverse Counting Lemma, case with finite space $\Space$]
	\label{InverseCountingLemmaWithFiniteSpace}
Assume that the space $\Space$ is finite with cardinality $n_1$,
	for simplicity say $\Space = [n_1]$.
Define the indicator functions $f_n : z\mapsto \ind_{\{z = n\}}$ for $n\in [n_1]$,
	in particular $\cH= (f_n)_{1\leq n \leq n_1}$ is a finite convergence determining sequence.
Let $U,W \in \UGraphon$ be two probability-graphons, and let $k\in \N^*$.
Assume that we have $| t(F^g,U) - t(F^g,W) | < 2^{-k - \log_2(n_1) k^2}$
for every (complete) $\cH$-graph $F^g$ with $k$  vertices.

Then, for any (possibly finite) convergence determining sequence $\F$, we have:
\[ \ddcutF(U,W) \leq \frac{44}{\sqrt{\log(k)}}  \cdot \]
\end{lemme}

Abusing notations, we can identify a weight-value $n\in\Space$
	with its indicator functions $f_n$,
	and doing this identification for edge-decoration functions,
	we can identify a $\F$-graph $F^g$ 
		with its corresponding weighted graph.
In particular, doing so we get $t(F^g,W) = \Prb( \bbG(k,W) = F^g)$
	for every $\F$-graph $F^g$ with $k$ vertices.
The proof of \Cref{InverseCountingLemmaWithFiniteSpace} is
	then a straightforward adaptation of the proof of \cite[Lemma 10.31 and Lemma 10.32]{Lovasz}.

\begin{proof}[Proof of Lemma~\ref{InverseCountingLemma}]
As the functions $(f_n)_{n\in \N}$ take value in $[0,1]$, 
for all $\varphi$ measure-preserving map,
for all $S,T \subset [0,1]$ measurable sets
and for all $n\in \N$, we have:
\[ \Bigl\vert U(S\times T;f_n) - W^\varphi(S\times T;f_n) \Bigr\vert \leq 1 .\]
Using this bound, we get the following bound
(remind that $f_0 = \un$):
\begin{equation}\label{eq_ICL_1}
\ddcutF(U,W) \leq
\inf_{\varphi\in\InvRelabel} \sup_{S,T \subset [0,1]} \sum_{n=1}^{n_0} 2^{-n} 
			\Bigl\vert U(S\times T;f_n) - W^\varphi(S\times T;f_n) \Bigr\vert
+ 2^{-n_0}		.
\end{equation}

Hence, for a point $z\in\Space$, the upper bound in \eqref{eq_ICL_1}
uses only the information given by $(f_n(z))_{n\in [n_0]}$.
In order to discretize the space $[0,1]^{n_0}$, we replace 
a point $p = (p_1,\dots, p_{n_0}) \in [0,1]^{n_0}$
by a random point $(Y_1,\dots,Y_{n_0}) \in \{0,1\}^{n_0}$
where the $Y_i$ are independent random variables
with Bernoulli distribution of parameter $p_i$.
This leads us to replace a $\Proba$-valued kernel $W$
by the $\mathcal{M}_1(\{0,1\}^{n_0})$-valued kernel $\tilde{W}$
defined for all $(x,y) \in [0,1]^2$, and for all $s =(s_1, \dots, s_{n_0})\in  \{0,1\}^{n_0}$ as:
\begin{equation*}
\tilde{W}(x,y;\{s\}) = W(x,y; f^s) 
\qquad\text{ where }\qquad
f^s = \prod_{n=1}^{n_0} f_n^{s_n} (1- f_n)^{1-s_n} . 
\end{equation*}

Fix some enumeration $(s^m)_{m \in [2^{n_0}]}$ of the points in $\{0,1\}^{n_0}$,
and define the indicator functions $\tilde{h}_m : s \mapsto \ind_{\{s=s^m\}}$ for $m\in [2^{n_0}]$,
	in particular $\tilde{\cH} = (\tilde{h}_m)_{1\leq m \leq 2^{n_0}}$
	is a finite convergence determining sequence on $\MeasX{\{0,1\}^{n_0}}$.
Let $F^{\tilde{g}}$ be a $\tilde{\cH}$-graph with vertex set $V(F)=[k]$,
and for every edge $e\in E(F)$, let $m_e\in [2^{n_0}]$ be such that
	$\tilde{g}_e = \tilde{h}_{m_e}$.
Define the edge-decoration functions $g=(g_e)_{e\in E(F)}$
	for every edge $e\in E(F)$ as $g_e = f^{s^{m_e}}$,
then we get:
\begin{equation*}
t(F^{\tilde g},\tilde{W}) 
= \int_{[0,1]^k} 	\prod_{(i, j)\in E(F)}
			\tilde{W}(x_i,x_j; \{ s^{m_e} \})
		\prod_{i=1}^k \drv x_i
= t(F^g,W)	.
\end{equation*}
Thus, the $\ProbaX{\{0,1\}^{n_0}}$-valued graphons $\tilde{U}$ and $\tilde{W}$
	inherit the bounds on the homomorphism densities:
	for every $\tilde{\cH}$-graph $F^{\tilde{g}}$, 
we have $\vert t(F^{\tilde{g}}, \tilde U) - t(F^{\tilde{g}}, \tilde W) \vert \leq 2^{-k - n_0 k^2}$.

Define for all $n\in [n_0]$ the function $\tilde{f}_n : s \mapsto \ind_{\{s_n = 1 \}}$,
and let $\tilde{\F}$ be the concatenation of $(\tilde{f}_n)_{n\in [n_0]}$ and $\tilde{\cH}$,
in particular $\tilde{\F}$ is a finite convergence determining sequence on $\MeasX{\{0,1\}^{n_0}}$.
Finally, as $\ddcutFtilde(\tilde{U},\tilde{W})$ upper bounds the first term
in the upper bound of \eqref{eq_ICL_1},
applying Lemma~\ref{InverseCountingLemmaWithFiniteSpace} 
	with the finite space $\Space = \{0,1\}^{n_0}$
    and $n_1 = 2^{n_0}$,
	the finite convergence determining sequences $\tilde{\F}$ and $\tilde{\cH}$,
	and the $\ProbaX{\{0,1\}^{n_0}}$-valued graphons $\tilde{U}$ and $\tilde{W}$, 
we get:
\[ \ddcutF(U,W) \leq \frac{44}{\sqrt{\ln(k)}}  + 2^{-n_0} ,\]
which concludes the proof.
\end{proof}

\subsection{Subgraph sampling and the topology of probability-graphons}

Thanks to the Weak Counting Lemma~\ref{WeakCountingLemma}
and the Inverse Counting Lemma~\ref{InverseCountingLemma},
we can formulate a new informative characterization of 
weak isomorphism, \ie equality in the space
of unlabeled probability-graphons $\UGraphon$.

\begin{prop}[Characterization of equality for $\ddcut$]
	\label{prop_equiv_dist_dens_homo}
Let $U,W\in\Graphon$ be two probability-graphons.
The following properties are equivalent:
\begin{enumerate}[label=(\roman*)]
\item\label{equiv_dist_dens_homo_1} 
$\ddcut(U,W) =0$ for some (and hence for every) choice 
of the distance $\dmeas$ on $\SubProba$
such that the cut distance $\dcut$ on $\Graphon$
is (invariant) smooth.

\item\label{equiv_dist_dens_homo_2}
There exist $\varphi, \psi \in \Relabel$
such that $U^\varphi = W^\psi$ almost everywhere
	on $[0,1]^2$.

\item\label{equiv_dist_dens_homo_3}
$t(F^g,U)=t(F^g,W)$ 
for all $\CbFunct$-graph $F^g$.

\item\label{equiv_dist_dens_homo_4}
$t(F^g,U)=t(F^g,W)$ 
for all $\F$-graph $F^g$.

\end{enumerate}
\end{prop}

\begin{proof}
The equivalence between Properties~\ref{equiv_dist_dens_homo_1} 
and \ref{equiv_dist_dens_homo_2} is a consequence of
Proposition~\ref{thm_min_dist} on the cut distance.
\Cref{rem_dens_homo_inv} gives that Property~\ref{equiv_dist_dens_homo_2} 
	implies Property~\ref{equiv_dist_dens_homo_3}.
It is clear that Property~\ref{equiv_dist_dens_homo_3} 
	implies Property~\ref{equiv_dist_dens_homo_4}.
The Inverse Counting Lemma~\ref{InverseCountingLemma}
with the Weak Counting Lemma~\ref{WeakCountingLemma}
give that Property~\ref{equiv_dist_dens_homo_4} 
	implies Property~\ref{equiv_dist_dens_homo_1} 
	(with $\dmeas=\dmeasF$).
Hence, we have the desired equivalence.
\end{proof}

Thanks to the Weak Counting Lemma~\ref{WeakCountingLemma}
and the Inverse Counting Lemma~\ref{InverseCountingLemma},
we get the following characterization
of the topology induced by the cut distance $\ddcut$
on the space
of unlabeled probability-graphons $\UGraphon$
in terms of homomorphism densities

\begin{theo}[Characterization of the topology induced by $\ddcut$]
	\label{theo_equiv_topo_dens_homo}
Let $(W_n )_{n\in\N}$ and $W$ be unlabeled probability-graphons
	from $\UGraphon$.
The following properties are equivalent:
\begin{enumerate}[label=(\roman*)]
\item
$\lim_{n\to\infty}\ddcut(W_n,W) = 0$
for some (and hence for every) choice of the distance $\dmeas$ on $\SubProba$
such that $\dmeas$ induces the weak topology on $\SubProba$
and the cut distance $\dcut$ on $\Graphon$
is (invariant) smooth, weakly regular and 
regular \wrt the stepping operator.

\item
$\lim_{n\to\infty} t(F^g,W_n) = t(F^g,W)$
for all $\CbFunct$-graph $F^g$.

\item
$\lim_{n\to\infty} t(F^g,W_n) = t(F^g,W)$
for all $\F$-graph $F^g$.

\item
For all $k\geq 2$, the sequence of sampled subgraphs $(\bbG(k,W_n))_{n\in\N}$
	converges in distribution to $\bbG(k,W)$.
\end{enumerate}
\end{theo}

In particular, the topology induced by the cut distance $\ddcutF$
on the space of unlabeled probability-graphons 
$\UGraphon$
coincides with the topology generated by
the homomorphism densities functions
$W\mapsto t(F^g,W)$
for all $\CbFunct$-graph $F^g$.

\begin{proof}
  By   Theorem~\ref{theo_equiv_topo},  convergence   for  $\ddcutF$   is
  equivalent  to  convergence  for  $\ddcut$ for  every  choice  of  the
  distance $\dmeas$ on  $\SubProba$ such that $\dmeas$  induces the weak
  topology on $\SubProba$ and the  cut distance $\dcut$ on $\Graphon$ is
  (invariant)  smooth,  weakly regular  and  regular  \wrt the  stepping
  operator.     Taking     $\dmeas=\dmeasP$,    the     Weak    Counting
  Lemma~\ref{WeakCountingLemma}                gives                that
  Property~\ref{equiv_topo_dens_homo_1}                          implies
  Property~\ref{equiv_topo_dens_homo_2}.      It    is     clear    that
  Property~\ref{equiv_topo_dens_homo_2}                          implies
  Property~\ref{equiv_topo_dens_homo_3}.     The     Inverse    Counting
  Lemma~\ref{InverseCountingLemma}     with     the    Weak     Counting
  Lemma~\ref{WeakCountingLemma}                 give                that
  Property~\ref{equiv_topo_dens_homo_3}                          implies
  Property~\ref{equiv_topo_dens_homo_1} (with $\dmeas=\dmeasF$).  Notice
  that when $F$ is the complete  graph with $k$ vertices, $M_W^F$ is the
  joint measure of all the edge-weights of the random graph $\bbG(k,W)$,
  and  thus characterizes  the  distribution  random graph  $\bbG(k,W)$.
  Thus              (remind              Definition~\ref{def_hom_dens}),
  Property~\ref{equiv_topo_dens_homo_2}                              and
  Property~\ref{equiv_topo_dens_homo_4} are equivalent.   Hence, we have
  the desired equivalence.
\end{proof}

\begin{quest}[Do the distances $\dcutF$ all induce the same topology?]
Even though every distance $\ddcutF$ generates the same
topology on the space of unlabeled probability-graphons 
$\UGraphon$,
it is an open question whether or not this is also the case that
every distance $\dcutF$ induces the same topology
on the space of labeled probability-graphons $\Graphon$.
\end{quest}

The following proposition states that to prove
existence of a limit unlabeled probability-graphon
it is enough to prove that there exists
a convergence determining sequence $\F$
such that for every $\F$-graph $F^g$
the homomorphism densities $t(F^g,\cdot)$ converge.

\begin{prop}[Existence of a limit unlabeled probability-graphon]
	\label{prop_existence_limit_graphon}
Let $\dmeas$ be a distance on $\SubProba$
such that $\dmeas$ induces the weak topology
	on $\SubProba$
and the cut distance $\dcut$ on $\Graphon$
is (invariant) smooth, weakly regular and 
regular \wrt the stepping operator.

Let $(W_n)_{n\in\N}$ be sequence of unlabeled probability-graphons in $\UGraphon$ that is tight.
Let $\F$ be a convergence determining sequence
such that for every $\F$-graph $F^g$
the sequence $(t(F^g,W_n))_{n\in\N}$ converges.
Then, there exists an unlabeled probability-graphon $W\in\UGraphon$
such that the sequence $(W_n)_{n\in\N}$ converges to $W$ for $\ddcut$.
\end{prop}

\begin{proof}
Since the sequence $(W_n)_{n\in\N}$ is tight,
by \Cref{theo_tension_conv}, there exists
a subsequence $(W_{n_k})_{k\in\N}$ of 
the sequence $(W_n)_{n\in\N}$ that converges to some $W$ for $\ddcut$.
By Theorem~\ref{theo_equiv_topo_dens_homo},
we have for every $\F$-graph $F^g$
that $\lim_{k\to\infty} t(F^g,W_{n_k}) = t(F^g,W)$;
and as we already know that the sequence
	$(t(F^g,W_n))_{n\in\N}$ converges,
we have that $\lim_{n\to \infty} t(F^g,W_n) = t(F^g, W)$.
Hence, by \Cref{theo_equiv_topo_dens_homo}, we get
that the sequence $(W_n)_{n\in\N}$ converges to $W$ for $\ddcut$.
\end{proof}

\begin{remark}
For the special case $\Space = \{0,1\}$, which is compact,
we find back that convergence for real-valued graphons
is characterized by the convergence of the homomorphism densities.
Notice the tightness condition of \Cref{prop_existence_limit_graphon}
is automatically satisfied as $\Space$ is compact.
\end{remark}

\section{Proofs of Theorem~\ref{theo_tension_conv} and Theorem~\ref{theo_equiv_topo}}
	\label{section_proof_theo_tension_conv}

We start by proving a lemma that allows to construct 
a convergent subsequence and its limit kernel
for a tight sequence of measure-valued kernels.
This lemma is useful for the proofs of both 
Theorem~\ref{theo_tension_conv} and Theorem~\ref{theo_equiv_topo}.
For the proof of Theorem~\ref{theo_equiv_topo},
we will also need the convergence to hold simultaneously 
for two distances $\dd$ and $\ddPrime$.
Remind from \Cref{def_stepping_operator} the definition of the stepfunction $W_{\cP}$
for a signed measure-valued kernel $W$ and a finite partition $\cP$ of $[0,1]$.
For a finite partition $\cP$ of $[0,1]$, define its diameter as the smallest diameter of its sets,
\ie $\diam(\cP) = \min_{S\in\cP} \diam(S) = \min_{S\in\cP} \sup_{x,y\in S} \vert x-y \vert$.

\begin{lemme}[Convergence using given approximation partitions]
	\label{lemme_partitions_convergence}
  Let  $d$ be  an invariant  smooth 
  distance  on
  $\Graphon$ (resp. $\Kernelp$ or $\Kernel$).
Let $(W_n )_{n\in\N}$ be a sequence in   $\Graphon$ 
(resp. $\Kernelp$ or $\Kernel$)
which is tight (resp. uniformly bounded and tight). 
Further assume that we are given, for every $n,k\in\N$,
partitions $\mathcal{P}_{n,k}$ of $[0,1]$,
such that these partitions and 
the corresponding stepfunctions $W_{n,k} = (W_n)_{\mathcal{P}_{n,k}}$
satisfy the following conditions:
\begin{enumerate}[label=(\roman*)]
\item \label{prop_Wnk_1} 
the partition $\mathcal{P}_{n,k+1}$ is a refinement of $\mathcal{P}_{n,k}$,
\item \label{prop_Wnk_2} 
$\diam(\cP_{n,k}) \leq 2^{-k}$ and $|\mathcal{P}_{n,k}| = m_k$
	depends only on $k$ (and not on $n$),
\item \label{prop_Wnk_3} 
$d(W_n, W_{n,k}) \ \leq\ 1/(k+1)$.
\end{enumerate}
Then, there exists a subsequence $(W_{n_\ell} )_{\ell\in\N}$ of the sequence $(W_n )_{n\in\N}$
and a measure-valued kernel  $W\in\Graphon$ (resp. $W\in \Kernelp$
or $W\in\Kernel$) such that
$(W_{n_\ell} )_{\ell\in\N}$ converges to $W$ for $\dd$.
\medskip

Moreover, assume that $\dPrime$ is another  invariant  smooth  
  distance  on
    $\Graphon$ (resp. $\Kernelp$ or $\Kernel$) such that  for every $n\in\N$ and
    $k\in \N$, $W_{n,k}$ also satisfies: 
\begin{enumerate}[label=(\roman*)]
\setcounter{enumi}{3}
\item \label{prop_Wnk_4} 
$\dPrime(W_n, W_{n,k}) \ \leq\ 1/(k+1)$.
\end{enumerate}
Then, there exists a subsequence $(W_{n_\ell} )_{\ell\in\N}$ of the sequence $(W_n )_{n\in\N}$
and a measure-valued kernel  $W\in\Graphon$ (resp. $W\in \Kernelp$ or $W\in\Kernel$) such that
$(W_{n_\ell} )_{\ell\in\N}$ converges to the same measure-valued kernel $W$ 
simultaneously both for $\dd$ and for $\ddPrime$, the cut
distance associated with $\dPrime$.
\end{lemme}

\begin{proof}
We adapt here the general scheme from the proof of Theorem~9.23 in \cite{Lovasz},
but the argument for the convergence of the $U_k$, defined below, 
takes into  account that  measure-valued
kernels are infinite-dimensional valued.
We set (remind from \eqref{eq:def:TM} the definition of $\TM{\cdot}$):
\[
  C=\sup_{n\in \N} \TM{W_n}<+\infty .
\]
The proof is divided into four steps.

\textbf{Step 1: Without loss of generality, the partitions
  $\mathcal{P}_{n,k}$ are made of intervals.} 
For every $n\in \N$,
we can rearrange the points of $[0,1]$ by a measure-preserving map 
so that the partitions $\mathcal{P}_{n,k}$ are made of intervals,
and we replace $W_n$ by its rearranged version.

An argument similar to the next lemma  
is used in the proof in \cite[Proof of Theorem 9.23]{Lovasz} without any reference.
So, we provide a proof
and stress that   diameters of the partitions shrinking to zero
is an important assumption
(see \Cref{rem_shrinking_diameter} below).

\begin{lemme}[Kernel rearrangement with interval partitions]
	\label{lemma_part_intervals}
Let $(\cP_k)_{k\in\N}$ be a refining sequence of finite partitions of $[0,1]$
	whose diameter converges to zero.
Then, there exist a measure-preserving map $\varphi\in\Relabel$ and a
refining sequence of partitions made of intervals  $(\cQ_k)_{k\in\N}$ 
such that  for all $k\in\N$, and all set $S\in\cP_k$
	there exists a set $R\in\cQ_k$ such that \aE $\ind_{R} = \ind_{\varphi^{-1}(S)}$.
\end{lemme}

In particular, if $W$ is a signed measure-valued kernel,
then for $U=W^\varphi$,
we have that \aE  
$U_{\cQ_k} = (W^\varphi)_{\cQ_k} = (W_{\cP_k})^\varphi$
for all $k\in\N$.

Notice that, according to \Cref{rem_refining_partitions},
the sequence  of refining partition $(\cP_{k})_{k\in\N}$,
  with a partition diameter converging to $0$, separates points and thus
  generates the Borel $\sigma$-field of $[0,1]$.

\begin{proof}
Consider the infinite Ulam-Harris tree $\cT^\infty = \{ u_1 \cdots u_k : k\in\N, u_1, \dots, u_k \in \N^* \}$, 
where for $k=0$ the empty word $u=\rooot$
	is called the root node of the tree;
for a node $u = u_1 \cdots u_k \in \cT^\infty$ ,
we define its height as $h(u)=k$,
and if $k>0$ we define its parent node as $p(u) = u_1 \cdots u_{k-1}$
	and we say that $u$ is a child node of $p(u)$.
We order vertices on the tree $\cT^\infty$ with the lexicographical (total) order
$\linflex$.
As a first step, we construct a subtree $\cT\subset\cT^\infty$
	that indexes the sets in the partitions $(\cP_k)_{k\in\N}$,
	such that for every $k\in\N$,
		$\cP_k = \{ S_u : u\in\cT, h(u)=k \}$,
	and such that if $S_v \subset S_u$ with $S_v\in\cP_{k}$ and $S_u\in\cP_{k-1}$,
		then $p(v) = u$.

Without loss of generality,  we may assume that $\cP_0 = \{ [0,1] \}$,
and we label its only set by the empty word $\rooot$, and we set $S_{\rooot} = [0,1]$.
Then, suppose we have already labeled the sets from $\cP_0, \dots, \cP_k$,
and we proceed to label the sets from $\cP_{k+1}$.
Because the partition $\cP_{k+1}$ is a refinement of $\cP_k$,
we can group the sets of $\cP_{k+1}$ by their unique parent set from $\cP_k$,
\ie for every $S_u\in\cP_k$, let $\cO_u = \{ S\in\cP_{k+1} : S \subset S_u \}$,
then $S_u = \cup_{S\in\cO_u} S$.
For $S_u\in\cP_k$, we fix an arbitrary enumeration of $\cO_u = \{ S^1, \dots, S^\ell \}$
	with $\ell = \vert \cO_u \vert$,
then label the set $S^j$ by $uj$, and set $S^j = S_{u j}$;
remark that the parent node of $w=u j$ is $p(w)=u$,
and the height of node $w$ is $h(w) = h(u) + 1 = k+1$.
Hence, we have labeled every set from $\cP_{k+1}$.
To finish the construction, we set 
	$\cT = \{ u : \exists k\in\N, \exists S\in\cP_k,\ S \text{ has label } u \}$.
\medskip

We now proceed to construct a measure-preserving map $\psi$
such that the image of every set $S_u$ is \aE equal to an interval,
and such that those intervals are ordered \wrt to the order of their labels in $\cT$.

Define the map $\sigma : [0,1] \to \cT^\N$ by
$\sigma(x) = (u^k(x) )_{k\in\N} \in \cT^\N$ 
where $u^k(x)$ is the only node of $\cT$ with height $k$
such that $x\in S_{u^k(x)}$
(and thus $u^{k+1}(x)$ is a child node of $u^k(x)$).
Remark that if $u^{k_0}(x) \linflex u^{k_0}(y)$ for some $k_0\in\N$,
	then $u^{k}(x) \linflex u^{k}(y)$ for every $k\geq k_0$.
We extend naturally the total order $\linflex$ from $\cT$ to a the total order on $\cT^\N$:
	for $(u^k)_{k\in\N}, (v^k)_{k\in\N}\in\cT^\N$,
	$(u^k)_{k\in\N} \linflex (v^k)_{k\in\N}$
	if $u^{k_0} \linflex v^{k_0}$ where $k_0$ is the smallest $k$
		such that $u^k \neq v^k$.

For every $u\in\cT$, define:
\[ A^-(u) = \bigcup_{v \linflex u \, : \, h(v)=h(u)} S_v 
\quad \text{ and } \quad
A^+(u) = A^-(u) \cup S_u
, \]
and then define $C^-(u) = \lambda(A^-(u))$ and $C^+(u) = \lambda(A^+(u))$.
Now, define $\psi$ as, for $x\in [0,1]$:
\[ \psi(x) = \lambda(A^-(x))
\quad \text{ where }\quad
A^-(x) = \{ y\in [0,1] : \sigma(y) \linflex \sigma(x) \}
= \cup_{k\in\N} A^-(u^k(x))
. \]
Moreover, as the sequence of partitions $(\cP_k)_{k\in\N}$ 
	has a diameter that converges to zero, and thus separates points,
	the map $\sigma$ is injective.
Thus, we also have:
\[ \psi(x) = \lambda(A^+(x))
\quad \text{ where }\quad
A^+(x) = \{ y\in [0,1] : \sigma(y) \leqlex \sigma(x) \}
= A^-(x) \cup \{ x \}
. \]
Remark that both $A^-(x)$ and $A^+(x)$ are Borel measurable.

Remark that for every $k\in\N$, we have $A^-(u^k(x)) \subset A^-(x) \subset A^+(x) \subset A^+(u^k(x))$.
In particular, for every $u\in\cT$,
	we have $\psi(S_u) \subset [ C^-(u), C^+(u) ]$;
	however $\psi(S_u)$ is not necessarily an interval, 
	but we shall see that $\lambda(\psi(S_u)) = C^+(u) - C^-(u)$,
	\ie $\psi(S_u)$ is \aE equal to $[ C^-(u), C^+(u) ]$.
Remark that, as the sequence of partitions $(\cP_k)_{k\in\N}$ is refining, 
we get that $[C^-(u), C^+(u)] = \cup_{v \, : \, p(v)=u} [C^-(v), C^+(v)]$ for every $u\in\cT\setminus\{\rooot\}$.

As the diameter of the partitions $(\cP_k)_{k\in\N}$ converges to zero,
we have the following alternative formula for $\psi$:
\begin{equation*} 
\psi(x) 
 = \lim_{k \to \infty}   C^-(u^k(x))
 = \lim_{k \to \infty}   C^+(u^k(x)) .
\end{equation*}
For every $k\in\N$,
the map $x\mapsto C^-(u^k(x))$ is a simple function
	(constant on each $S\in\cP_k$ and takes finitely-many values),
and thus $\psi$ is measurable as a limit
	of measurable maps.

\medskip

We outline the rest of the proof.
We first prove that $\psi$ is measure-preserving.
Secondly, we prove that $\psi$ is \aE bijective
	and construct its \aE inverse map $\varphi$.
Thirdly, we prove that $(\varphi^{-1}(\cP_k))_{k\in\N}$ is a refining sequence of partitions.
And lastly, we approximate almost everywhere the sequence
	of partitions $(\varphi^{-1}(\cP_k))_{k\in\N}$
by a sequence of refining partitions composed of intervals.

\medskip

We now prove that $\psi$ is measure preserving.
Remark that $\psi(x)$ is a non-decreasing function of $\sigma(x)$ for 
the total relation order $\linflex$, 
\ie $\psi(y)\leq \psi(x)$ if and only if $\sigma(y) \leqlex \sigma(x)$.
Hence, $\psi^{-1}([0,\psi(x)]) = \{ y\in [0,1] : \sigma(y) \leqlex \sigma(x) \}$, 
and we have:
\[ 
\lambda( \psi^{-1}([0,\psi(x)]) )
= \lambda( \{ y\in [0,1] : \sigma(y) \leqlex \sigma(x) \} )
= \psi(x) . \]
Thus, to show that $\psi$ is measure preserving we just need to show that
$\psi([0,1])$ is dense in $[0,1]$.
For every $u\in\cT$, as $\psi(S_u) \subset [ C^-(u), C^+(u) ]$,
	we know that the interval $[ C^-(u), C^+(u) ]$ contains
	at least one point of the form $\psi(x)$.
Remark that for all $k\in\N$, we have $[0,1] = \cup_{u\in\cT : h(u) = k} [ C^-(u), C^+(u) ]$.
Hence, as $\lambda( [ C^-(u), C^+(u) ] ) = \lambda(S_u) \leq \diam(\cP_{h(u)})$
	for every $u\in\cT$,
	and as the diameter of the partitions $(\cP_k)_{k\in\N}$ converges to zero,
we know that each interval of positive length contains
	a point of the form $\psi(x)$ for some $x\in [0,1]$,
which implies that $\psi([0,1])$ is indeed dense in $[0,1]$.

\medskip

We now prove that $\psi$ is \aE bijective
	and construct its \aE inverse map $\varphi$.
Without loss of generality, assume that there is no set $S_u$ with null measure.
Consider two distinct elements $x,y\in [0,1]$ such that $\sigma(x) \linflex \sigma(y)$.
Assume that $\psi(x) = \psi(y)$, and let $N\in\N$ be 
	the last index $k$ such that $u^k(x) = u^k(y)$.
Then, for every $k>N$, we have $u^k(x) \linflex u^k(y)$,
which implies that 
$\psi(x) \leq C^+(u^k(x)) \leq C^-(u^k(y)) \leq \psi(y)$;
and thus $\psi(x) = \psi(y)  = C^+(u^k(x)) = C^-(u^k(y))$,
which in turn implies that there is no node of $\cT$
	between $u^k(x)$ and $u^k(y)$.
Remark that this situation is analogous to the terminating decimal
	versus repeating decimal situation.
Hence, we proved that there is no node between $u^{N+1}(x)$ and $u^{N+1}(y)$
and that for every $k>N$,
	$u^{k+1}(x)$ is the right-most child of $u^k(x)$,
	and $u^{k+1}(y)$ is the left-most child of $u^k(y)$
	(\ie $u^{k+1}(x) = u^{k}(x) \vert \cO_{u^{k}(x)} \vert$
		and $u^{k+1}(y) = u^{k}(y) 1$).
Remind that the map $\sigma$ is injective.
Putting all of this together, we get that the set 
	$\{ (x,y)\in [0,1] : \psi(x) = \psi(y),\, x<y \}$
	can be indexed by the nodes of $\cT$, and is thus at most countable.
Hence, the map $\psi$ is injective on a subset $D\subset [0,1]$ with measure one
	(indeed $[0,1]\setminus D$ is at most countable),
	and as $\psi$ is measure preserving, we get that $\psi(D)$ has measure one,
	and thus $\psi$ is bijective from $D$ to $\psi(D)$, that is, $\psi$ is \aE bijective.
We construct the map $\varphi$ as the inverse map of $\psi$ for $x\in\psi(D)$
	and $\varphi(x)=0$ for $x\in [0,1]\setminus\psi(D)$.
Without loss of generality, we assume that $0\not\in D$.
Thus, $\varphi$ is the \aE inverse map of $\psi$, that is, 
	$\varphi \circ \psi(x) = \psi \circ \varphi(x) = x$ for almost every $x\in [0,1]$.
	
We are left to prove that $\varphi$ is measurable and measure preserving.
As we saw that each point $z\in[0,1]$ as a pre-image 
	$\psi^{-1}(z) = \{ x\in[0,1] : \psi(x) =z \}$ at most countable
	(indeed of cardinal at most 2),
	thus \cite{purves} insures that $\psi$ is bimeasurable
	(\ie $\psi$ is (Borel) measurable and for all Borel set $B\subset [0,1]$,
		$\psi(B)$ is also a Borel set).
Let $B\subset ]0,1]$ be a Borel set. We have that
$ \varphi^{-1}(B) = \varphi^{-1}(B\cap D) = \psi(B\cap D)$ is a Borel set,
	where the first equality uses that $\varphi([0,1]) = D\cup \{0\}$,
	the second equality uses that $\psi$ is the inverse of $\varphi$ on $D$,
	and lastly we used that $\psi$ is bimeasurable.
We also have that $\varphi^{-1}(B\cup\{0\}) = \varphi^{-1}(B) \cup ([0,1]\setminus \psi(D))$ 
		is  a Borel set.
Moreover, we have:
	\[ \lambda(\varphi^{-1}(B)) 
		= \lambda(\psi(B\cap D)) 
		= \lambda( \psi^{-1}( \psi(B\cap D))) 
		= \lambda(B\cap D)
		= \lambda(B) , \]
	where we used that $\varphi^{-1}(B) = \psi(B\cap D)$ for the first equality,
	that $\psi$ is measure preserving for the second equality,
	that $\psi$ is bijective from $D$ to $\psi(D)$ for the third equality,
	and that $D$ has measure one for the last equality.
We also have:	
\[ \lambda(\varphi^{-1}(B\cup\{0\})) 
		= \lambda( \varphi^{-1}(B) ) + \lambda( [0,1]\setminus \psi(D) )
		= \lambda(B) = \lambda( B \cup\{0\}) , \]
	where we used that $\varphi^{-1}(B)\subset \psi(D)$ and $[0,1]\setminus \psi(D)$ are disjoint sets
		for the first equality,
	that $\lambda(\varphi^{-1}(B)) = \lambda(B)$ and that $\psi(D)$ has measure one
		for the second equality.
Hence, the map $\varphi$ is measurable and measure preserving.

\medskip

We now prove that $(\varphi^{-1}(\cP_k))_{k\in\N}$ is a refining sequence of partitions.
For $k\in\N$, as $\cP_k$ is a finite partition of $[0,1]$,
	we have that $\varphi^{-1}(\cP_k) = \{ \varphi^{-1}(S_u) \, : \, u\in\cT,\, h(u)=k \}$
	is also a finite partition of $[0,1]$.
Moreover, as $(\cP_k)_{k\in\N}$ is a refining sequence of  partitions, we get  that
the sequence of partitions $(\varphi^{-1}(\cP_k))_{k\in\N}$ is also refining.
Remark that the sets $\varphi^{-1}(S_u)$ are not necessarily intervals,
	they are intervals minus some at most countable sets
	(this is similar to the unit line minus the Cantor set). 

\medskip

To finish the proof, we are left to construct a refining sequence of
 partitions made of intervals $(\cQ_k )_{k\in\N}$ that agrees
almost everywhere with the refining sequence of partitions $(\varphi^{-1}(\cP_k))_{k\in\N}$.
For $u\in\cT$, 
define $R_u = [C^-(u), C^+(u)[$
	(and $R_u = [C^-(u), C^+(u)]$ if $u$ is the unique node
		such that $v\leqlex u$ for every $v\in\cT$ with $h(v)=h(u)$).
As $\psi$ is measure preserving, and as $\psi(S_u) \subset [C^-(u), C^+(u)]$
	with $\lambda(S_u) = C^+(u) - C^-(u)$, 
	we get that $\lambda([C^-(u), C^+(u)] \setminus \psi(S_u)) =0$.
As $\varphi$ is the \aE inverse map of $\psi$, we have that \aE
	$\ind_{\varphi^{-1}(S_u)} = \ind_{\psi(S_u)} = \ind_{[C^-(u), C^+(u)]} = \ind_{R_u}$,
\ie  $R_u$ agrees almost everywhere 
with  $\varphi^{-1}(S_u)$.
For $k\in\N$, define the finite partition $\cQ_k = \{ R_u : h(u)=k \}$.
Then, by definition of the sets $R_u$,
	the sequence of partitions $(\cQ_k)_{k\in\N}$ is refining.
This concludes the proof.
\end{proof}

\begin{remark}[The shrinking diameter assumption is important]
	\label{rem_shrinking_diameter}
Even if it is not stressed in \cite[Proof of Theorem 9.23]{Lovasz},
	the measure preserving map $\varphi$ (a fortiori an \aE inversible one)
	in Lemma~\ref{lemma_part_intervals}
	cannot be obtained without any assumption on the refining sequence
	of partitions $(\cP_k)_{k\in\N}$ (in our case, we assumed that their diameter
		 converges to zero).
Indeed consider the sequence of partitions where for every $k\in\N$,
$\cP_k$ is composed of the sets:
	\[ S_{k,j} = [j 2^{-k-1}, (j+1) 2^{-k-1}[\ \cup\ [1/2 + j 2^{-k-1}, 1/2 + (j+1) 2^{-k-1}[ , 
		\qquad 0\leq j< 2^{k}, \]
\ie $S_{k,j}$ is the union of two dyadic interval translated by $1/2$, 
(also add $1$ to the set $S_{k,0}$ to get a complete partition).
Then, for every $x\in[0,1/2[$, $x$ and $x+1/2$ belong to the same set of $\cP_k$
	for every $k\in\N$; in particular the diameter of the partitions $(\cP_k)_{k\in\N}$
	does not converge to zero.
By contradiction, assume there exist a measure preserving map $\varphi\in\Relabel$
and a sequence of interval partitions $(\cQ_k)_{k\in\N}$ such that
	for all $k\in\N$ and all set $S_{k,j}\in\cP_k$ with $0\leq j < 2^{k}$,
	there exists a interval set $I_{k,j}\in\cQ_k$
	such that \aE $\ind_{I_{k,j}} = \ind_{\varphi^{-1}(S_{k,j})}$.
In particular, the set $I_{k,j}$ must be an interval of length $2^{-k}$.
Hence, $\cQ_k$ is a dyadic partition with stepsize $2^{-k}$,
	and thus the diameter of the partitions $(\cQ_k)_{k\in\N}$ 
	converges to zero.
For every $x\in[0,1/2[$, 
	we get that $\diam( \varphi^{-1}(\{ x, x+1/2 \}) )
		\leq \diam(\cQ_k) = 2^{-k}$ for all $k\in\N$;
	this implies  that $\varphi^{-1}(\{ x, x+1/2 \})$ is a singleton,
	\ie either $x\not\in\varphi([0,1])$ or  $x+1/2 \not\in\varphi([0,1])$.
Hence, we have 
	$\lambda([0,1/2[ \cap \varphi([0,1])) = \lambda( [1/2,1[ \setminus \varphi([0,1]))$
	and $\lambda([0,1/2[ \setminus \varphi([0,1])) = \lambda( [1/2,1[ \cap \varphi([0,1]))$.
As $\lambda([0,1/2[) = \lambda([0,1/2[ \cap \varphi([0,1])) + \lambda([0,1/2[ \setminus \varphi([0,1])) 
	= 1/2$
	because $\varphi$ is measure preserving,
	we get that $\lambda(\varphi([0,1]) )
		= \lambda([0,1/2[ \cap \varphi([0,1])) + \lambda([1/2,1[ \cap \varphi([0,1]))
		= 1/2$,
	which contradicts the fact that $\varphi$ is measure preserving.	
\end{remark}

Now, for every $n\in\N$, applying Lemma~\ref{lemma_part_intervals} to 
	$(\cP_{n,k})_{k\in\N}$ and $W_n$,
we get a  measure-preserving map $\varphi_{n}$ and a refining sequence of partitions
$(\cP'_{n,k})_{k\in\N}$ made of intervals such that
for all $k\in\N$, and all set $R\in\cP'_{n,k}$
	there exists a set $S\in\cP_{n,k}$ such that \aE $\ind_{R} = \ind_{\varphi_n^{-1}(S)}$.
In particular, for all $k\in\N$, the sequence of partitions $(\cP_{n,k})_{k\in\N}$
	still satisfy \ref{prop_Wnk_1}--\ref{prop_Wnk_2}.
Set $W'_n = W_n^{\varphi_n}$ and $W'_{n,k}=W_{n,k}^{\varphi_{n}}$ so that
almost everywhere:
\[
W'_{n,k} = \left((W_n)_{\mathcal{P}_{n,k}}\right)^{\varphi_{n}}
=(W_n^{\varphi_{n}})_{\mathcal{P}'_{n,k}}
=(W')_{\mathcal{P}'_{n,k}}.
\]
As $d$ and $d'$ are invariant, we have for every $n,k\in\N$ that
$d(W_n, W_{n,k}) = d(W'_n, W'_{n,k})$, and similarly for $d'$.
This insures that the signed measure-valued kernels $(W'_n)_{n\in\N}$ 
	and $(W'_{n,k})_{n\in\N}$, $k\in\N$,
	still satisfy \ref{prop_Wnk_3}--\ref{prop_Wnk_4}.
Remind that for a measure-valued kernel $W$ and a measure-preserving map $\varphi$,
	$\ddcut(W,W^\varphi)=0$.
Hence, we can replace the signed measure-valued kernels $(W_n)_{n\in\N}$ 
	and $(W_{n,k})_{n\in\N}$, $k\in\N$,
	by $(W'_n)_{n\in\N}$ and $(W'_{n,k})_{n\in\N}$,  $k\in\N$,
	and assume that the partitions $\cP_{n,k}$ are made of intervals.

\medskip

\textbf{Step 2: There exists a subsequence $(W_{n_\ell} )_{\ell\in\N}$ such that 
for every $k\in \N$ and $\epsilon\in\{+,-\}$, 
the subsequence $(W_{n_\ell,k}^\epsilon )_{\ell\in\N}$ weakly converges, 
as $\ell \to\infty$, 
almost everywhere to a limit, say $U_k^\epsilon$ which is a stepfunction adapted
to a partition with $m_k$ elements
(some elements might be empty sets).}

Fix        some        $k\in       \N$.         The        stepfunctions
$(W_{n,k}=(W_n)_{\cp_{n,k}}  )_{n\in\N}$ all  have the  same number  of
steps       $m_k$.       For       $n\in\N$,      denote       by
$\mathcal{P}_{n,k} =  \{ S_{n,k, i}  : 1\leq i\leq m_k\}$  the interval partition
adapted  to  $W_{n,k}$ 
where the intervals are order according to the natural order on $[0,1]$
(note that some intervals might be empty, simply put them at the end).
For  $n\in\N$ and  $1\leq  i\leq  m_k$,  let $\lambda(S_{n,k,i})$ denote 
the  length of the  interval $S_{n,k,i}\in\mathcal{P}_{n,k}$.  
As  the lengths of steps  take values in
the  compact  set $[0,1]$,  there exists  
a subsequence of indices $(n_\ell)_{\ell\in\N}$
such  that  for  every  $1\leq i  \leq  m_k$,  there  exists
$s_{k,i}            \in            [0,1]$           such            that
$\lim_{\ell\rightarrow \infty } \lambda(S_{n_\ell,k,i}) =s_{k,i}$. 
Denote by $\cp_k = \{ S_{k, i} : 1\leq i\leq m_k\}$ the interval partition
composed of $m_k$ intervals where the $i$-th interval $S_{k,i}$
	has length $s_{k,i}$
	(note that some intervals might be empty).
Up to a diagonal extraction, we can assume that 
the convergence holds for every $k\in\N$ simultaneously.
Remark that for all $n,k\in\N$, 
	the fact that $\cP_{n,k+1}$ is a refinement of $\cP_{n,k}$
	can be simply restated as linear relations on the interval lengths 
	$(\lambda(S_{n,k,i}))_{1\leq i \leq m_k}$ and $(\lambda(S_{n,k+1,i}))_{1\leq i \leq m_{k+1}}$.
As linear relations are preserved when taking the limit, we get that 
the partition $\cp_{k+1}$ is a refinement
of $\cp_k$ for  all $k\in \N$.
We  assume from now on
that $(W_n )_{n\in\N}$ and $(W_{n,k} )_{n\in\N}$, $k\in\N$, are the corresponding
subsequences.

For every $n\in\N$, we decompose $W_n = W_n^+ - W_n^-$ 
into its positive and negative kernel parts, see Lemma~\ref{lem:mesurability_W_+}.
For $n,k\in\N$ and $\epsilon\in\{+,-\}$, 
	we define $W_{n,k}^\epsilon = (W_n^\epsilon)_{\cP_{n,k}}$.
In particular, remark that $W_{n,k} = W_{n,k}^+ - W_{n,k}^-$
	and    for   all    $\ell\geq    k$,    that
	$W_{n,k}^\epsilon = (W_{n,\ell}^\epsilon)_{\mathcal{P}_{n,k}}$.
Let $\epsilon\in\{+,-\}$ and $1\leq i,j \leq m_k$ such that $s_{k,i}s_{k,j}>0$ be fixed. 
For every $n\in\N$,
we have on $S_{n,k,i}\times S_{n,k,j}$ that
$W_{n,k}^\eps = \mu_{n,k}^{i,j,\epsilon} \in\Meas$ with:
\[
\mu_{n,k}^{i,j,\epsilon}(\cdot)
= \frac{1}{\lambda(S_{n,k,i})\lambda(S_{n,k,j})}
W_n^\epsilon(S_{n,k,i}\times  S_{n,k,j}; \cdot)
. \]
We have that:
\[
  \TotalMass{\mu_{n,k}^{i,j,\epsilon}}
  \leq  \TM{W_n} \leq C.
\]
This gives that  the sequence $(\mu_{n,k}^{i,j,\epsilon} )_{n\in\N}$ in $\SignedMeas$ is
 bounded. We now prove it is tight. 
Let $\eps>0$.
As $\lim_{n\rightarrow \infty } \lambda(S_{n,k,\ell}) =s_{k,\ell}>0$ for
$\ell=i,j$,
we deduce that there exists  $c>0$ such that for every $n\in
\N$ large enough and $\ell=i,j$, we have
$\lambda(S_{n,k,\ell}) > c$. 
Set $\eps' = c^2 \eps$.
As the sequence $( W_n )_{n\in\N}$ in $\UKernel$ is  tight,
there exists  a compact set $K \subset \Space$ such that
for every $n\in\N$, $M_{W_n}(K^c) \leq   \eps'$.
Hence, for every $n\in \N$ large enough, we have:
\begin{equation*}
   \mu_{n,k}^{i,j,\epsilon} (K^c)\leq  
  \frac{1}{\lambda(S_{n,k,i})\lambda(S_{n,k,j})} M_{W_n}(K^c)
  \leq  \varepsilon.
\end{equation*}
This gives that  the sequence $(\mu_{n,k}^{i,j,\epsilon} )_{n\in\N}$  in $\Meas$
	is bounded and tight, and thus by Lemma~\ref{lemme_Bogachev_Prohorov_theorem},
	it has a convergent subsequence.
By diagonal extraction, we can assume there is a subsequence  $(W_{n_\ell})_{\ell\in\N}$ 
such that for all $k\in \N$, all $1\leq i, j\leq  m_k$ such
that $s_{k,i}s_{k,j}>0$, and all $\epsilon \in \{+,-\}$,
  the subsequence $(\mu_{n_\ell,k}^{i,j,\epsilon} )_{\ell\in\N}$
weakly converges to a limit say $\mu_{k}^{i,j,\epsilon}$. 
Define the stepfunction $U_k^\epsilon\in \Kernelp$ adapted to the partition $\cp_k$
which is  equal to $\mu_{k}^{i,j,\epsilon}$ on $S_{k,i}\times S_{k,j}$
(if $s_{k,i}s_{k,j}=0$, set $\mu_{k}^{i,j,\epsilon} = 0$). 
We have in particular obtained that, for every $k\in \N$, the subsequence
$(W_{n_\ell,k}^\epsilon )_{\ell\in\N}$
 weakly converges \aE to $U_k^\epsilon$ which is 
a stepfunction adapted to a partition with $m_k$ elements;
this implies that the subsequence
$(W_{n_\ell,k} )_{\ell\in\N}$ also weakly converges \aE to $U_k = U_k^+ - U_k^-$.
We now assume that $(W_{n})_{n\in\N}$ is such a subsequence. 
With this convention, notice that for all $k, n\in \N$ and $\epsilon \in \{+,-\}$:
\begin{equation}
  \label{eq:Uk-bded}
  \TM{U_k^\epsilon} \leq \sup_{n\in\N} \TM{W_{n,k}^\epsilon}\leq  \sup_{n\in\N} \TM{W_n} =
C<+\infty .
\end{equation}

\textbf{Step 3: There exists a subsequence
of   $(U_k )_{k\in\N}$ which weakly
  converges to a limit $U\in \Kernel$  almost everywhere on $[0,1]^2$.}
The proof of this step is postponed to the end. Without loss of
generality we still write  $(U_k )_{k\in\N}$ for this subsequence. 

\medskip

\textbf{Step 4: We have $\lim_{n\rightarrow \infty } \dd (U, W_n)=
  \lim_{n\rightarrow \infty } \ddPrime (U, W_n)=0$.}
Let  $\eps > 0$.
As the cut distances $d$ is smooth, we deduce from Step 3, that for $k$ large enough
$d(U, U_k) \leq  \varepsilon$. By hypothesis~\ref{prop_Wnk_3} on the sequence $(W_{n,k})_{n\in \N}$, 
we also have that for $k$ large enough
$d(W_n, W_{n,k}) \leq  \varepsilon$. 
For such large $k$, as by step 2 the sequence $(W_{n,k} )_{n\in\N}$ 
	weakly converges almost everywhere to $U_k$,
and again as the cut distances $d$ is smooth, 
there is a $n_0$ such that for every $n\geq n_0$,
$d(U_k, W_{n,k}) \leq  \eps$.
Then for all $n\geq n_0$,
we have:
\begin{align*}
\dd(U,W_n) 
&  \leq   \dd(U,U_k) + \dd(U_k, W_{n,k}) + \dd(W_{n,k},W_n)  \\
&  \leq   d(U,U_k) + d(U_k, W_{n,k}) + d(W_{n,k}, W_n) \\
&  \leq   3\varepsilon.
\end{align*}
This gives that $\lim_{n\rightarrow \infty } \dd (W_n, U)=0$.

If we consider a second distance $\dPrime$
as in the lemma, then similarly $\lim_{n\rightarrow \infty } \ddPrime (W_n, U)=0$.

\medskip

\textbf{Proof of Step 3.} 
Assume that the claim is true for measure-valued kernels.
Then, if $(U_k)_{k\in\N}$ is a sequence of signed-measure valued kernels,
	applying the claim to $(U_k^\epsilon)_{k\in\N}$, for $\epsilon\in\{+,-\}$,
	we get a measure-valued $U^\epsilon\in\Kernelp$ such that
	the sequence $(U_k^\epsilon)_{k\in\N}$ weakly conveges \aE to $U^\epsilon$.
Thus, the sequence $(U_k)_{k\in\N}$ weakly conveges \aE to $U = U^+ - U^-$.

Hence, we are left to prove the claim for measure-valued kernels.
The proof is divided in four steps.
The first three steps also work for signed-measure valued kernels,
	but the last argument of step 3.d.\ only works for measures.

\textbf{Step 3.a: The sequence $(U_k )_{k\in\N}$ 
inherit the tightness property from the sequence $(W_n )_{n\in\N}$.}
Let $\eps > 0$. Since the sequence  $(W_n )_{n\in\N}$ is tight,
there exists  a compact set $K\subset \Space$ such that for every $n\in\N$, 
we have $M_{W_n}(K^c) \leq  \eps$. 
Remark that:
\begin{equation*}
M_{W_{n,k}} = \sum_{1\leq i,j \leq m_k} 
		{\lambda(S_{n,k,i})\lambda(S_{n,k,j})}	 \mu_{n,k}^{i,j} 
	= M_{W_n}
	\quad \text{ and }\quad 
M_{U_{k}} = \sum_{1\leq i,j \leq m_k} 
		{s_{k,i}s_{k,j}}	 \mu_{k}^{i,j}  .
\end{equation*}
For all $k\in\N$ and $1\leq i,j \leq m_k$, 
as the sequence $(\mu_{n,k}^{i,j})_{n\in\N}$ weakly converges to $\mu_{k}^{i,j}$,
using \cite[Theorem 2.7.4]{Bogachev} with the open subset $K^c \subset \Space$,
we get that
$  \mu_{k}^{i,j}  (K^c) \leq
	\liminf_{n\rightarrow \infty }	 \mu_{n,k}^{i,j}  (K^c)$.
As $\lim_{n\to\infty} \lambda(S_{n,k,i}) = s_{k,i}$ for all $1\leq i\leq m_k$,
and summing those bounds, we get:
\[
  M_{U_k}(K^c) \leq \liminf_{n\rightarrow \infty } M_{W_{n,k}}(K^c) 
  = \liminf_{n\rightarrow \infty } M_{W_{n}}(K^c) 
  \leq \eps.
\]
Consequently, the sequence $(U_k )_{k\in\N}$ is  tight.

\medskip

\textbf{Step 3.b: Convergence of the  measures 
$\hat U_k $} in $\mathcal{M}_{+}([0,1]^2\times \Space)$ defined for $k\in \N$
as:
\begin{equation}
   \label{eq:def-hat-U}
\hat   U_k(\drv x,\drv y, \drv z) = U_k(x,y;\drv z)\, \lambda_2(\drv x, \drv y).
\end{equation}
Since  the sequence  $(M_{U_k})_{k\in  \N}$   in  $\Meas$ is  tight,   for  all
$\varepsilon>0$, there exists a compact  set $K\subset \Space$ such that
for every $k\in\N$, $M_{U_k}(K^c)         \leq         \varepsilon$;        and         thus
$ \hat U_k (\hat K^c) = M_{U_k}(K^c) \leq \varepsilon$ where $\hat K=[0, 1]^2 \times K$ is
a   compact subset  of   $[0,1]^2\times   \Space$,  that   is,   the   sequence
$(\hat      U_k  )_{k\in       \N}$   in
$\mathcal{M}_{+}([0,1]^2\times \Space)$   is      tight.  
The   sequence $(\hat U_k)_{k\in\N}$ is  also  bounded as
$\TotalMass{\hat      U_k}\leq     \TM{U_k}      \leq     C$      thanks
to~\eqref{eq:Uk-bded}.    Hence, using Lemma~\ref{lemme_Bogachev_Prohorov_theorem},
  there   exists   a   subsequence $(\hat  U_{k_\ell} )_{\ell\in\N}$  of the sequence
$(\hat U_k )_{k\in\N}$ that converges to some  measure, say
$\hat U$, in $\mathcal{M}_{+}([0,1]^2\times  \Space)$.  
Remark that, when considering the   subsequence  of   indices  $(k_\ell)_{\ell\in\N}$,  
the   subsequences
$(W_{n,k_\ell})_{\ell\in\N}$, $n\in\N$,   still   satisfy   properties
\ref{prop_Wnk_1}-\ref{prop_Wnk_4} of
Lemma~\ref{lemme_partitions_convergence},
and  for    all     $\ell\in\N$,
the sequence
$(W_{n,k_\ell})_{n\in\N}$ still weakly converges to $U_{k_\ell}$. 
Without loss  of generality, we now work with this subsequence 
and thus write $k$ instead of $k_\ell$.

\medskip

\textbf{Step 3.c: The measure $\hat U (\rd x, \rd y, \rd z)$ can be disintegrated
\wrt $\lambda_2(\rd x, \rd y)$ giving us an element of $\Kernelp$.}
To prove this, we need the following disintegration theorem for measures,
	see \cite[Theorem 1.23]{KallenbergRMTA} (stated in more the general
        framework of Borel spaces) which generalizes 
	the disintegration theorem for probability measures \cite[Theorem 6.3]{Kallenberg}.
 The notation $\mu \sim \nu$ for two measures $\mu$ and $\nu$
 means that $\mu \ll \nu$ and $\nu \ll \mu$,
 where $\mu \ll \nu$ means that $\mu$ is absolutely continuous
 \wrt $\nu$.

\begin{lemme}[Disintegration theorem for measures, {\cite[Theorem 1.23]{KallenbergRMTA}}]
	\label{disintegration_theorem_positive_measure}
Let $\rho$ be a measure on $S\times T$ , where $S$ is
a measurable space and $T$  a Polish space. 
Then there exist a measure $\nu\equiv \rho(\cdot \times T)$ on $S$ 
	and a probability kernel $\mu:S\to \ProbaX{T}$ 
such that $\rho = \nu \otimes \mu$ (\ie $\rho(ds,dt)=\nu(ds)\mu(s;dt)$).
Moreover, the measures $\mu_s = \mu(s;\cdot)$ are unique for $\nu$-\aE $s \in S$.
\end{lemme}

Using  Lemma~\ref{disintegration_theorem_positive_measure} with $S=[0,1]^2$ and $T=\Space$,
we get that there exists a probability kernel $U'$ in $\Graphon$ such that:
\[
  \hat U(\drv x,\drv y , \drv z) = U'(x,y;\drv z)\, \pi(\drv x,\drv y),
\]
where $\pi =  \hat U  (\cdot \times \Space)$
is a measure on $[0,1]^2$.

We now need to prove that $\pi \ll \lambda_2$.
By contradiction, assume this is false,
then there exists a measurable set $A\in \BorelX{[0,1]^2}$
	such that $\lambda_2(A) = 0$ and $\pi(A)>0$.
As the measure $\int_A U'(x,y;\cdot)\, \pi(\drv x, \drv y)$ is not null,
	there exists $f\in\CbFunct$ such that 
		$\int_A U'(x,y;f)\, \pi(\drv x, \drv y) \neq 0$.
As the sequence $(\hat U_k)_{k\in\N}$ weakly converges to $\hat U$ 
	in $\MeasX{[0,1]^2\times \Space}$ by step 3.b,
we have that the sequence of measures 
$\hat U_k(\drv x,\drv y; f) = U_k(x,y;f)\, \lambda_2(\drv x, \drv y)$
weakly converges as $k\to\infty$  to 
$\hat U(\drv x,\drv y; f) = U'(x,y;f)\, \pi(\drv x, \drv y)$
in $\MeasX{[0,1]^2}$.
Moreover, as the maps $x,y \mapsto U_k(x,y;f)$ are uniformly bounded 
(by $\norm{f}_\infty \norm{U_k}_\infty \leq C \norm{f}_\infty$, see~\eqref{eq:Uk-bded}),
they are also uniformly integrable (\wrt $\lambda_2$),
and applying \cite[Corollary 4.7.19]{BogachevMT1}
there exist a subsequence $(U_{k_\ell})_{\ell\in\N}$ and
a bounded function $g_f$ on $[0,1]^2$
such that for every bounded measurable function $h\in L^\infty([0,1]^2)$, we have:
\[ \lim_{\ell\to\infty} \int U_{k_\ell}(x,y;f) h(x,y)\, \lambda_2(\drv x, \drv y)
	= \int g_f(x,y)  h(x,y)\, \lambda_2(\drv x, \drv y) .\]
In particular, the sequence of measures
	$( U_{k_\ell}(x,y;f) \, \lambda_2(\drv x, \drv y) )_{\ell\in\N}$
	weakly converges to the measure 
		$g_f(x,y)\, \lambda_2(\drv x, \drv y)$,
	which imposes the equality between measures:
		\[ \hat U(\drv x,\drv y, f) = U'(x,y;f) \pi(\drv x, \drv y) 
				= g_f(x,y)\, \lambda_2(\drv x, \drv y) . \]
Hence, taking $h = \ind_A$, we get:
\[
\hat U(A, f) = \int_A U'(x,y;f) \pi(\drv x, \drv y)
= \int_A g_f(x,y) \lambda_2(\drv x, \drv y) = 0 ,
\]
	which yields a contradiction.
Consequently, 
the measure $\pi$ 
is absolutely continuous  \wrt $\lambda_2$, with
density  still  denoted  by  $\pi$,  and  we  set  $\lambda_2$-\aE  on
$[0,1]^2$:
\begin{equation}
   \label{eq:hatU=U-Leb}
  U(x,y; \rd z)=\pi(x,y)\,  U'(x,y;\rd z)
  \quad\text{and thus}\quad
 \hat U(\drv x,\drv y , \drv z) = U(x,y;\drv z)\, \lambda_2(\drv x,\drv y).
\end{equation}

\medskip

\textbf{Step 3.d: The sequence $(U_k )_{k\in\N} $ weakly converges  to $U$ almost
everywhere on $[0, 1]^2$.}  Recall that  by construction, the stepfunction $U_k$
is  adapted  to the  partition  $\cp_k$  defined  in  Step 2,  and  that
$\cp_{k+1}$  is  a refinement  of  $\cp_k$.  A  closer  look at  Step  2
yields    that    for   all    $\ell\geq    k$,    since
$W_{n,k} = (W_{n,\ell})_{\cP_{n,k}}$, we also get:
\begin{equation}\label{eq_tens_mart_stepfunction}
 U_k = (U_\ell)_{\cP_{k}} .
\end{equation}
We prove~\eqref{eq_tens_mart_stepfunction}
for $\ell = k+1$, the other cases follow by induction.
As $\cP_{n,k+1}$ is a refinement of $\cP_{n,k}$,
	we already know that $U_k$ and $(U_{k+1})_{\cP_k}$
	are both stepfunctions adapted to the finite partition $\cP_k$.
Thus, we only need to verify that $U_k$ and $(U_{k+1})_{\cP_k}$
	take the same value on each step.
For every $n\in\N$, the fact that $W_{n,k} = (W_{n,k+1})_{\cP_{n,k}}$
implies that for all $1\leq i,j \leq m_k$ 
	such that $\lambda(S_{n,k,i}) \lambda(S_{n,k,j})>0$,
	we have:
\[ \mu_{i,j}^{n,k} = \sum_{i' \in I_i, j'\in I_j} 
			\frac{\lambda(S_{n,k+1,i'}) \lambda(S_{n,k+1,j'})}
				{\lambda(S_{n,k,i}) \lambda(S_{n,k,j})}
			\mu_{i',j'}^{n,k+1} 
, \]
and this equation is preserved when taking the limit $n\to\infty$, which gives us:
\[ \mu_{i,j}^{k} = \sum_{i' \in I_i, j'\in I_j} 
			\frac{s_{k+1,i'} s_{k+1,j'}}
				{s_{k,i} s_{k,j}}
			\mu_{i',j'}^{k+1} 
, \quad \text{ for all $1\leq i,j \leq m_k$ such that $s_{k,i} s_{k,j} >0$}
. \]
This proves that the stepfunctions $U_k$ and $(U_{k+1})_{\cP_k}$
take the same value on each step $S_{k,i} \times S_{k,j}$
	with positive size $s_{k,i} s_{k,j} > 0$
(on a step with null size $s_{k,i} s_{k,j} = 0$,
	$U_k$ and $(U_{k+1})_{\cP_k}$
	are both equal to the null measure).
This gives  that $ U_k = (U_{k+1})_{\cP_{k}}$.

Let  $f\in \CbFunct$  be a  bounded  continuous function,  and $X,Y$  be
independent    uniform    random     variables    on    $[0,1]$.     Then
\eqref{eq_tens_mart_stepfunction} and~\eqref{eq:Uk-bded} imply
that  the sequence  $N^f=(N_k^f=U_k(X,Y;f) )_{k\in\N}$  is a  
 martingale bounded by $C\norm{f}_{\infty}$  for the filtration $(\cf_k)_{k\in \N}$, 
 where the $\sigma$-field $\cf_k$ is generated  by the events $\{X\in S_{k,i}\}\cap\{Y\in S_{k, j}\}$ 
 for $1\leq  i, j \leq  m_k$ and $S_{k, \ell}\in\cp_k$. 
By the martingale convergence theorem, 
the martingale $N^f$ is almost surely convergent, that is, 
 the sequence $(U_k[f])_{k\in\N}$ converges
 $\lambda_2$-\aE  to a bounded measurable function $u_f$.
Let $g:[0,1]^2\to \R$ be a bounded measurable function.
We get:
\begin{align*}
  \int_{[0, 1]^2} g(x,y)\, U(x,y; f) \,  \lambda_2(\rd x \rd y)
&  =
 \int_{[0, 1]^2\times \Space} g(x,y)\,f(z) \, \hat U(\rd x, \rd y, \rd
 z)\\
&  = \lim_{k\rightarrow \infty }
 \int_{[0, 1]^2\times \Space} g(x,y)\,f(z) \, \hat U_k(\rd x, \rd y, \rd
 z)\\
&  = \lim_{k\rightarrow \infty }
 \E\left[g(X,Y) U_k(X,Y;f) \right]\\
& =\int_{[0, 1]^2} g(x,y)\, u_f(x,y)\, \lambda_2(\rd x \rd y), 
\end{align*}
where we used the definition~\eqref{eq:hatU=U-Leb}  of $U$ for the first
equality, that $(\hat  U_k )_{k\in\N}$ weakly converges  to $\hat U$
for the second, the definition~\eqref{eq:def-hat-U}  of $\hat U_k$ for the
third, and the convergence of the martingale $N^f$ for the last. Since $g$
is arbitrary,  we deduce that $\lambda_2$-\aE  $U(\cdot,\cdot; f)=u_f$
and   thus    that   the    sequence   $(U_k[f])_{k\in\N}$   converges
$\lambda_2$-\aE    to    $U[f]$.     Applying    this    result    for all
$f\in  \F  =  (f_m)_{m\in\N}$ a  convergence  determining  sequence 
(with the convention $f_0=\un$), we deduce
that  the sequence $(U_k )_{k\in\N} $ weakly converges  to $U$ almost
everywhere on $[0, 1]^2$. 
Remind from Section~\ref{section_notations}
that convergence determining sequences exist only for measures
	and not for signed measures in general,
	this is why we worked with measures in Step 3.
This ends the proof of Step 3, and thus ends the proof of the lemma.
\end{proof}

We are now ready to prove Theorem~\ref{theo_tension_conv}.

\begin{proof}[Proof of Theorem~\ref{theo_tension_conv}]
  We first  prove Point~\ref{it:tension-cv} on $\Kernel$  (the proof on
  $\Graphon$ is similar).  Since the  distance $d$ is weakly regular and
  the sequence  $( W_n  )_{n\in\N}$ is uniformly  bounded and  tight in
  $\Kernel$,  we  can  construct  inductively for  every  $n\in  \N$  a
  sequence $(\mathcal{P}_{n,k} )_{k\in\N}$   of partitions of $[0,1]$ such
  that          hypothesis~\ref{prop_Wnk_1}-\ref{prop_Wnk_3}         of
  Lemma~\ref{lemme_partitions_convergence} are satisfied:
  $\cP_{n,k+1}$ being obtained by applying the weak regularity property 
  	(see Definition~\ref{defi:extra-prop}-\ref{hypo_dist_WRL})
  	with starting partition $\cQ_{n,k} = \cP_{n,k} \land \cD_{k}$,
	where $\cD_k$ is the dyadic partition with stepsize $2^{-(k+1)}$.
  (We may assume that the partitions $P_{n,k}$ for all $n\in\N$
  	have the same size $m_k$ by adding empty sets.)
   Then as $d$ is
  also  invariant   and  smooth  on   $\Kernel$,  the  first   part  of
  Lemma~\ref{lemme_partitions_convergence} directly gives
  Point~\ref{it:tension-cv}.
  \medskip

Before proving Point~\ref{it:M-compact}, we first need to prove the following lemma.

\begin{lemme}[Compactness theorem for $\KernelM$]
	\label{lemma_M_convex_closed}
Let  $d$ be  an invariant,  smooth  and weakly  regular distance  on
    $\Graphon$ (resp. $\Kernelp$ or $\Kernel$).
Let $\cM$ be a convex and weakly closed subset of $\Proba$ (resp. $\Meas$ or $\SignedMeas$).
Let $(W_n)_{n\in\N}$ be a sequence of $\cM$-valued kernels
	which is tight and uniformly bounded.
Then, $(W_n)_{n\in\N}$ has a subsequence that converges for $\dd$
	to some $\cM$-valued kernel.
\end{lemme}

\begin{proof}
First remark that, as $\cM$ is convex, the image  of ${\cW} _\cm$ by  the stepping
  operator  $W \mapsto  W_\cp$, where  $\cp$  is a  finite partition  of
  $[0, 1]$,  is a subset of  ${\cW} _\cm$.
  Hence, a close look at the
  proof  of  Lemma~\ref{lemme_partitions_convergence}
  (the partitions are constructed as in the proof of Point~\ref{it:tension-cv} 
  		from Theorem~\ref{theo_tension_conv}),   
  and  using  the
  notation therein, shows that, up  to taking subsequences, one can take
  the  stepping kernels  $W_{n,k}$  and $U_k$  in  $\cW_\cm$, such  that
  $(U_k )_{k\in\N}$ weakly converges  to $U$ \aE and  the subsequence
  $(W_{n_\ell})_{\ell\in\N}$  converges  to  $U$   \wrt  $\dd$.   Since
  $U_k(x,y;\cdot) \in \cm$ weakly converges to $U(x,y;\cdot)$ for almost
  every  $x,y\in  [0,1]$ and  since  $\cm$  is weakly closed (and thus sequentially weakly closed), 
  we  deduce  that $U(x,y;\cdot)$ belongs to $\cm$ for almost every $x,y\in [0,1]$.  This
  means   that  $U\in   \cW_\cm$. 
\end{proof}

  We     prove     Point~\ref{it:M-compact} for $\cm\subset\SignedMeas$
  	(the proof for $\cm\subset \Proba$ is identical). 
  The     fact     that
  ${\cW} _\cm$ and  $\widetilde{\cW} _\cm$ are convex is  clear as $\cm$
  is convex.     Let  $(W_n )_{n\in\N}$  be a
  sequence  of  elements  of  $\widetilde{\cW} _\cm$.   Since  $\cm$  is
  convex,  we  deduce that  $(M_{W_n}  )_{n\in\N}$  is a  sequence  in
  $\cm$. As $\cm$ is sequentially compact for the weak topology,
  $\cm$ is tight and bounded by Lemma~\ref{lemme_Bogachev_Prohorov_theorem},
  and thus the sequence $(W_n )_{n\in\N}$ is tight and uniformly bounded
  	(remind Definition~\ref{defi:tight}).
   Hence, using \Cref{lemma_M_convex_closed},
   we get that from  any  sequence   in
  $\widetilde{\cW} _\cm$,  we can  extract a subsequence  which converges
  for $\dd$ to  an element in $\widetilde{\cW} _\cm$.  This implies that
  $(\widetilde{\cW} _\cm, \dd)$ is compact.

  \medskip Point~\ref{it:WG=compact} is a direct consequence of
  Point~\ref{it:M-compact} as if $\Space$ is compact, so is $\Proba$. 
\end{proof}

\begin{proof}[Proof of Point~\ref{it:cvx-closed} from \Cref{prop_equiv_tight_compact}]
	\label{page_proof_point_iii}
We prove Point~\ref{it:cvx-closed}.
The fact that ${\cW} _\cm$ and  $\widetilde{\cW} _\cm$ are convex 
 	is  clear as $\cm$ is convex.
To prove that $\UKernelM$ is closed, we consider a 
sequence $(W_n)_{n\in\N}$ in $\UKernelM$ that converges for $\ddcut$ to some $W\in\UKernel$.
As $(W_n)_{n\in\N}$ is a Cauchy sequence for $\ddcut$, by \Cref{conv_graphon_conv_measure}, 
	$(M_{W_n})_{n\in\N}$ is a Cauchy sequence for $\dmeas$ and thus is tight.
Hence, $(W_n)_{n\in\N}$ is uniformly bounded and tight.
Applying \Cref{lemma_M_convex_closed}, there exists a subsequence $(W_{n_k})_{k\in\N}$
	of the sequence $(W_n)_{n\in\N}$ which converges for $\ddcut$
	to some $\cM$-valued kernel $U\in\UKernelM$.
But as a subsequence, $(W_{n_k})_{k\in\N}$ must also converge for $\ddcut$ to $W$.
This implies that $W=U$ is a $\cM$-valued kernel.
\end{proof}

In order to prove Theorem~\ref{theo_equiv_topo},
we first prove a lemma that allows to construct the partitions
needed to use Lemma~\ref{lemme_partitions_convergence}.

\begin{lemme}[Construction of partitions for two distances]
	\label{lemme_partitions_two_distances}
  Let  $d$  and  $d'$ be   two  distances  on  $\Graphon$  (resp.
  $\Kernelp$ or $\Kernel$) which  are invariant,  smooth, weakly regular  and regular
  \wrt  the stepping  operator (see  Definitions~\ref{def:inv-smooth}
  and~\ref{defi:extra-prop}).
  Let $(W_n  )_{n\in\N}$  be a
  sequence in  $\Graphon$ (resp. $\Kernelp$ or $\Kernel$) which is tight
  	(resp. uniformly bounded and   tight).    
  Then,   there   exists   sequences   
  $(\mathcal{P}_{n,k})_{k\in   \N}$,  $n\in\N$,   of   partitions of   $[0,1]$ 
  such   that  hypothesis~\ref{prop_Wnk_1}--\ref{prop_Wnk_4} of
  Lemma~\ref{lemme_partitions_convergence} are satisfied.
\end{lemme}

\begin{proof}
We prove the result on $\Kernel$ 
(the proof on $\Graphon$ and $\Kernelp$ is similar).
To simplify notations, write $d^1=d$ and $d^2=d'$.
We  proceed  by  induction  on $k\in    \N \cup\{-1\}$.     
For every $n\in\N$, set $\cP_{n,-1} = \{ [0,1] \}$ the trivial partition with size $1$.
  Let  $k\in\N$  and assume 
  that we have already constructed partitions $(\mathcal{P}_{n,k-1} )_{n\in\N}$
  	that have the same size $m_{k-1}$.  
  Now we proceed to  construct partitions $(\mathcal{P}_{n,k})_{n\in\N}$
  that satisfy hypothesis~\ref{prop_Wnk_1}-\ref{prop_Wnk_4}.

Set $C=\sup_{n\in \N} \TM{W_n}$, which is finite as the sequence 
$(W_n)_{n\in \N}$ is uniformly bounded. 
As $d^i$, with $i=1,2$, are  regular       \wrt      the       stepping operator,
there exists a finite constant $C_0>0$ such that for every $W, U\in \Kernel$, 
with $\TM{W}\leq C$ and $\TM{U}\leq C$, and $U$ a stepfunction
adapted to a finite partition  $\mathcal{Q}$:
\begin{equation}
   \label{eq:di<Cdi}
  d^i(W,W_{\mathcal{Q}}) \leq C_0 \, d^i(W,U).
\end{equation}

Set  $\varepsilon=  1/C_0(k+1)$.   Since  $d^i$,  with
$i=1,2$, are  weakly regular and  the sequence $(  W_n )_{n\in\N}$ is  tight and
uniformly  bounded,  there  exists   $r_k\in\N^*$,  such  that  for  every
$n\in \N$,  there exists a
partition  $\cR_{n,k}^i$ of  $[0,1]$  that  refines $\cQ_{n,k} = \cp_{n,k-1} \land \cD_{k}$,
	where $\cD_k$ is the dyadic partition with stepsize $2^{-k}$, 
such that:
\begin{equation}
   \label{eq:Pi-di}
  \vert\cR_{n,k}^i\vert    \leq    r_k   \vert\cQ_{n,k}\vert
  			\leq    2^{k} r_k   \vert\cP_{n,k-1}\vert
  \quad\text{and}\quad
  d^i\Bigl(W_n,   (W_n)_{\cR^i_{n,k}}\Bigr)    \leq    \eps.
\end{equation}
(Indeed, a close look at the proof shows that
	$\cP_{n,k-1}$ refines $\cD_{k-1}$ by construction,
	thus $\cQ_{n,k}$ cuts each set of $\cP_{n,k-1}$ in at most $2$ sets,
	and we get $\vert\cQ_{n,k}\vert \leq 2  \vert\cP_{n,k-1}\vert$.)
Now, let $\mathcal{P}_{n,k}$ be the common refinement of
$\cR_{n,k}^1$  and
$\cR_{n,k}^2$; it is  a refinement of $\mathcal{P}_{n,k-1}$,
has diameter at most $2^{-k}$ and size:
\[
  \vert \mathcal{P}_{n,k} \vert \leq 2^{2k} r_k^2 \vert \cp_{n,k-1} \vert^2
  = 2^{2k} r_k^2 m_{k-1}^2.
\]
If necessary, by completing $\mathcal{P}_{n,k}$ with null sets, we
may assume that $\vert \cP_{n,k} \vert 
	= m_k$, where $m_k = 2^{2k} r_k^2 m_{k-1}^2$.
As $ (W_n)_{\cR^i_{n,k}}$ is a stepfunction 
adapted to the partition $\cp_{n,k}$, we deduce from~\eqref{eq:di<Cdi}
and~\eqref{eq:Pi-di}
that for $i=1,2$ and $n\in \N$:
\[
  d^i(W_n,(W_n)_{\mathcal{P}_{n,k}}) \leq C_0 \, 
  d^i\Bigl(W_n,(W_n)_{\cR_{n,k}^i}\Bigr) 
  \leq  C_0 \, \eps = \frac{1}{k+1}\cdot
\]
Hence, for every $n\in\N$, the partition $\mathcal{P}_{n,k}$
satisfies the   hypothesis~\ref{prop_Wnk_1}-\ref{prop_Wnk_4}         of
Lemma~\ref{lemme_partitions_convergence}. 
Thus,  the induction is complete.
\end{proof}

\begin{proof}[Proof of Theorem~\ref{theo_equiv_topo}]
Let $\dmeas$ and $\dmeasPrime$ be as in  Theorem~\ref{theo_equiv_topo}. 

Let  $(W_n )_{n\in\N}$ be  a  sequence of  probability-graphons  that
converges     to    some     $W\in\UGraphon$    for     $\ddcut$.     
By Lemma~\ref{conv_graphon_conv_measure},  
the sequence of probability measure $(M_{W_n})_{n\in\N}$ 
converges to $M_W$ for the distance $\dmeas$.
As $\dmeas$ induces the weak topology on $\SubProba$,
we  have  that the  sequence $(M_{W_n})_{n\in\N}$ is tight,
and thus the sequence $(W_n )_{n\in\N}$ is also tight
	(remind Definition~\ref{defi:tight}).
The sequence $(W_n )_{n\in\N}$ is  also uniformly  bounded  as a  sequence in  $\UGraphon$.
Applying  Lemma~\ref{lemme_partitions_two_distances} with  the distances $d=\dcut$ and
$\dPrime=\dcutPrime$,  which are invariant, smooth, weakly regular and regular \wrt the stepping operator,
we              get      sequences of partitions $(\cP_{n,k})_{k\in\N}$, $n\in\N$,        that
satisfy hypothesis~\ref{prop_Wnk_1}-\ref{prop_Wnk_4}                         of
Lemma~\ref{lemme_partitions_convergence}. 
We  then deduce
from the last part  of Lemma~\ref{lemme_partitions_convergence} that any
subsequence of  $(W_n)_{n\in\N}$ has  a further  subsequence  
which   converges    to the same limit for both $\ddcut$ and $\ddcutPrime$,
	this limit must then be $W$.
This implies  that the sequence $(W_n )_{n\in\N}$ converges  to $W$ for
$\ddcutPrime$.

The role of $\dmeas$ and $\dmeasPrime$ being symmetric,
we conclude that the distances $\ddcut$ and $\ddcutPrime$
induce the same topology on $\UGraphon$.
\end{proof}

\section*{Acknowledgement}

We thank Pierre-André Zitt for some helpful discussion 
in particular on  \Cref{lemma_sup_SxT}.

\cuthere \medskip
\begin{center}
 \textbf{Index of notation}
\end{center}
\vspace{1em}

\renewcommand\labelitemi{-}
\setlength{\columnseprule}{1pt}
\begin{multicols}{2}
\noindent
\hrulefill

\begin{center}
   \textbf{Measures}
\end{center}
\begin{itemize}[leftmargin=*,itemsep=0.5em]
\item $(\Space, \Topo)$ a topological Polish space
\item $\Borel$ the Borel $\sigma$-field induced by $\Topo$
\item $\CbFunct$ the set of continuous bounded real-valued functions on $\Space$
\item \emph{measure} = positive measure
\item $\SignedMeas$ the set of signed measures on $\Space$
\item $\Meas$ the set of  measures on $\Space$
\item $\Proba$ the set of probability measures on $\Space$
\item $\SubProba$ the set of sub-probability measures on $\Space$,
	\ie measures with total mass at most $1$
\item $\mu^+$, $\mu^-$ the positive and negative parts of $\mu$
	from its Hahn-Jordan decomposition
\item $\vert\mu\vert = \mu_+ + \mu_-$ the total variation measure of $\mu$
\item $\TM{\mu} = \vert\mu\vert(\Space)$ the total mass of $\mu$
\item $\dmeas$ a distance on either $\SubProba$, $\Meas$ or $\SignedMeas$
\item $\NmeasSymbol$ a norm on $\SignedMeas$
\item $\dmeasP$ the Prohorov distance
\item $\NmeasKRSymbol$ the Kantorovitch-Rubinstein norm
\item $\NmeasFMSymbol$ the Fortet-Mourier norm
\item $\NmeasFSymbol$ the norm based on a convergence determining sequence $\F$
\end{itemize}

  \hrulefill
  \begin{center}
   \textbf{Relabelings and partitions}
\end{center}

\begin{itemize}[leftmargin=*,itemsep=0.5em]
\item $\InvRelabel$ the set of bijective 
	measure-preserving maps from $([0,1],\lambda)$ to itself
\item $\Relabel$ the set of measure-preserving maps from $([0,1],\lambda)$ to itself
\item $\vert \cP\vert$ the number of sets in the finite partition $\cP$
\end{itemize}

  \hrulefill
  \begin{center}
   \textbf{Kernels and graphons spaces}
\end{center}

  \begin{itemize}[leftmargin=*,itemsep=0.5em]
\item $\Graphon$ the set of probability-graphons 
\item $\Kernelp$ the set of measure-valued kernels
\item $\Kernel$ the set of signed measure-valued kernels
\item $\KernelM$ the set of $\cm$-valued kernels
	with $\cM \subset \SignedMeas$
\item $\UGraphon$ the set of unlabeled probability-graphons 
\item $\UKernelp$ the set of unlabeled measure-valued kernels
\item $\UKernel$ the set of  unlabeled signed measure-valued kernels
\item $\UKernelM$ the set of unlabeled $\cm$-valued kernels
\end{itemize}

 \hrulefill
  \begin{center}
   \textbf{Kernels and graphons}
\end{center}

  \begin{itemize}[leftmargin=*,itemsep=0.5em]
 \item $W^+$ and $W^-$ the positive and negative part of $W\in\Kernel$
\item $\vert W\vert = W^+ + W^-$
\item $W(A;\cdot) = \int_{A} W(x,y;\cdot)\ \drv x \drv y$ for $A\subset [0,1]^2$
\item $W[f](x,y) = W(x,y;f)$ for  $f\in\CbFunct$
\item $W_\mathcal{P}$ the stepping of $W$ \wrt a partition $\cP$

\item $\TM{W}:=    \sup_{x,y\in [0, 1] }\, \TotalMass{W(x, y; \cdot)}$
\item $M_W(\drv z) = \vert W \vert ([0,1]^2; \drv z)$

\item $W_G$ the probability-graphon associated to a $\Proba$-graph or a weighted graph $G$
\item $\bbH(k,W)$ the $\Proba$-graph with $k$ vertices sampled from $W\in\Graphon$
\item $\bbG(k,W)$ the $\Proba$-graph with $k$ vertices sampled from $W\in\Graphon$
\item $F^g$ a finite graph whose edges are decorated with functions in $\CbFunct$
\item $t(F^g,W) = M_W^F(g)$ the homomorphism density of $F^g$ in $W$
\end{itemize}
  \hrulefill
  \begin{center}
   \textbf{Distances/norms on graphon spaces}
\begin{itemize}[leftmargin=*,itemsep=0.5em]
\item $\dcut$ the cut distance associated to $\dmeas$
\item $\NcutSymbol$ the cut norm associated to $\NmeasSymbol$
\item $\dd$ the unlabeled distance associated to an arbitrary distance $d$
\item $\ddcut$ the unlabeled cut distance associated to $\dcut$ or $\NcutSymbol$
\item $\NcutRSymbol$ the cut norm for real-valued kernels
\item $\NcutRpos{\cdot}$ the positive part of the cut norm for real-valued kernels

      \end{itemize}
    \end{center}
 \hrulefill
  \begin{center}
   \textbf{Definitions}
\begin{itemize}[leftmargin=*,itemsep=0.5em]
\item  \emph{weak  isomorphism} of kernels and graphons in  \Cref{def_weak_isomorphism} on
  page~\pageref{def_weak_isomorphism}
\item  \emph{tightness}  for  sets of kernels or graphons   in  \Cref{defi:tight}  on
  page~\pageref{defi:tight}
\item \emph{invariant} and \emph{smooth}  for a distance $d$ on graphon spaces
	in \Cref{def:inv-smooth} on page~\pageref{def:inv-smooth}

      \item  \emph{weakly   regular}  and  \emph{regular   \wrt  the
          stepping  operator} for  a distance  $d$ on  graphon spaces  in
          \Cref{defi:extra-prop} on page~\pageref{defi:extra-prop}
      \end{itemize}
    \end{center}

  \hrulefill
\end{multicols}
\cuthere

\bibliographystyle{alpha}
\bibliography{graphon_de_lois.bib}

\end{document}